\documentclass[a4paper]{cas-sc}

\usepackage[authoryear]{natbib}

\usepackage[absolute,overlay]{textpos}
\usepackage{enumerate}
\usepackage{enumitem} 
\usepackage{subcaption}
\usepackage{amsthm}
\usepackage{amsmath}
\usepackage{amssymb}
\usepackage{xfrac}
\usepackage{multirow}
\usepackage{tabularx,booktabs}
\newcolumntype{C}{>{\centering\arraybackslash}X} 
\newcolumntype{R}{>{$}r<{:$}} 

\newtheorem{theorem}{Theorem}
\newtheorem*{theorem*}{Theorem}

\theoremstyle{remark}
\newtheorem{remark}{Remark}
\newtheorem{insight}{Insight}
\newtheorem*{claim*}{Claim}

\theoremstyle{plain}
\newtheorem{example}{Example}

\newtheorem{constraint}{Prior Knowledge}

\newcommand*\MyExp{{\mathbb E}}

\newproof{pf}{Proof}
\newproof{pot}{Proof of Theorem \ref{thm2}}

\usepackage{amsmath,amssymb,amsfonts}
\usepackage{algorithmic}
\usepackage{graphicx}
\usepackage{textcomp}
\usepackage{xcolor}
\def\BibTeX{{\rm B\kern-.05em{\sc i\kern-.025em b}\kern-.08em
		T\kern-.1667em\lower.7ex\hbox{E}\kern-.125emX}}

\usepackage[utf8]{inputenc}
\usepackage{times}
\usepackage{comment}
\usepackage{soul}
\usepackage{url}
\usepackage[utf8]{inputenc}
\usepackage{caption}
\usepackage{booktabs}
\usepackage{algorithm}
\usepackage{amsthm}
\newcount\Comments  
\usepackage{color}
\definecolor{darkgreen}{rgb}{0,0.5,0}
\definecolor{purple}{rgb}{1,0,1}
\newcommand{\kibitz}[2]{\ifnum\Comments=1\textcolor{#1}{#2}\fi}

\newcommand{\kizito}[1]{\kibitz{red}      {[Kizito: #1]}}

\newcommand{\xingyu}[1]  {\kibitz{darkgreen}   {[Xingyu: #1]}}

\usepackage{hyperref}

\usepackage{acro}
\acsetup{first-style=short}
\DeclareAcronym{cbi}{
	short = CBI ,
	long  = Conservative Bayesian Inference\kizito{{\bf BI} too?},
	tag = abbrev
}
\DeclareAcronym{av}{
	short = AV ,
	long  = Autonomous Vehicle,
	tag = abbrev
}
\DeclareAcronym{pfe}{
	short = \textit{pfe} ,
	long  = probability of failure per execution,
	tag = abbrev
}
\DeclareAcronym{pk}{
	short = PK ,
	long  = Prior Knowledge,
	tag = abbrev
}
\DeclareAcronym{csc}{
	short = SCS ,
	long  = Safety-Critical System,
	tag = abbrev
}
\DeclareAcronym{NHPP}{
	  short = NHPP ,
	  long  = non-homogeneous Poisson process,
	  tag = abbrev
	}
\DeclareAcronym{SRGM}{
	  short = SRGM ,
	  long  = software reliability growth model,
	  tag = abbrev
	}
\DeclareAcronym{pfd_var}{
	short = \textit{pfd},
	long  = probability of failure per demand,
	tag = abbrev
}
\DeclareAcronym{pfm_var}{
	short = \textit{pfm},
	long  = probability of fatality-event per mile,
	tag = abbrev
}
\DeclareAcronym{Bernoulli_trail}{
	short = \ensuremath{T_i},
	long  =  Random variable that indicates when the $i$th execution fails,
	sort  = 000 ,
	tag = nomencl
}
\DeclareAcronym{pfd}{
	short = \ensuremath{X},
	long  = The unknown probability of an execution failing,
	sort  = 001,
	tag = nomencl
}
\DeclareAcronym{dependence_factor}{
	short = \ensuremath{\Lambda},
	long  = {The unknown probability of the next execution failing when the last execution was a failure.},
	sort  = 010,
	tag = nomencl
}
\DeclareAcronym{engineering_goal}{
	short = $\epsilon$,
	long  = A bound on $X$ representing an ``engineering goal'',
	sort  = 020,
	tag = nomencl
}
\DeclareAcronym{confidence_in_engineering_goal}{
	short = $\theta$,
	long  = Prior confidence in the engineering goal being satisfied,
	sort  = 021,
	tag = nomencl
}
\DeclareAcronym{required_bound}{
	short = $b$,
	long  =  A required upper bound on $X$,
	sort  = 022,
	tag = nomencl
}
\DeclareAcronym{post_confidence_in_b}{
	short = $c$,
	long  = Posterior confidence in $X$ satisfying the bound $b$,
	sort  = 023,
	tag = nomencl
}
\DeclareAcronym{conf_positive_col}{
	short = $\phi_1$,
	long  = Prior confidence in negatively dependent executions,
	sort  = 024,
	tag = nomencl
}
\DeclareAcronym{conf_negative_col}{
	short = $\phi_2$,
	long  = Prior confidence in positively dependent executions,
	sort  = 025,
	tag = nomencl
}
\DeclareAcronym{number_of_tests}{
	short = $n$,
	long  = The total number of executions,
	sort  = 026,
	tag = nomencl
}
\DeclareAcronym{failure_s}{
	short = $s$,
	long  = The number of executions that are failures,
	sort  = 027,
	tag = nomencl
}
\DeclareAcronym{failure_r}{
	short = $r$,
	long  = The number of failures that are preceded by a failure,
	sort  = 028,
	tag = nomencl
}
\DeclareAcronym{likelihood_func}{
	short = $L$,
	long  = The likelihood function,
	sort  = 036,
	tag = nomencl
}
\DeclareAcronym{dom_of_lklhd_fn}{
	short = $\mathcal{R}$,
	long  = The domain of the Klotz failure model,
	sort  = 037,
	tag = nomencl
}

\begin{document}
\let\WriteBookmarks\relax
\def\floatpagepagefraction{1}
\def\textpagefraction{.001}

\begin{textblock*}{20cm}(1cm,1cm)
\textcolor{red}{Accepted by the journal of Quality and Reliability Engineering International}
\end{textblock*}

\shorttitle{Demonstrating Software Reliability using Possibly Correlated Tests}

\shortauthors{Salako and Zhao}

\title[mode = title]{ 
Demonstrating Software Reliability using Possibly Correlated Tests: Insights from a Conservative Bayesian Approach 
}                      



%



\author[1]{Kizito Salako}[orcid=0000-0003-0394-7833]
\ead{k.o.salako@city.ac.uk}



\affiliation[1]{organization={Centre for Software Reliability, City, University of London},
    addressline={Northampton Square}, 
    city={London},
    citysep={}, 
    postcode={EC1V 0HB}, 
    country={United Kingdom}}

\author[2,3]{Xingyu Zhao}[orcid=0000-0002-3474-349X]

 \ead{xingyu.zhao@liverpool.ac.uk}


\affiliation[2]{organization={Department of Computer Science, University of Liverpool},
     addressline={Ashton Stree}, 
    city={Liverpool},
    citysep={}, 
    postcode={L69 3BX}, 
    country={United Kingdom}}
\affiliation[3]{organization={Warwick Manufacturing Group, University of Warwick},
     addressline={Lord Bhattacharyya Way}, 
    city={Coventry},
    citysep={}, 
    postcode={CV4 7AL}, 
    country={United Kingdom}}






\begin{abstract}
	This paper presents Bayesian techniques for conservative claims about software reliability, particularly when evidence suggests the software's executions are not statistically independent. We formalise informal notions of ``\emph{doubting}'' that the executions are independent, and incorporate such doubts into reliability assessments. We develop techniques that reveal the extent to which independence assumptions can undermine conservatism in assessments, and identify conditions under which this impact is not significant. These techniques -- novel extensions of \emph{conservative Bayesian inference} (CBI) approaches -- give conservative confidence bounds on the software's failure probability per execution. With illustrations in two application areas -- nuclear power-plant safety and autonomous vehicle (AV) safety --  our analyses reveals: {\bf 1)} the confidence an assessor should possess before subjecting a system to operational testing. Otherwise, such testing is futile -- favourable operational testing evidence will eventually decrease one's confidence in the system being sufficiently reliable; {\bf 2)} the independence assumption supports conservative claims sometimes; {\bf 3)} in some scenarios, observing a system operate without failure gives less confidence in the system than if some failures \emph{had} been observed; {\bf 4)} building confidence in a system is very sensitive to failures -- each additional failure means significantly more operational testing is required, in order to support a reliability claim.
\end{abstract}










\maketitle

\printacronyms[include=abbrev,name=Abbreviations]
\printacronyms[include=nomencl,name=Notation]
\acuseall 

\section{Introduction}
\label{sec_intro}
\begin{figure*}[t!]
	\centering
	\includegraphics[width=0.9\linewidth]{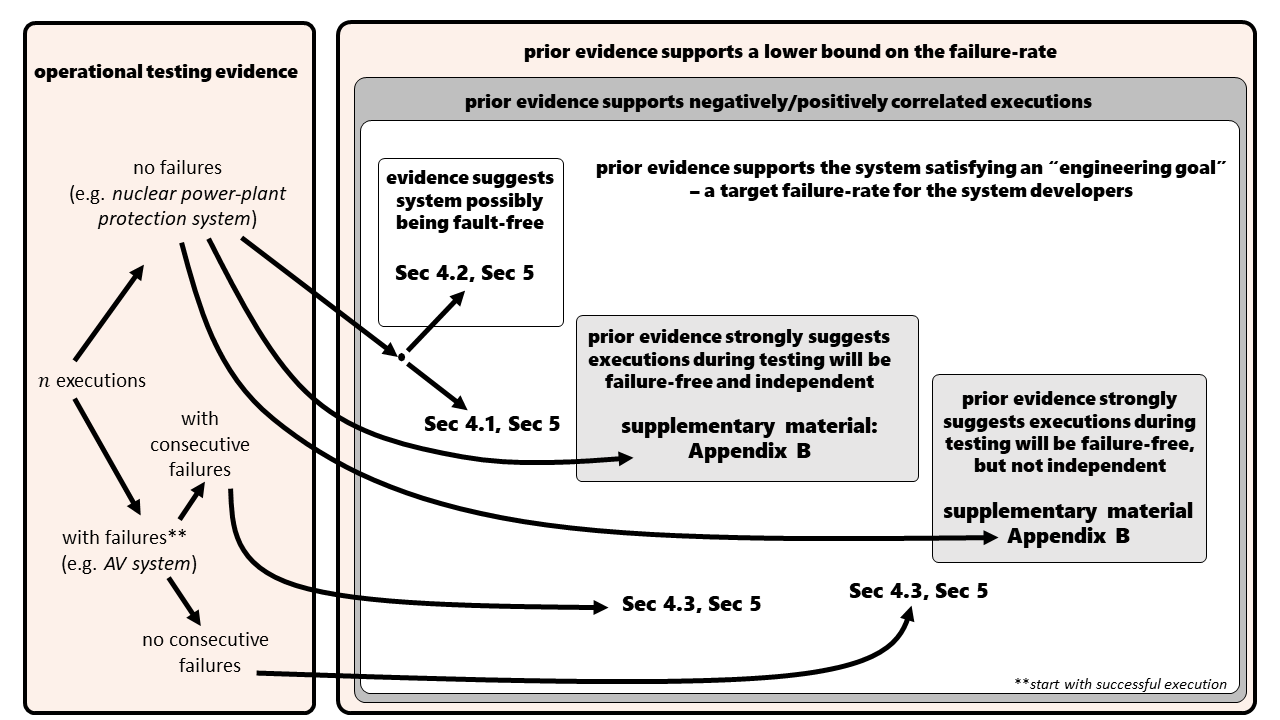}
	\caption{{\footnotesize various assessment scenarios and the sections in this paper that treat them. Each path through the diagram, starting at the ``n executions'' node on the left, indicates a system's behaviour (during operational testing) and the implications of reliability evidence considered by an assessor (before operational testing).   \normalsize}}
	\label{fig_theBigPicture_intro}
\end{figure*}
It is prudent to be conservative when assessing a software-based \emph{safety-critical system} (SCS), since software failure could significantly harm stakeholders in the system.  Rigorous statistical arguments can give support for conservative claims about whether the software is sufficiently reliable, where such claims are based on evidence of achieved levels of reliability. In particular, Bayesian methods provide a natural formalism/calculus for combining various forms of reliability evidence, resulting in probabilistic measures that (given the evidence) articulate one's uncertainty about the reliability of the SCS (see \cite{atwood_handbook_2003}); in particular, the reliability of its software. Examples of evidence that can be utilised in Bayesian approaches for reliability assessment include the observed failure behaviour of (similar) software during past operation (e.g., see \cite{thomas_e_wierman_reliability_2001}, \cite{bunea_two_stage_2005} and \cite{porn_two_stage_1996}).

When using statistical arguments in assessments, a central question is whether the system's software has statistically \emph{independent and identically distributed} (i.i.d.) ``executions''. By ``an execution'' we mean \emph{a set of actions performed by the software that can be regarded as a unit of software operation}\footnote{This is also known as a ``\emph{run}'' (e.g., see \cite{Strigini_Littlewood_EuropeanSpaceAgency_1997}).}. For example, actions performed in response to \emph{each} demand/input the software receives from its environment, or actions in response to \emph{a sequence} of inputs (received over a unit amount of time or distance). In this paper, when a software execution occurs, either all of the actions performed\footnote{Not performing any action could be the required action.} are the required actions or at least one of these actions is incorrect -- i.e., a software execution is either correct or a failure. Our focus is the assessment of the system's software, so we consider only software failures (so, no hardware failures) as defining system failure.

Executions that are i.i.d. make sense for some systems, such as an on-demand system where demands rarely occur, and the system's state and operational environment are reset inbetween demand occurrences. But there are often reasons to doubt the i.i.d. assumption. An \emph{autonomous vehicle} (AV) could experience sudden changes in driving conditions that make it very likely for the AV to make a series of consecutive mistakes; or an airplane (and its flight control systems) can be put under increasing operational stresses when they encounter aggressive weather mid-flight (with turbulent ``air pockets''); and ``failure clustering'' has been observed in various (control) systems. In many situations, at least \emph{some} doubt about independent executions is warranted.\kizito{Might be nice to have examples of real-life incidents where, say, AVs make a series of consecutive errors} 

Even when executions \emph{are} assumed i.i.d., SCS software will typically be required to exhibit \emph{no failures} over a large number of executions. An assessor might (must?) consider whether these successful executions \emph{are} positively correlated after all, and might account for this possibility in assessments. Because, at face-value, an assumption of i.i.d. executions can seem quite strong -- it significantly limits an assessor's hypothesis about which probabilistic laws could characterise a software's failure process. Consequently, one might suspect that assuming independence results in optimistic reliability claims -- it's useful to ask whether this is \emph{actually} the case. Are assessments significantly undermined by assuming independence?

The answer depends on: {\bf i)} how reliable the software actually is, {\bf ii)} the sequence of the software's successes/failures during operational testing, and {\bf iii)} the nature of any dependence between executions. Prior to testing, the assessor is uncertain about (i)--(iii), and is reliant on reliability evidence to shape their beliefs about how reliable the system is. Operational testing provides more evidence that refines these beliefs further.

This process -- of an assessor's uncertainties being initially shaped by evidence obtained prior to operational testing, and then further shaped by the system's performance during testing -- is formalised in this paper in \emph{conservative Bayesian inference} (CBI) terms, for an assessor making claims about the system's probability of failure per ``\emph{execution}'' (\emph{pfe}). An example \emph{pfe} is the \emph{probability of failure per demand} (\emph{pfd}) for an on-demand system, and another example is the \emph{probability of a fatality-event per mile} (\emph{pfm}) for an AV; we consider both examples later in the paper. 
Our primary interest is in assessing software reliability -- specifically, solving constrained optimisation problems, to obtain the least confidence an assessor can justifiably have in a system's \emph{pfe} being ``small enough''.

CBI makes explicit how (i)--(iii) above affect an assessor's uncertainty under various assessment scenarios. Fig.~\ref{fig_theBigPicture_intro} summarises these scenarios and indicates sections in this paper where the scenarios are treated. Prior to testing, an assessor may express some confidence in, say, {\bf i)} the system being sufficiently reliable (e.g. confidence in the unknown \emph{pfe} being smaller than some target value set for the system developers); {\bf ii)} the system being fault-free; {\bf iii)} future executions being negatively or positively dependent, and being failure-free. An assessor uses execution outcomes during testing -- such as \emph{no} failed executions occurring, or \emph{some} failures separated by runs of successes occurring -- to update their confidence.

\xingyu{Thank you Kizito, I like the story of the introduction. One minor---we don't have a single citation in the introduction, which seems a bit unusual... to make the story more convincing to others, maybe we would like to cite some papers? But, indeed, I see we have a quite comprehensive related work section, any citations would be repeating refs already in section 2... So i am both ok.. }\kizito{yes, we should probably include some ``soft'' citations in the introduction. I have left comments in the respective areas. Please add more references that you think might be appropriate.}

\paragraph{Summary of the paper's contributions}
\begin{enumerate}[wide,label={\arabic*)}]
	\item extending CBI techniques that allow an assessor to quantify the potential negative impact of invalid statistical modelling assumptions on reliability claims (e.g., when software executions are assumed i.i.d. when, in actuality, they are not). See section~\ref{sec_conf_bound_in_reliability};
	\item  showing how the outcomes of software executions -- whether failures occurred and whether these were clustered or isolated -- significantly affects how confident an assessor can (justifiably) be of a system being sufficiently reliable (sections \ref{sec_conf_bound_in_reliability}, \ref{sec_senstivity_analysis});
	

	
	
		
	\item  several closed-form solutions for conservative posterior confidence in an upper bound on \emph{pfe} (see Fig.~\ref{fig_theBigPicture_intro}, sections \ref{sec_conf_bound_in_reliability}, \ref{sec_senstivity_analysis} and the supplementary material);
	
		\item illustrating these findings in two scenarios: nuclear power-plant and autonomous vehicle (AV) safety (section \ref{sec_conf_bound_in_reliability}). We give advice and caution for assessors/practitioners, concerning how confident they should be before embarking on operational testing (sections \ref{sec_conf_bound_in_reliability}, \ref{sec_senstivity_analysis}).
	
	
\end{enumerate}

\xingyu{Do we want to also mention that all our implementation of models and experimental results (i.e. examples and SA) are publicly available at some place... The AI/ML guys do this a lot, maybe this is not our thing..}\kizito{this would be great, but I'm a bit hesitant -- unless we can be sure that source code is "good enough". I'm thinking about the numerical instability we have recently noticed with our $r=0$ and $\phi_1$ calculations.\xingyu{I see, agreed.}}

The rest of the paper is organised as follows. Related work is detailed in section \ref{sec_related_work}, while section \ref{sec_preliminary} reviews CBI and the Klotz model for correlated executions. Formalisations of doubting i.i.d. executions are given in section \ref{sec_conf_bound_in_reliability}, and used to derive conservative confidence bounds on system \emph{pfe} under various scenarios. The sensitivity of these bounds to changes in model parameters is studied in section \ref{sec_senstivity_analysis}. 
Section \ref{sec_discussion} discusses results, and section \ref{sec_conclusion} concludes the paper.

\section{Related Work}
\label{sec_related_work}
\xingyu{I feel it would be better to place the related work section near the end than at the beginning... Because the readers might need to first know what is the binary Markov model we use, what is CBI and what is 2d-CBI etc. to follow this section better...}\kizito{I see the point. But I am concerned that many readers like to be motivated before they read the rest of the paper -- in particular, they like to understand why they should care about our work (the "intro"), and how it extends previous work that has already been done. So, I moved the related work back to near the beginning. I have tried writing it in a way that does not require readers to know details of our work.\xingyu{I see, sure, agreed}}
The current paper directly continues our development of statistical techniques for conservative reliability assessment first reported in \cite{SalakoZhao_TSE_2023}. Consequently, the related works detailed in that paper continue to be relevant here; we highlight these works in this section. 
\subsection{Why is Modelling Correlated Executions Necessary?}

An early model for sequences of statistically independent executions, used in works on random testing, is due to \cite{ThayerEtAl_1975} (see \cite{DuranEtal_1984_EvaluationofRandomTesting} for an application of the model). However, reasons to expect correlated failed executions in various systems became well-known. For example, a system can exhibit ``\emph{failure clustering}'' due to the system receiving sequences of inputs that cause the system to fail, where such inputs cluster into subsets of the system's failure region\footnote{A software component's failure region is a geometrical, or mathematically topological, characterisation of those inputs that trigger the software to fail.} -- see \cite{AmmanKnight_FailureClustering_1988} and \cite{Bishop_failureclustering_1993}. The system's operational environment generates input sequences as trajectories (within the set of all inputs) that eventually enter into, and linger in, these failure regions. This phenomenon motivated developing assessment approaches that account for positive failure correlation between executions -- see \cite{Csenki_RecoveryBlocks_1993,TomekTrivedi_recoveryblocks_1993,Huang_Sun_2021_AdaptiveRandomTestingSurvey}. 

\cite{strigini_testing_1996} gives other reasons for correlated executions; e.g. if the software's internal state is corrupted upon an initial failed execution, making subsequent executions more likely to fail. Or, if the system's operational environment becomes increasingly more stressful (i.e. there's an increasing probability of trajectories in the input space entering the failure region).

\subsection{Statistical Models of Correlated Executions}
A number of models with Markov dependence have been proposed for correlated executions. The binary Markov chain of \cite{chen_binary_1996}; the Markov renewal process of \cite{goseva_popstojanova_failure_2000} (that builds upon earlier work in \cite{Csenki_RecoveryBlocks_1993}, \cite{TomekTrivedi_recoveryblocks_1993}); and the \cite{bondavalli_dependability_1995} model that captures benign-failures, and the cumulative impact of such failures when assessing iterative software. \cite{bondavalli_modelling_1997} improve on this model, demonstrating the model's use with fitted steady-state and transition probabilities.

None of the aforementioned models are demonstrated using inference methods that explicitly account for one's uncertainty about whether the executions are i.i.d. or not. Nor do these models provide demonstrably conservative statistical support for reliability claims about software -- where such support is justified by various forms of reliability evidence (in addition to the outcomes of \emph{possibly} correlated executions during operational testing).

\subsection{Conservative Bayesian Methods for Assessments}
\label{subsec_CBAmethods}
A number of studies have applied Bayesian methods to support software reliability assessment, e.g.  \cite{WMiller_1992_BayesianReliabilityAssesssment,Singh_2001_BayesianReliabilityAssessement,LittlewoodPopov_2002_AssessingTheReliabilityDiversesoftware,Popov_2013_BayesianReliabilityAssessement}.
The utility of these methods is in the inference process. An assessor's beliefs about the reliability of a system are initially formed by evaluating relevant evidence. Then these beliefs are updated, upon seeing how the system performs during operation.

The usual challenge with Bayesian methods is the need to characterise one's initial (i.e. ``\emph{prior}'') beliefs as a prior probability distribution -- a distribution that captures all, and only all, of one's prior beliefs. Care must be taken when eliciting a prior distribution; an unrepresentative prior could lead to overly pessimistic, or dangerously optimistic, assessments. 

For reliability assessments, there is the added challenge that prior distributions often represent beliefs about continuous random variables, such as an on-demand system's unknown \emph{pfd}. Requiring that an assessor specify beliefs about the \emph{infinitely many} ranges of possible \emph{pfd} values is often impractical.

CBI methods have been developed to address these challenges. CBI is related to \emph{robust Bayesian analysis} which studies the sensitivity of the results of Bayesian inference to changes in the inference inputs -- see \cite{berger1994_robustBayesianOverview,berger1987_SensitivityToPrior,
	Lavine_1991_SensitivityInBayesianStats,
	Berger_1994_RobustnessInBidinesionalModels}. Inputs such as: the prior distribution; the statistical model that determines the likelihood function; and the posterior measure of interest. An assessor may not be able to use available evidence to fully specify a prior distribution, but the evidence may allow a much more limited \emph{partial} specification of a prior -- e.g. the assessor expects the prior, whatever it may be, to satisfy certain quantiles or moments. By considering \emph{all} of those prior distributions that satisfy these specifications, CBI determines the most conservative inference result (from using these priors) to give support for a claim that the system is sufficiently reliable. This is one way in which conservatism in assessments is realised via Bayesian inference.

\cite{bishop_toward_2011} introduced the CBI idea. A number of studies soon followed, applying CBI in various contexts. For example, {\bf i)} in \cite{strigini_software_2013}, Povyakalo \emph{et al.} use CBI to obtain the smallest probability of the system's next $m$ executions being successful, given prior evidence that the system is very reliable and its last $n$ executions were successful; {\bf ii)} in \cite{zhao_conservative_2015}, with evidence to support some confidence in the system possibly being fault-free, and some confidence in the system being very reliable, Zhao \emph{et al.} use CBI with operational testing to conservatively gain confidence in the system possibly being fault-free; {\bf iii)} \cite{salako_loss_size_2020} bounds the reliability of a binary classifier, given evidence that the classifier's past performance was (un)likely; and {\bf iv)} in \cite{zhao_assessing_2019}, Flynn \emph{et al.} apply CBI to the problem of assessing AV safety -- highlighting circumstances under which attempts to demonstrate the required levels of safety via road testing are in vain. More CBI applications to assessing SCSs are found in \cite{zhao_conservative_2015,
	zhao_modeling_2017,
	zhao_conservative_2018,
	zhao_assessing_2019}.




These applications all involve ``univariate'' priors; i.e., distributions of a single unknown, typically the system \emph{pfe}. More recently, CBI applications have involved ``bivariate'' priors. \cite{littlewood_reliability_2020} consider the assessment of a system in a ``new'' situation -- either the system replaces an older system in a given operational environment, or the system has been deployed in a new environment after operating in a previous environment for some time. They demonstrate how CBI supports dependability claims, when evidence suggests the system's failure propensity in the new situation is ``no worse'' than the propensity in the ``old'' situation. \cite{zhao_assessing_2020} study ``improvement arguments'' of this kind in the context of assessing AV safety -- but with different dependability  measures of interest and a more general failure model for the system. While \cite{salako_conservative_2021} consider more general ``improvement arguments'' for an even wider range of assessment scenarios.

In \cite{SalakoZhao_TSE_2023}, we introduce a CBI technique for incorporating doubts about the i.i.d. assumption into conservative reliability claims. The statistical model used was the first CBI model to capture correlated executions, and is based on the \cite{klotz_statistical_1973} model -- a binary Markov chain that predates and agrees with \cite{chen_binary_1996} and \cite{goseva_popstojanova_failure_2000}. We illustrated how assessments where i.i.d. executions are assumed can be very optimistic. That paper was concerned with assessing on-demand systems where (despite the software containing faults) no software failures are observed during extensive operational testing. The current paper significantly extends this CBI approach, to apply to a wider range of assessment scenarios -- e.g., scenarios where some failures are observed during extensive testing, and where the assessor can justify only very weak beliefs about the unknown \emph{pfe}.

\subsection{Assessing Continuously Operating Software}
The current paper focuses on assessment scenarios where one observes the software's success/failure behaviour on each of a number of ``unit'' software operations -- i.e., on each execution. Example scenarios include assessing on-demand systems (see \cite{iectr63161_2022,2014_Rausand_relOfSafetyCriticalSys}). Hence, we employ ``discrete-time'' statistical models (i.e., \emph{Bernoulli processes}) and our reliability measure of interest is \emph{pfe}. 

Contrastingly, for assessments where the reliability measure of interest is the software's failure-rate in continuous time, 
``continuous-time'' statistical models are more appropriate (e.g., \emph{non-homogeneous Poisson processes} (NHPPs)). 

Many \emph{Software reliability growth models} (SRGMs) -- useful in predicting future reliability for continuously operating software (see \cite{MichaelLyu_1996,xie_software_1991,musa1987software}) -- have been developed over the years;  \cite{1994_Sinpurwalla_RelGrowthModels}, \cite{1986_Miller_ExpOrderStatistics} and \cite{1991_BergmanXie_BayesianSRGMs} give good overviews of early SRGMs.

\section{Preliminaries: a CBI Model of Correlated Executions}
\label{sec_preliminary}
\subsection{A Review of CBI}
\label{sec_CBIreview}
Consider the following scenario from \cite{zhao_assessing_2019}. An on-demand system is subjected to operational testing, to determine if its \textit{probability of failing on a randomly occurring demand} (i.e. \emph{pfe}) is acceptably low. Let $X$ be this unknown \emph{pfe} -- i.e. on a random demand, the system fails (with probability $X$) or succeeds (with probability $1-X$). Demands occur randomly according to an \emph{operational profile}, see \cite{musa_operational_1993}. 

Before operational testing, an assessor might have sufficient evidence to fully specify a \textit{prior distribution} representing their beliefs about which values of $X$ are likely to be the true value, and which values aren't. During operational testing, the system correctly responds to all $n$ demands that occur -- these successes are assumed to occur in an ``\emph{independent and identically distributed}'' (i.i.d.) manner
. If the value of $X$ is $x$ then the probability of observing these successes is $L(x;n)=(1-x)^n$. Let $b$ be a required upper bound on $X$. The assessor's confidence (after seeing the successes) in $X$ being no larger than $b$ is:
\begin{align}	\label{eqn_post_cf_bound_with_complete_prior}
	P(X \leqslant b \mid n\mbox{\it\ successes})&=\frac{P(X\leqslant b,\,n\mbox{\it\ successes})}{P(n\mbox{\it\ successes})} 
=\frac{\MyExp[L(X;n){\bf 1}_{X\leqslant b}]}{\MyExp[L(X;n)]}
\end{align}
where ${\bf 1}_{\tt S}$ is an indicator function---it equals 1 when predicate ${\tt S}$ is true, and 0 otherwise. 

Often, there isn't enough evidence to fully justify a specific prior distribution for \eqref{eqn_post_cf_bound_with_complete_prior}, but there may be enough to justify weaker constraints on the prior (such as a few quantiles). We refer to such constraints as \emph{prior knowledge} (PK). A basic form of PK is:  
\begin{constraint}[certainty in a lower bound]
	\label{cons_pl_lowerbound}
	$100\%$ confidence in the system's pfe not being lower than $p_l$; i.e., $P(X\geqslant p_l) = 1$.
\end{constraint}
\noindent $X$ is a probability, so $p_l$ should be non-negative. $p_l$ can be 0 (see \cite{littlewood_reasoning_2012} for possibly ``fault-free'' software) or a very small number (e.g. the best reliability feasible for the system given current levels of technology).
\begin{constraint}[confidence in satisfying an engineering goal]
	\label{cons_engineering_goal}
	$\theta \times 100\%$ confidence in the system's pfe being better than, or equal to, an upper bound $\epsilon$; i.e., $P(X \leqslant \epsilon)=\theta$.
\end{constraint}
\noindent $\epsilon$ is an ``engineering goal'': a target \emph{pfe} value that system developers try to achieve. $\epsilon$ is typically chosen to be much smaller than the required bound $b$, so $\epsilon \leqslant b$. While $\theta$ is how confident the assessor is, before operational testing, that the engineering goal \emph{has} been achieved. $\theta$ has to be large enough to support conducting operational testing; reducing the chance that unreliable systems use up the operational testing budget.
\begin{figure}[htbp!]
	\centering
	\includegraphics[width=0.25\linewidth]{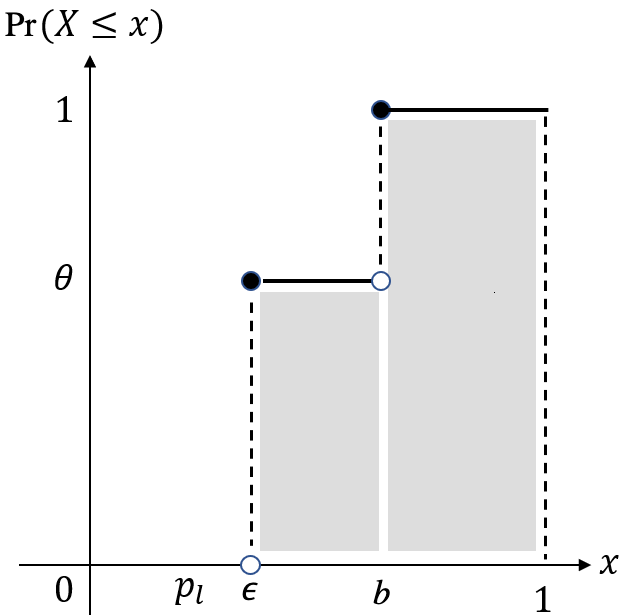}
	\caption{A conservative prior cumulative distribution function. Illustration from \cite{salako_conservative_2021}.}
	\label{fig_unicbi_wc_prior}
\end{figure}


The following theorem shows one can conservatively gain confidence in a bound $b$ on $X$ (see \cite{zhao_assessing_2019}).
\begin{theorem}[univariate CBI]
	Let ${\mathcal D}_u$ be the set of all probability distributions over the interval $[0,1]$. Using \eqref{eqn_post_cf_bound_with_complete_prior}, the optimisation problem\begin{align*}
		&\inf\limits_{{\mathcal D}_u}P(X\leqslant b \mid n\mbox{ executions without failure}  ) \\
		s.t.\,\,\,\,& PK\ref{cons_pl_lowerbound},\,\,\,PK\ref{cons_engineering_goal}
	\end{align*} 
	is solved by Fig.~\ref{fig_unicbi_wc_prior}'s prior distribution: using this prior, $P(X< b \mid n\mbox{ executions without failure})$ equals the infimum\footnote{To obtain this solution, one constructs sequences of feasible priors (where each subsequent prior in a sequence gives an increasingly worse value for the objective function) that converge ``\emph{almost everywhere}'' to a prior that gives the infimum. This prior differs from the feasible priors converging to it at $X=b$. As a consequence, it is the value of  $P(X<b\mid\ldots)$ from this prior, rather than $P(X\leqslant b\mid\ldots)$, that is the infimum.}.
	\label{theorem_univariateCBI}
\end{theorem}

\begin{insight}[The basic CBI idea]
	\label{insight_gist_cbi}
	One considers the set of all feasible priors that satisfy an assessor's PKs. For a given posterior measure of interest (e.g. posterior confidence in a bound on $X$), CBI determines a prior that gives the most pessimistic value for this measure -- no feasible prior gives a more pessimistic value, and any prior that does must violate at least one PK. The CBI prior is referred to as a ``worst-case'' prior.  
\end{insight}

\subsection{A Model of Correlated Software Executions}
\label{sec_KlotzModel}
The following stochastic failure process for software exhibiting correlated executions -- used in \cite{SalakoZhao_TSE_2023} -- is based on the \cite{klotz_statistical_1973} model. 
A sequence of Bernoulli random variables $T_1,\dots, T_n$, each take on the values 1 or 0; indicating failure or success, respectively, on each of $n$ software executions. Similar to section \ref{sec_CBIreview}, let $x$ be the unconditional \emph{probability the next execution is a failure} (\emph{pfe}). Let $\lambda$ be the probability that \emph{a failure is followed by another failure}. That is,
\begin{align}
	&P(T_i=1)=1-P(T_i=0)=x,  &i=1,\dots,n \nonumber
	\\
	&P(T_i=1\mid T_{i-1}=1)=\lambda,  &i=2,\dots,n \nonumber
\end{align}
If the process is \emph{1st-order stationary}, we obtain the Markov model in Fig.~\ref{fig_klotz} (see \cite{klotz_statistical_1973}, \cite{SalakoZhao_TSE_2023}).

\begin{figure}[htbp!]
	\centering
	\includegraphics[width=0.3\textwidth]{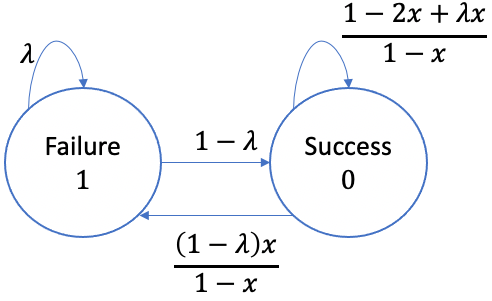}
	\caption{The Klotz model for Bernoulli trials with dependence. Illustration from  \cite{SalakoZhao_TSE_2023}.}
	\label{fig_klotz}
\end{figure}

This stochastic process is well-defined if the transition probabilities lie between zero and one. That is, if
\begin{align}
	\label{eq_ranges_x_lambda}
	0\leqslant x <1, \quad \max \left\{0,(2x-1)/x\right\} \leqslant\lambda\leqslant1
\end{align}
Inequalities \eqref{eq_ranges_x_lambda} define the subset $\mathcal R$ of the unit square (see Fig.~\ref{fig:fig_RegionR}).

\begin{remark}[correlation in the Klotz model]
	\label{remark_3_caes_of_dependance}
	The correlation coefficient for two successive executions is $\frac{\lambda - x}{1-x}{\bf 1}_{0\leqslant x<1} + {\bf 1}_{x=1}$. This defines 3 correlation types: {\bf i)} when $x=\lambda$, the model exhibits independent execution outcomes; {\bf ii)} when $\lambda>x$, execution outcomes tend to cluster more often (e.g. bursts of failures) -- positive correlation; and {\bf iii)} when $\lambda<x$, execution outcomes tend to alternate more often, between failure and success -- negative correlation.
\end{remark}

Over $n$ executions, suppose the software makes: $\alpha$ transitions from a successful execution to a failed execution, $\beta$ transitions from ``success'' to ``success'', $\gamma$ from ``failure'' to ``failure'', and $\delta$ from ``failure'' to ``success''. So, $\alpha+\beta+\gamma+\delta + 1=n$. Under the Klotz model, for given $(x,\lambda)$, the probability of observing these transitions is the likelihood $L(x,\lambda ; \alpha, \beta, \gamma, \delta)$: 
\begin{equation}
	L(x,\lambda;\alpha,\beta,\gamma,\delta) = \left\{ \begin{array}{lr}
		x\left(\frac{(1-\lambda)x}{1-x}\right)^{\alpha}\left(1-\frac{(1-\lambda)x}{1-x}\right)^{\beta}\lambda^{\gamma}(1-\lambda)^{\delta}; & \\
		\mbox{when the 1st execution is a failure} & \\ &
		\\
		(1-x)\left(\frac{(1-\lambda)x}{1-x}\right)^{\alpha}\left(1-\frac{(1-\lambda)x}{1-x}\right)^{\beta}\lambda^{\gamma}(1-\lambda)^{\delta}; &\\
		\mbox{when the 1st execution is a success}&\end{array}\right.
	\label{eqn_xKlotzlklhdFn_maintxt}
\end{equation}

During operational testing, an assessor observes $\alpha$, $\beta$, $\gamma$ and $\delta$. But both $x$ and $\lambda$ are unknown to the assessor\footnote{In theory, knowing these values would allow the assessor to completely characterise the system's failure process; e.g., one may then compute the probability of future failure-free operation over a sequence of software executions.}; in particular, the assessor is uncertain about the \emph{pfe} $X$. So, upon observing $n$ executions of the system, an assessor's confidence in $X$ being no larger than a bound $b$ is:
\begin{align}
	P(X\leqslant b\,\mid\,\mbox{\it\ outcomes of }n\mbox{\it\ executions}) 
	=\frac{P(X\leqslant b\,,\,\mbox{\it\ outcomes of }n\mbox{\it\ executions})}{P(\mbox{\it\ outcomes of }n\mbox{\it\ executions})} 
	=\frac{\MyExp[L(X,\Lambda;\alpha,\beta,\gamma,\delta){\bf 1}_{X\leqslant b}]}{\MyExp[L(X,\Lambda;\alpha,\beta,\gamma,\delta)]}
	\label{eqn_Klotzmodelpostconf}    
\end{align}
which generalises \eqref{eqn_post_cf_bound_with_complete_prior}. 
It will be useful to refer to $s$ -- the number of failed executions --  and $r$ -- the number of failed executions preceded by a failed execution (``consecutive failures'', for short). The likelihood \eqref{eqn_xKlotzlklhdFn_maintxt} can be re-expressed in terms of $n$, $s$ and $r$, which are related to $\alpha, \beta, \gamma, \delta$ (see supplementary material, \ref{sec_LikeFnAltForms}). In particular, $r=\gamma$.

\section{Conservative Upper Confidence Bounds on Probability of Failure per Execution}
\label{sec_conf_bound_in_reliability}
\subsection{Practical Context for Applying CBI techniques}\label{subsec_usingCBIinPractice}
Upon observing $n$ executions, what is the least confidence an assessor should have about $X$ being ``no bigger than'' the bound $b$? 
We determine this for different scenarios by deriving the greatest lower bound for \eqref{eqn_Klotzmodelpostconf} using PKs \ref{cons_pl_lowerbound}, \ref{cons_engineering_goal}, \ref{cons_negative_dependence}, \ref{cons_positive_dependence}. Section \ref{sec_CBIreview} introduced PKs  \ref{cons_pl_lowerbound}, \ref{cons_engineering_goal}, while section \ref{sec_basline_post_reliability} introduces PKs \ref{cons_negative_dependence}, \ref{cons_positive_dependence}.

Practical implications of these results -- with domain-specific PK parameterisations for nuclear power-plant safety protection systems and AV safety subsystems -- are given in sections \ref{sec_nuclear_post_reliability} and \ref{sec_av_post_reliability} respectively. The parameterisations are illustrative of plausible PK values when assessing functionally redundant software components for safety systems employing fault-tolerance -- e.g., systems highlighted in \cite{Wood_Nureg_2010,koopman2016challenges, horwick2010strategy, LuiSha_2001}.

These CBI techniques/results -- for conservatively gaining confidence in the software being sufficiently reliable -- are primarily intended for use in assessments based on operational/statistical testing, where the software is treated as a black-box. In operational/statistical testing, the test cases for the software are randomly generated software inputs. To do this correctly, the probability of generating a given test case must be consistent with an ``\emph{operational profile}'' -- i.e., this probability must be the same as the probability of the same ``case'' occurring when the software is deployed in real operation. Test cases can take different forms, depending on the type of software under study. For example, consider a batch program that receives numerical values for all of its input variables at the start of an execution, it executes, and then it produces all of its outputs. A test case for such a program could be a fixed-length vector of numerical values --  each number in the vector is the value for an input variable. Alternatively, consider a control program that receives multiple numerical values for each input variable over time; here, a test case could be a collection of numerical sequences -- each sequence represents the changing values over time for an input variable. More examples of test cases can be found in \cite{MichaelLyu_1996} and \cite{Strigini_Littlewood_EuropeanSpaceAgency_1997}.

Statistical testing supports direct estimates of reliability, for the purposes of reliability assessment and product acceptance. It also supports decisions on whether software is ready for use in a specific system. Thus, CBI techniques can be applied in any software development testing phase where statistical testing may be applied, and where decisions may be taken on whether software is ready to be deployed; e.g., integration or acceptance testing.

Of course, good practice for carrying out statistical testing must be followed when employing our theorems and results; e.g., \cite{Strigini_Littlewood_EuropeanSpaceAgency_1997} give detailed guidance in this regard. See also \cite{SalakoZhao_TSE_2023} for more discussion. Best practice approaches for expert belief elicitation should be followed to elicit the PKs; e.g., \cite{1994_NuRegTechreport,OHagan_2006}.

Table \ref{tab_all_wcp_examples} lists the worst-case priors used for the curves in the example plots. While these priors are consistent with those in \cite{SalakoZhao_TSE_2023} (for assessors that can justify continuous marginal prior distributions of $X$), the current priors are applicable in many more scenarios (e.g., when assessors can justify only very limited properties of the marginal prior distribution of $X$, in the form of PKs \ref{cons_pl_lowerbound} and \ref{cons_engineering_goal}).

\begin{figure}[h!]
	\captionsetup[figure]{format=hang}
	\centering	
	\begin{subfigure}[]{0.3\linewidth}
		\centering
		\includegraphics[width=1.0\linewidth]{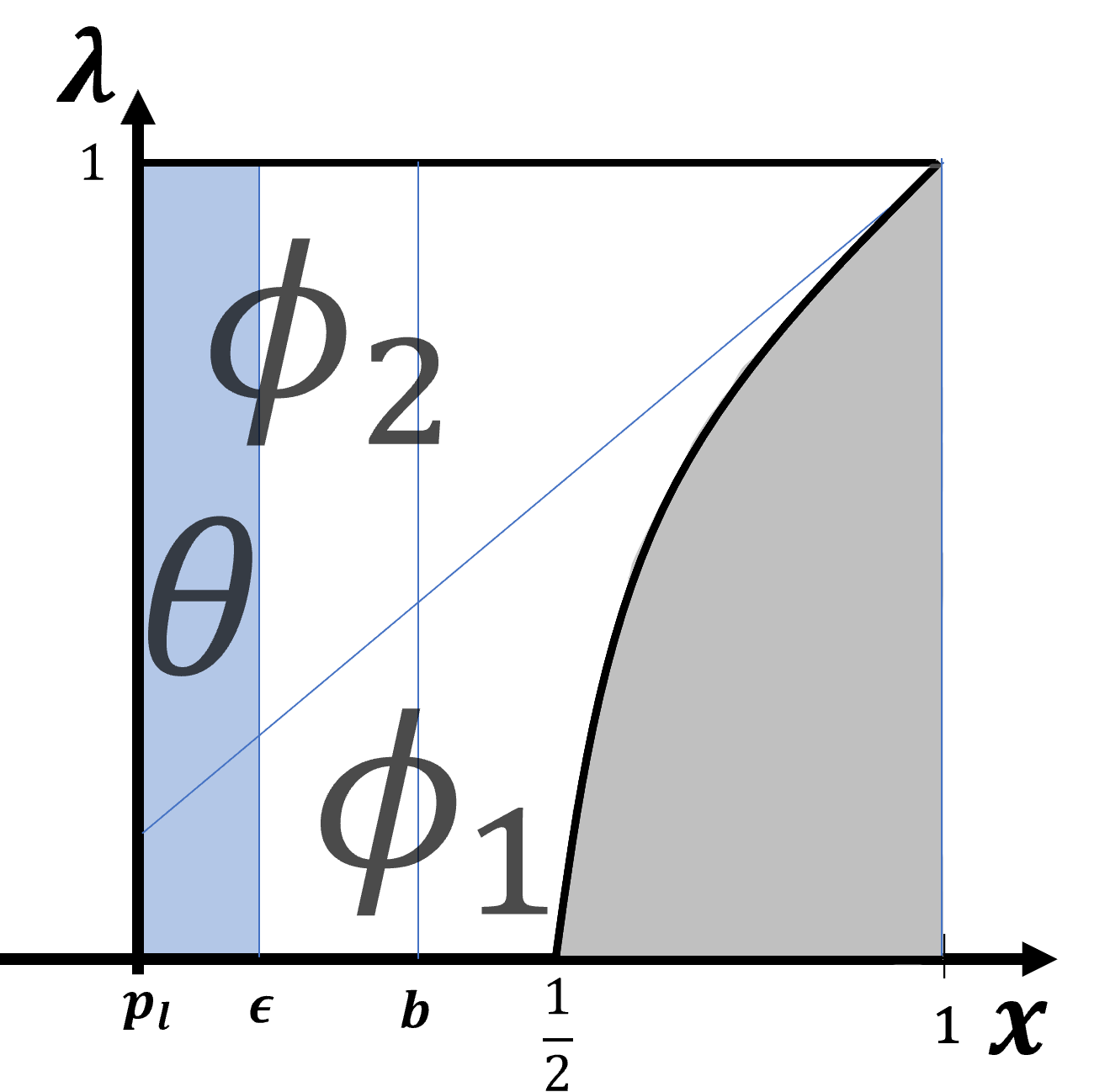}
		\caption{{\footnotesize  }}
		\label{fig:fig_RegionR}
	\end{subfigure}
	\\
	\begin{subfigure}[]{0.3\linewidth}
		\centering
			\includegraphics[width=1.0\linewidth]{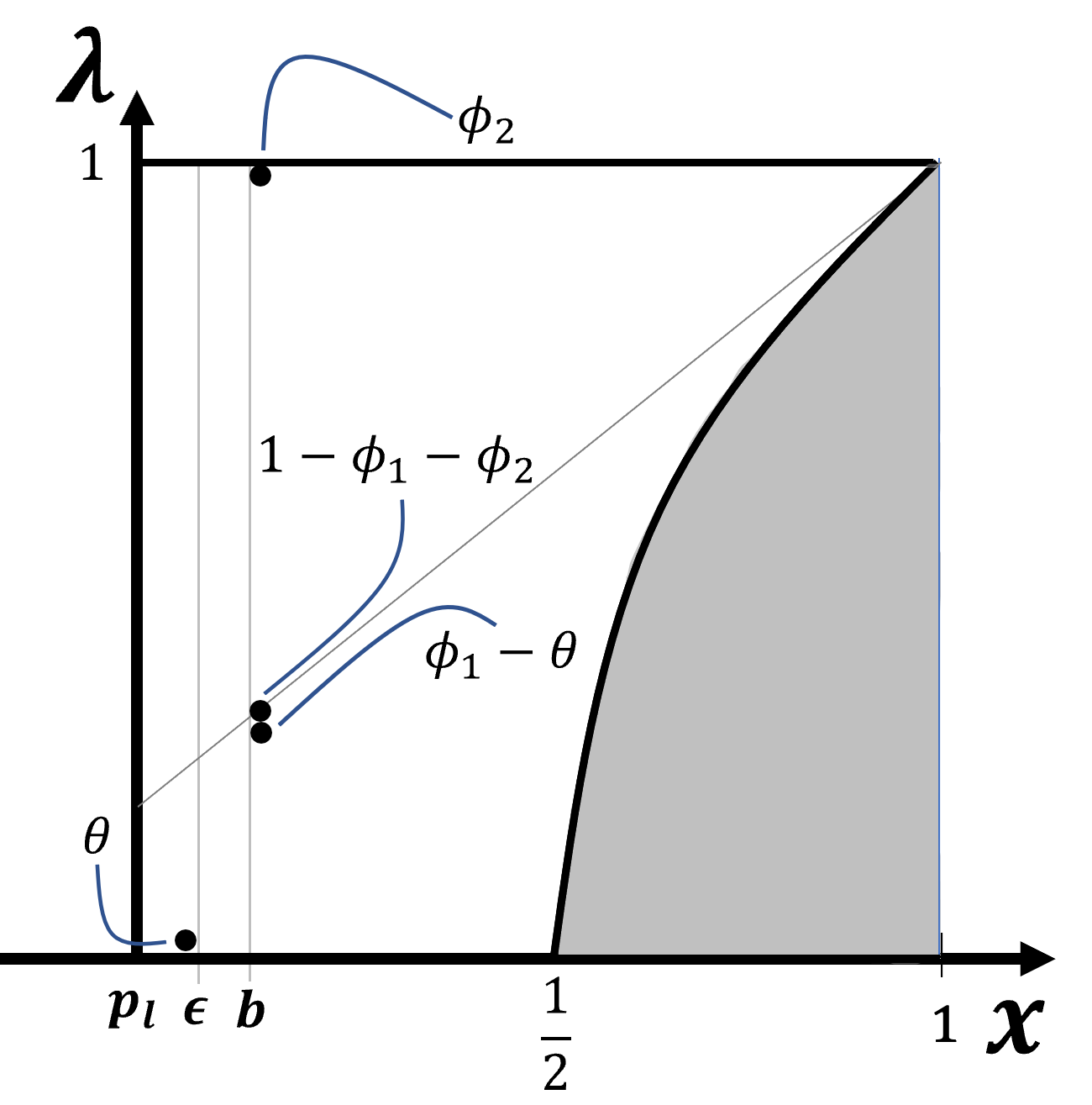}
		\caption{{\footnotesize $\phi_1\geqslant\theta\,$ }}
		\label{fig:fig_CBInoFails_maintxt_Phi1gtTheta}
	\end{subfigure}
	\begin{subfigure}[]{0.3\linewidth}
		\centering
			\includegraphics[width=1.0\linewidth]{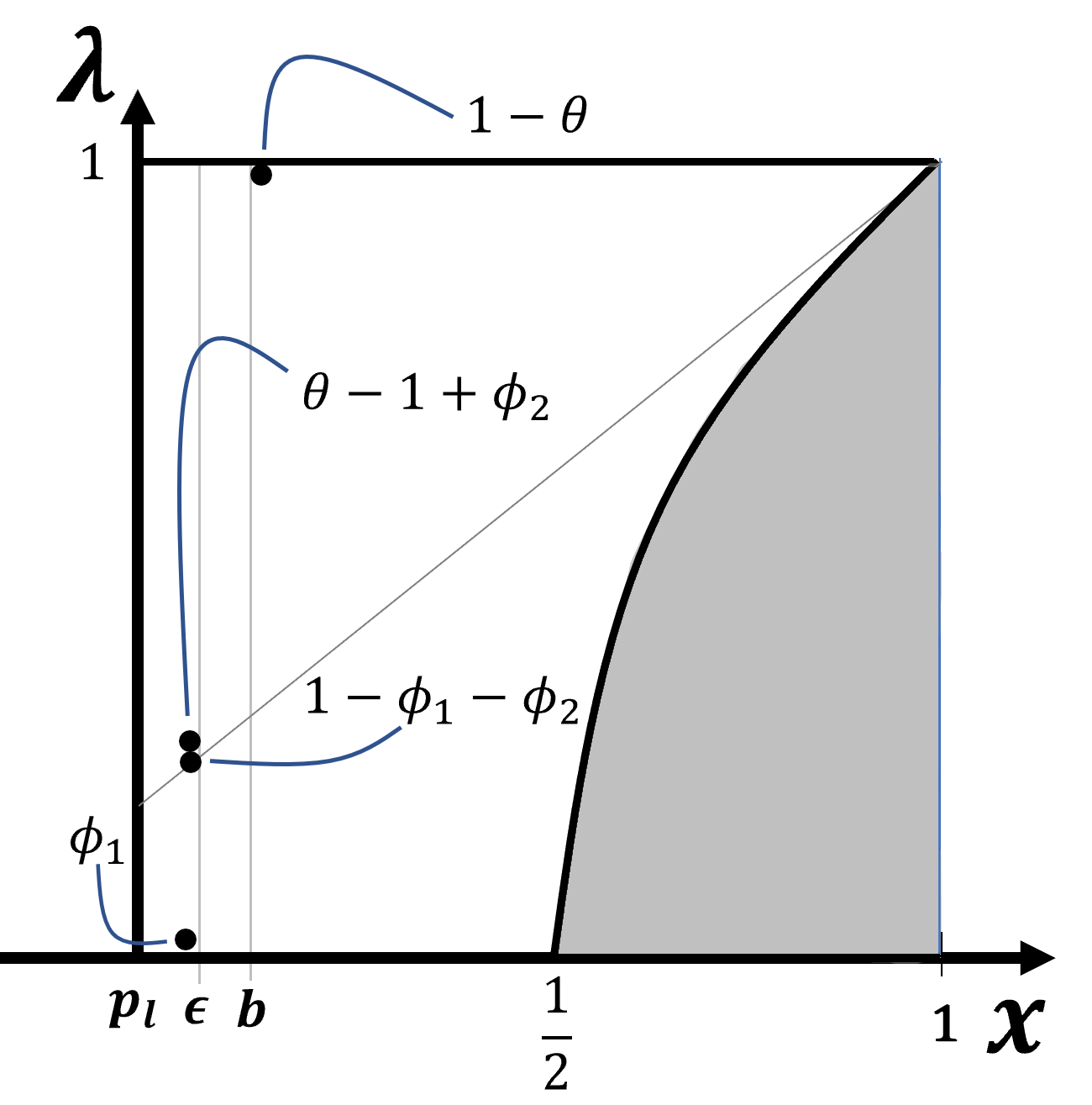}
		\caption{{\footnotesize $\ \phi_2\geqslant 1-\theta$  }}
		\label{fig:fig_CBInoFails_maintxt_Phi1ltThetaPhi2gt1minusTheta}
	\end{subfigure}
	\begin{subfigure}[]{0.3\linewidth}
		\centering
			\includegraphics[width=1.0\linewidth]{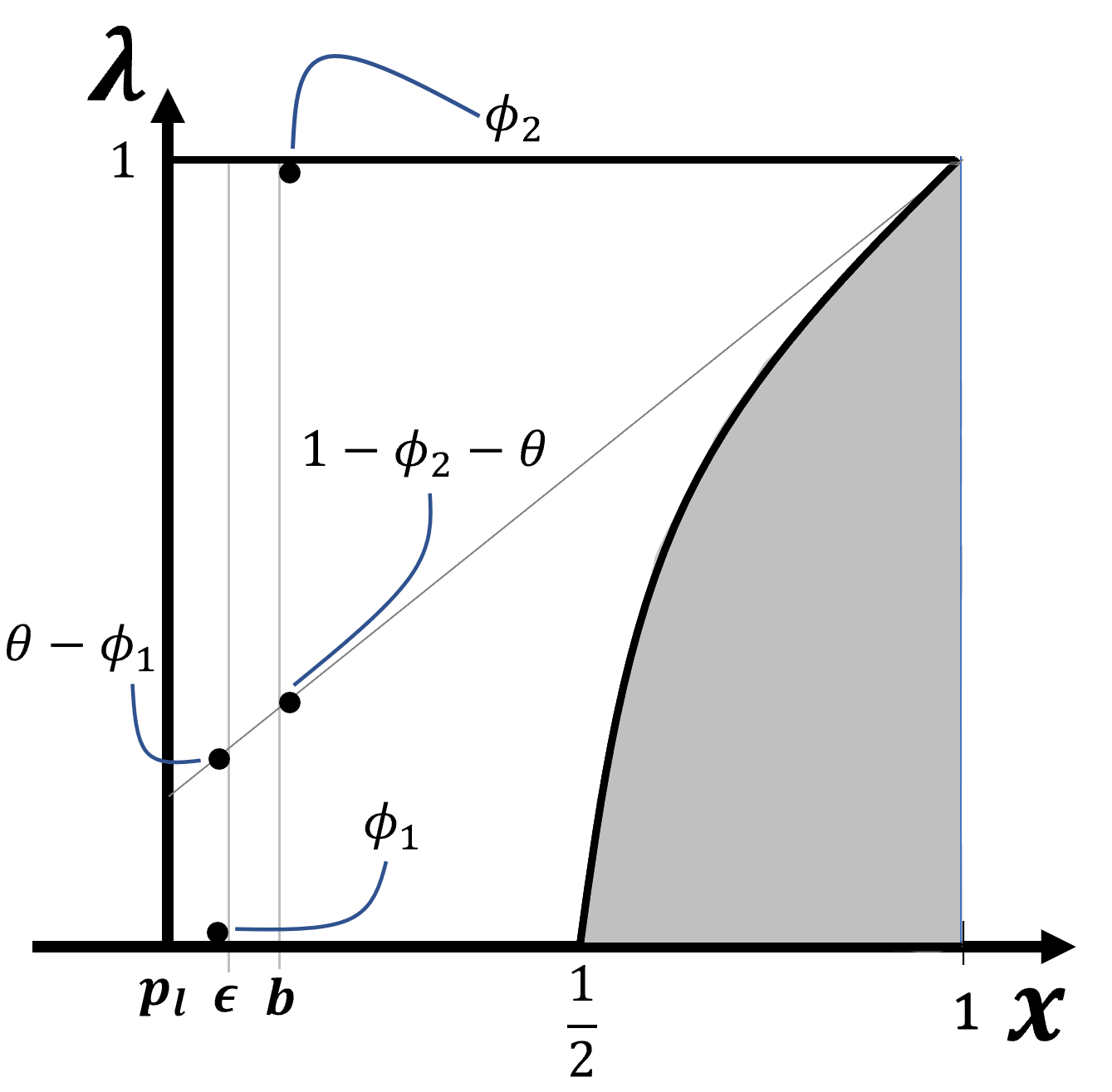}
		\caption{{\footnotesize $\phi_1\leqslant\theta\ $ and $\ \phi_2\leqslant 1-\theta$ }}
		\label{fig:fig_CBInoFails_maintxt_Phi1ltThetaPhi2lt1minusTheta}
	\end{subfigure}
	\caption[Region $\mathcal R$ and CBI priors: no failures]{{\footnotesize  The support $\mathcal R$ of the Klotz likelihood function, and the subsets of $\mathcal R$ related to PKs
			\ref{cons_pl_lowerbound},
			\ref{cons_engineering_goal}, \ref{cons_negative_dependence} and \ref{cons_positive_dependence}, are shown in subfig.~\ref{fig:fig_RegionR}. Upon observing executions with {\bf no failures}, subfig.s \ref{fig:fig_CBInoFails_maintxt_Phi1gtTheta}, \ref{fig:fig_CBInoFails_maintxt_Phi1ltThetaPhi2gt1minusTheta} and \ref{fig:fig_CBInoFails_maintxt_Phi1ltThetaPhi2lt1minusTheta} are 3 prior distributions that solve the optimisation problem in Theorem \ref{theorem_thm1_baseline}. These priors are relevant for the ranges of parameter values indicated in each subfigure. \normalsize}}
	\label{fig_CBInoFails_maintxt}
\end{figure}

\subsection{Assessment with Doubts about i.i.d. Executions}
\label{sec_basline_post_reliability}
Our first assessment scenario is a baseline. Before operational testing, an assessor uses reliability evidence to justify the engineering goal PK \ref{cons_engineering_goal}, and the following two PKs about the independence assumption (\emph{cf.} Remark \ref{remark_3_caes_of_dependance}):
\begin{constraint}[confidence in negative dependence]
	\label{cons_negative_dependence}
	$\phi_1 \times 100\%$ confidence in successive software executions having negative dependence; i.e., $P(\Lambda < X)=\phi_1$.
\end{constraint}
\begin{constraint}[confidence in positive dependence]
	\label{cons_positive_dependence}
	$\phi_2 \times 100\%$ confidence in successive software executions having positive dependence; i.e., $P(\Lambda > X)=\phi_2$.
\end{constraint}
\noindent Consequently, the assessor's \textit{prior confidence in independence} is $(1-\phi_1-\phi_2)$, i.e. $P(\Lambda = X)=1-\phi_1-\phi_2$. Note that in all of the remaining theorems in this paper, $\phi_1=\phi_2=0$ is the special case of i.i.d. execution outcomes -- in this limit, all the theorems agree with previously published univariate CBI results.

An assessor observes all $n$ executions during testing are successful. Using the Klotz model, the CBI problem of determining the least amount of confidence the assessor can justifiably have, about the system \emph{pfe} satisfying bound $b$, is the following constrained optimisation problem. 
Consider the support $\mathcal R$ of the Klotz likelihood, defined by \eqref{eq_ranges_x_lambda} and depicted in Fig.~\ref{fig:fig_RegionR}. Let $\mathcal D$ be the set of all prior probability distributions over $\mathcal R$, and $0\leqslant p_l\leqslant \epsilon <b<\frac{1}{2}$. Then, the following theorem holds (generalisation proved in supplementary material, \ref{sec_proofThrm1}).

\begin{theorem}
	Using \eqref{eqn_xKlotzlklhdFn_maintxt} and \eqref{eqn_Klotzmodelpostconf}, the optimisation problem
	\begin{align*}
		&\qquad\inf\limits_{\mathcal D} P(\,X\leqslant b \mid n\mbox{ executions without failure} ) \\
		&\mbox{s.t.} \,\,\,\,\,\,PK\ref{cons_pl_lowerbound},\,\,\,PK\ref{cons_engineering_goal},\,\,\,PK\ref{cons_negative_dependence},\,\,\,PK\ref{cons_positive_dependence} 
	\end{align*}
	is solved by the prior distributions in Fig.s~\ref{fig:fig_CBInoFails_maintxt_Phi1gtTheta}, \ref{fig:fig_CBInoFails_maintxt_Phi1ltThetaPhi2gt1minusTheta} and \ref{fig:fig_CBInoFails_maintxt_Phi1ltThetaPhi2lt1minusTheta}, since $P(\, X<b \mid n\mbox{ executions without failure} )$ from these priors equals the infimum.
	\label{theorem_thm1_baseline}
\end{theorem}


Each of Fig.s~\ref{fig:fig_CBInoFails_maintxt_Phi1gtTheta}, \ref{fig:fig_CBInoFails_maintxt_Phi1ltThetaPhi2gt1minusTheta} and \ref{fig:fig_CBInoFails_maintxt_Phi1ltThetaPhi2lt1minusTheta} shows the domain $\mathcal R$ of a joint prior distribution for $(X,\Lambda)$ random variables, and the $4$ points (i.e., black dots) in $\mathcal R$ assigned nonzero probabilities by this distribution. So, these joint priors are depicted as if one were looking down on the distribution and its domain ``from above''.

\begin{example}[baseline]
	\label{exp_baseline_scenario}
	Consider an on-demand SCS which acts only upon receipt of a demand from its environment. An assessor is $75\%$ confident the software's pfe -- i.e., its probability of failure per demand (pfd) -- is no worse than $\epsilon=10^{-5}$ (i.e. the engineering goal of PK\ref{cons_engineering_goal} with $\theta=0.75$). After $n$ failure-free tests of the SCS, the assessor is $c\times 100\%$ confident that the system meets Safety Integrity Level (SIL) 4, i.e. $b=10^{-4}$ (see \cite{iec_61508_2010}). Fig.~\ref{fig_example1_baseline} shows three plots of $c$ as a function of $n$ using three Bayesian models\footnote{Note, with failure-free executions, the posterior confidence from CBI is not a function of $p_l$ (since the distributions in Fig.~\ref{fig_CBIsoln_noFails} have no probabilities along the $p_l$ line). Thus, \emph{w.l.o.g.}, we set $p_l=0$ in PK\ref{cons_pl_lowerbound}.\xingyu{actually, we can only use $p_l=0$ to compare with the beta prior...}\kizito{I wonder if we could, perhaps, use a ``re-scaled'' beta prior in the $p_l>0$ case?}}: univariate CBI (cf. Theorem \ref{theorem_univariateCBI}),
	Bayesian Inference (BI) using a Beta prior\footnote{Specifically, to illustrate with a $Beta(\alpha,\beta)$, we first set the $\alpha$ parameter to $0.03$, then fitted a ``$\beta$'' value to satisfy the quantile in PK\ref{cons_engineering_goal}.} satisfying PK\ref{cons_engineering_goal}, and CBI with confidence $\phi_1$ and $\phi_2$ in negative and positive dependence respectively (cf. Theorem \ref{theorem_thm1_baseline}).
\end{example}
\begin{figure}[hbtp!]
	\centering
	\includegraphics[width=0.45\linewidth]{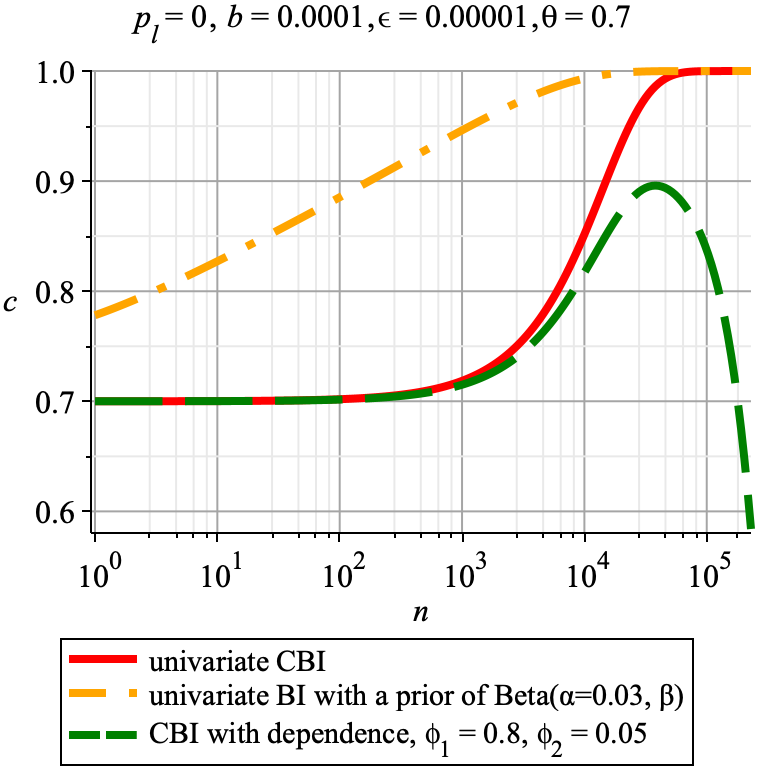}
	\caption{(Example \ref{exp_baseline_scenario}) $c\times 100\%$ posterior confidence in $[X\leqslant 10^{-4}]$ upon seeing $n$ failure-free tests, from three Bayesian models with different PKs.\kizito{Perhaps have the curves here match the curves in Fig.~\ref{fig_example2_nuclear}.\xingyu{Done, and also changed the table of priors..}}}
	\label{fig_example1_baseline}
\end{figure}
\kizito{TODO: what is the value of the beta parameter for the beta prior?}

Univariate CBI and BI (using a Beta prior) both assume the executions are statistically independent (so, have likelihood $L(x;n)=(1-x)^n$). In Fig.~\ref{fig_example1_baseline}, as expected, univariate CBI gives less confidence (so, is more conservative) than BI using a Beta prior\footnote{This observation holds for any feasible prior, not only Beta distributions.}; while both of these are more optimistic than CBI with doubts about independence. Failure-free evidence is always ``good news'' under univariate CBI and BI (i.e., the solid and dashed-dot curves monotonically increase to ``certainty'' in the bound). Contrastingly, with doubts in independence, such evidence can eventually undermine an assessor's posterior confidence in the $10^{-4}$ bound for \emph{all} sufficiently large $n$ (i.e., the uni-modal pattern of the dashed curve). This is because there are pessimistic reasons for why the failure-free tests could be occurring -- reasons that suggest the successes are occurring despite the SCS not being very reliable. For example, the test-cases might be unrepresentatively ``easy'' for the SCS to correctly respond to, or there may be problems with the test oracle which mean some failures go undetected, see \cite{littlewood_use_2007,Barr_Oracleproblem_2015,SalakoZhao_TSE_2023}.

So, as failure-free evidence accumulates, \emph{any} prior confidence the assessor has in the executions being positively correlated -- i.e., any $\phi_2>0$ -- will eventually undermine confidence in \emph{any} \emph{pfd} upper bound. On the other hand, prior confidence in negatively correlated executions (i.e. $\phi_1>0$) has a negligible impact on posterior confidence in the bound (see section \ref{sec_senstivity_analysis}'s sensitivity analysis). Intuitively, the longer testing goes on for without failure, the greater the evidence \emph{against} the tests being negatively dependent. An example of how negative dependence can occur is if an assessor intentionally tries to ``stress'' the software during testing, by randomly including a disproportionate number of ``difficult'' demands -- demands that are thought will likely cause software failure. So one might expect testing to reveal some negative dependence -- a failure quickly followed by a success, then followed by another failure relatively soon afterwards, and so on. However, ``no failures'' may suggest the ``difficult'' demands are not actually difficult for the software.  

\begin{table*}[h!]
	\centering
	\caption{A list of the prior distributions for the curves in the example plots of Section \ref{sec_conf_bound_in_reliability}. See supplementary material for any listed figures not included in the paper.}
	\label{tab_all_wcp_examples}
	\begin{tabular}{|ccccccc|}
		\hline
		\hline
		\multicolumn{2}{|c|}{} &
		\multicolumn{1}{c|}{red solid} &
		\multicolumn{1}{c|}{green dashed} &
		\multicolumn{1}{c|}{blue dotted} &
		\multicolumn{1}{c|}{orange (dash)dotted} &
		\multicolumn{1}{c|}{pink spacedashed} \\ \hline
		\multicolumn{2}{|c|}{Example \ref{exp_baseline_scenario} (Fig.~\ref{fig_example1_baseline})} &
		\multicolumn{1}{c|}{Fig.~\ref{fig_unicbi_wc_prior}} &
		\multicolumn{1}{c|}{Fig.~\ref{fig:fig_CBInoFails_maintxt_Phi1gtTheta}} &
		\multicolumn{1}{c|}{n/a} &
		\multicolumn{1}{c|}{a Beta distribution$\ast$} &
		\multicolumn{1}{c|}{n/a} 
		\\ \hline
		\multicolumn{2}{|c|}{Example \ref{exp_nuclear_scenario} (Fig.~\ref{fig_example2_nuclear})} &
		\multicolumn{1}{c|}{Fig.~\ref{fig_unicbi_wc_prior}} &
		\multicolumn{1}{c|}{Fig.~\ref{fig:fig_CBInoFails_maintxt_Phi1gtTheta}} &
		\multicolumn{1}{c|}{Fig.~\ref{fig:fig_CBInoFails_maintxt_Phi1gtTheta}} &
		\multicolumn{1}{c|}{n/a} &
		\multicolumn{1}{c|}{n/a} 
		\\ \hline
		\multicolumn{2}{|c|}{Example \ref{exp_av} (Fig.~\ref{fig_av_with_failures_example3})} &
		\multicolumn{1}{c|}{Fig.~\ref{fig_unicbi_wc_prior}} &
		\multicolumn{1}{c|}{Fig.~\ref{fig:fig_CBInoFails_maintxt_Phi1ltThetaPhi2lt1minusTheta}} &
		\multicolumn{1}{c|}{\begin{tabular}[c]{@{}c@{}}  $\ast\ast$Fig.s~\ref{fig:fig_CBIsoln_withnoconsFails_Phi1gt1minustheta_1},\\ \ref{fig:fig_CBIsoln_withnoconsFails_Phi1gt1minustheta_2}, \ref{fig:fig_CBIsoln_withnoconsFails_Phi1gt1minustheta_3},\\ \ref{fig:fig_CBIsoln_withnoconsFails_Phi1gt1minustheta_4} as $n$ \\ increases \end{tabular} } &
		\multicolumn{1}{c|}{\begin{tabular}[c]{@{}c@{}} $\ast\ast$Fig.s~\ref{fig:fig_CBIsoln_withnoconsFails_Phi1gt1minustheta_1},\\ \ref{fig:fig_CBIsoln_withnoconsFails_Phi1gt1minustheta_2}, \ref{fig:fig_CBIsoln_withnoconsFails_Phi1gt1minustheta_3},\\ \ref{fig:fig_CBIsoln_withnoconsFails_Phi1gt1minustheta_4} as $n$ \\ increases \end{tabular} } & 
		\multicolumn{1}{c|}{\begin{tabular}[c]{@{}c@{}} $\ast\ast$Fig.s~\ref{fig:fig_CBIsoln_withFails_Phi2lt1minustheta_1},\\ \ref{fig:fig_CBIsoln_withFails_Phi2lt1minustheta_2}, \ref{fig:fig_CBIsoln_withFails_Phi2lt1minustheta_3},\\ \ref{fig:fig_CBIsoln_withFails_Phi2lt1minustheta_4} as $n$ \\ increases \end{tabular}} 
		\\ \hline \hline
	\end{tabular}
	\\
	{\footnotesize $\ast$An arbitrary Beta distribution satisfying the PKs. \\
		$\ast\ast$This curve is a piecewise function -- i.e. it's the confidence from the listed priors, in sequence, as $n$ increases. The precise values $n$ at which the curve switches between confidence from different priors depends on the execution outcomes (see proof in supplementary material, \ref{sec_app_B}).
	}
\end{table*}

\subsection{Assessment with PKs for Nuclear Reactor Safety Systems}
\label{sec_nuclear_post_reliability}
Next consider the assessment of a nuclear reactor safety protection system that is simple enough to possibly be fault-free (i.e., the system's \emph{pfd} could be $0$), see \cite{littlewood_reasoning_2012}. Typically, failure-free operational testing from such a system is required -- otherwise, if a failure occurs, the system is taken offline and fixed, before testing resumes with a new sequence of demands. This scenario is very similar to the baseline of the last section -- the testing evidence and most of the PKs are the same -- but now, the engineering goal is ``perfection'' (i.e., PK\ref{cons_engineering_goal} with $\epsilon=0$).

As before, we are interested in an assessor's confidence in a \emph{pfd} bound $b$ upon seeing $n$ failure-free runs, where the assessor harbours doubts about the execution outcomes being i.i.d.  (i.e., nonzero $\phi_1$ or $\phi_2$). This confidence in $b$ is given by Theorem~\ref{theorem_thm1_baseline}, simply by replacing PK\ref{cons_engineering_goal} with $P(X=0)=\theta$; that is, $\epsilon=0$ in the distributions of Fig.s~\ref{fig:fig_CBInoFails_maintxt_Phi1gtTheta}, \ref{fig:fig_CBInoFails_maintxt_Phi1ltThetaPhi2gt1minusTheta} and \ref{fig:fig_CBInoFails_maintxt_Phi1ltThetaPhi2lt1minusTheta}.

\begin{example}[nuclear reactor protection systems]
	\label{exp_nuclear_scenario}
	Consider a nuclear reactor safety protection system that an assessor is $70\%$ confident contains no faults (i.e., PK\ref{cons_engineering_goal} with $\epsilon=0$ and $\theta=0.7$). Upon seeing $n$ failure-free tests, an assessor's conservative posterior confidence in the pfd bound $10^{-4}$ (SIL 4) is shown in Fig.~\ref{fig_example2_nuclear}, for three Bayesian models with different PKs: univariate CBI with $\epsilon=0$, CBI with doubts in the independence assumption and $\epsilon=10^{-5}$ (i.e., the baseline Example \ref{exp_baseline_scenario}), and CBI with doubts in independence and $\epsilon=0$.
\end{example}

\begin{figure}[hbtp!]
	\centering
	\includegraphics[width=0.45\linewidth]{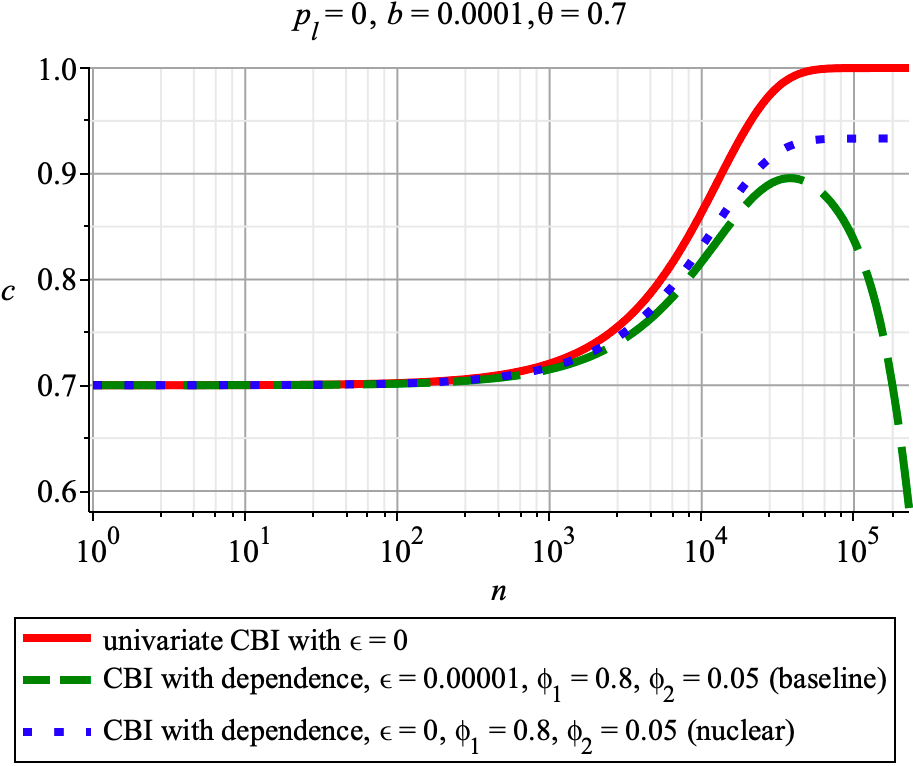}
	\caption{(Example \ref{exp_nuclear_scenario}) $c\times 100\%$ posterior confidence in $[X\leqslant 10^{-4}]$ upon seeing $n$ failure-free tests, from three Bayesian models with different PKs.\kizito{Perhaps have the curves here match the curves in Fig.~\ref{fig_example1_baseline}.}}
	\label{fig_example2_nuclear}
\end{figure}

Example \ref{exp_nuclear_scenario} highlights the benefit of the software possibly being fault-free: in contrast to Example \ref{exp_baseline_scenario}, accumulated failure-free evidence \emph{will not} eventually undermine posterior confidence in $b$. That is, the dotted curve in Fig.~\ref{fig_example2_nuclear} is an increasing function of $n$ that is asymptotic to the horizontal line $c=\frac{\theta}{\theta+(1-b)\phi_2}$, so it lies above the dashed curve\footnote{The asymptote is obtained by setting $\epsilon=0$ in the distribution of  Fig.~\ref{fig:fig_CBInoFails_maintxt_Phi1gtTheta}, and computing $\lim\limits_{n\rightarrow\infty}P(\, X<b \mid n\mbox{ executions without failure} )$.}, yet lies below the solid curve (i.e., its more conservative than the confidence from univariate CBI).

\begin{insight}[For a possibly perfect system, failure-free testing cannot undermine confidence in a bound]
	\label{insight_perf_prevent_cor}
	As more successful executions are observed, the more likely it is that these observations are the result of either a fault-free system (so $\epsilon=0$) or a perfectly positively correlated system (so $\lambda=1$). That is, as $n$ increases, the distribution in Fig.~\ref{fig:fig_CBInoFails_maintxt_Phi1gtTheta} tends to a distribution that has probability mass at only two points: a probability $\frac{\theta}{\theta+(1-b)\phi_2}$ at $(0,0)$, and a complementary probability $\frac{(1-b)\phi_2}{\theta+(1-b)\phi_2}$ at $(b,1)$. In the limit of large $n$, the assessor will need more evidence to distinguish between these two possibilities. 
\end{insight}

The sensitivity of these insights -- to changes in the strength of the PKs -- is explored in section \ref{sec_senstivity_analysis}. While further implications of these insights are discussed in section \ref{sec_discussion}.

\subsection{Assessment with PKs for Autonomous Vehicles}
\label{sec_av_post_reliability}
We turn our attention to assessing AV-safety. In line with \cite{kalra_driving_2016} and \cite{zhao_assessing_2019}, the \emph{pfe} is the \textit{probability of a fatality-event per mile} (\textit{pfm}). Here, each mile is treated as a ``unit distance'' over which software that enacts AV safety functions must correctly operate\footnote{This is a coarse model of operation represented by a Bernoulli process: AV software responds to, say, discrete visual/positional stimulus it receives at a constant rate per mile. It would be interesting to see if the conclusions from the following analyses change significantly with a more sophisticated model of continuous operation (e.g., failures occurring according to an NHPP).}. Unlike the (possibly fault-free) protection software of subsection \ref{sec_nuclear_post_reliability}, AV software cannot be expected to be fault-free -- it relies on imperfect, sophisticated machine learning solutions performing a complex driving task. Consequently, unfortunately, AV safety software failures (leading to fatalities, in particular) have been known to occur; see \cite{2022_NHTSAreport}. This suggests a nonzero lower bound $p_l$ on the \emph{pfm} (see PK\ref{cons_pl_lowerbound}), and an engineering goal of a ``safe enough''\footnote{ How safe is ``safe enough'' is a separate topic, see \cite{liu_how_2019}. A typical target would be several times safer than the average human driver.} system rather than ``perfection''. 

The following theorem gives conservative confidence in the \emph{pfm} bound $b$ (see general proof in supplementary material, \ref{sec_app_B}). Failures change the form of the Klotz likelihood from the one in Theorem \ref{theorem_thm1_baseline} (see \eqref{eqn_xKlotzlklhdFn_maintxt}).  Here, an assessor doubts the independence of successive software executions as the AV drives over successive miles.

\begin{theorem}
	Using \eqref{eqn_xKlotzlklhdFn_maintxt} and \eqref{eqn_Klotzmodelpostconf}, the optimisation problem
	\begin{align*}
		&\,\,\inf\limits_{\mathcal D} P(\,X\leqslant b \mid n\mbox{ executions, }s\mbox{ failures, }r\mbox{ consecutive failures} ) \\
		&\mbox{s.t.} \,\,\,\,\,\,PK\ref{cons_pl_lowerbound},\,\,\,PK\ref{cons_engineering_goal},\,\,\,PK\ref{cons_negative_dependence},\,\,\,PK\ref{cons_positive_dependence}
	\end{align*}
	is solved by prior distributions such as those in Fig.~\ref{fig_CBIsoln_pl}, since $P(\,X< b \mid n\mbox{ executions, }s\mbox{ failures, }r\mbox{ consecutive failures} )$ from these priors equals the infimum.
	\label{theorem_CBI_withFailures_and_pl}
\end{theorem}

\begin{figure}[h!]
	\captionsetup[figure]{format=hang}
	\centering	
	\begin{subfigure}[]{0.3\linewidth}
		\centering
		\includegraphics[width=1.0\linewidth]{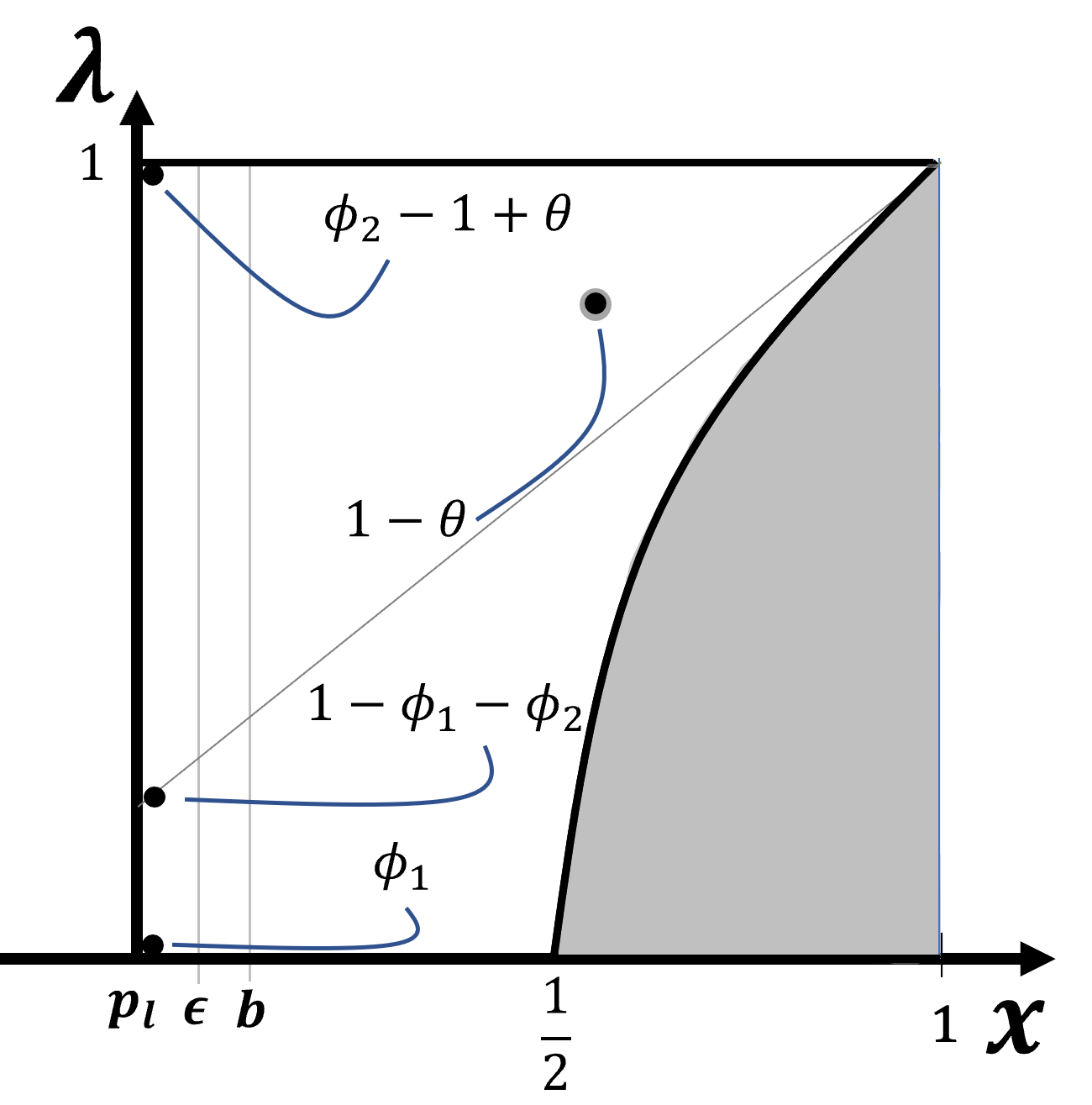}
		\caption{{\footnotesize $\,\,\phi_2\geqslant 1-\theta$ }}
		\label{fig:fig_CBIsoln_plgeneralFailsPhi2gt1minTheta}
	\end{subfigure}
	\\
	\begin{subfigure}[]{0.3\linewidth}
		\centering
			\includegraphics[width=1.0\linewidth]{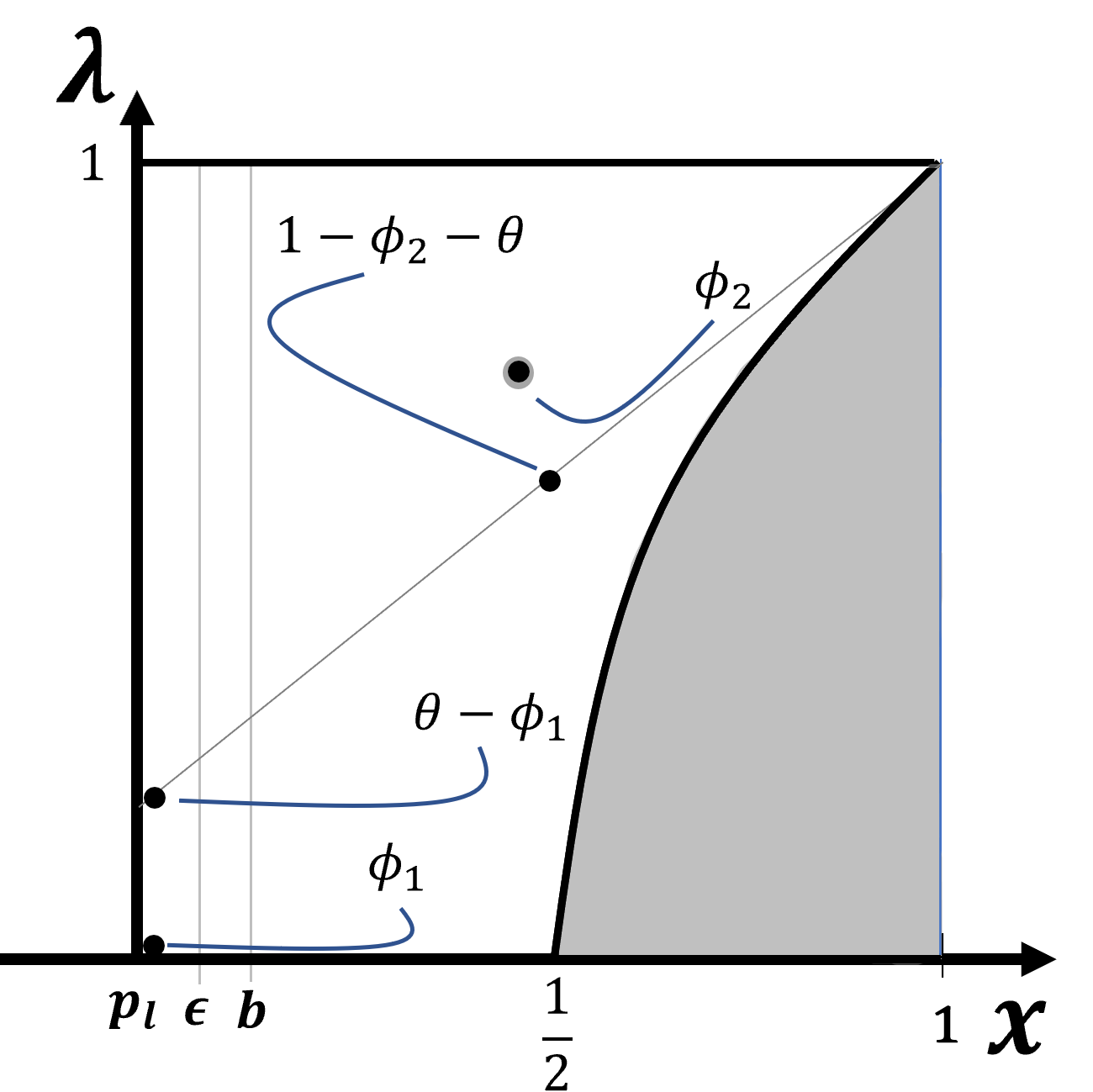}
		\caption{{\footnotesize $\,\,\phi_1\leqslant\theta\,\,$ and $\,\,\phi_2\leqslant 1-\theta$ }}
		\label{fig:fig_CBIsoln_plgeneralFailsPhi2lt1minTheta_Phi1ltTheta}
	\end{subfigure}
	\begin{subfigure}[]{0.3\linewidth}
		\centering
			\includegraphics[width=1.0\linewidth]{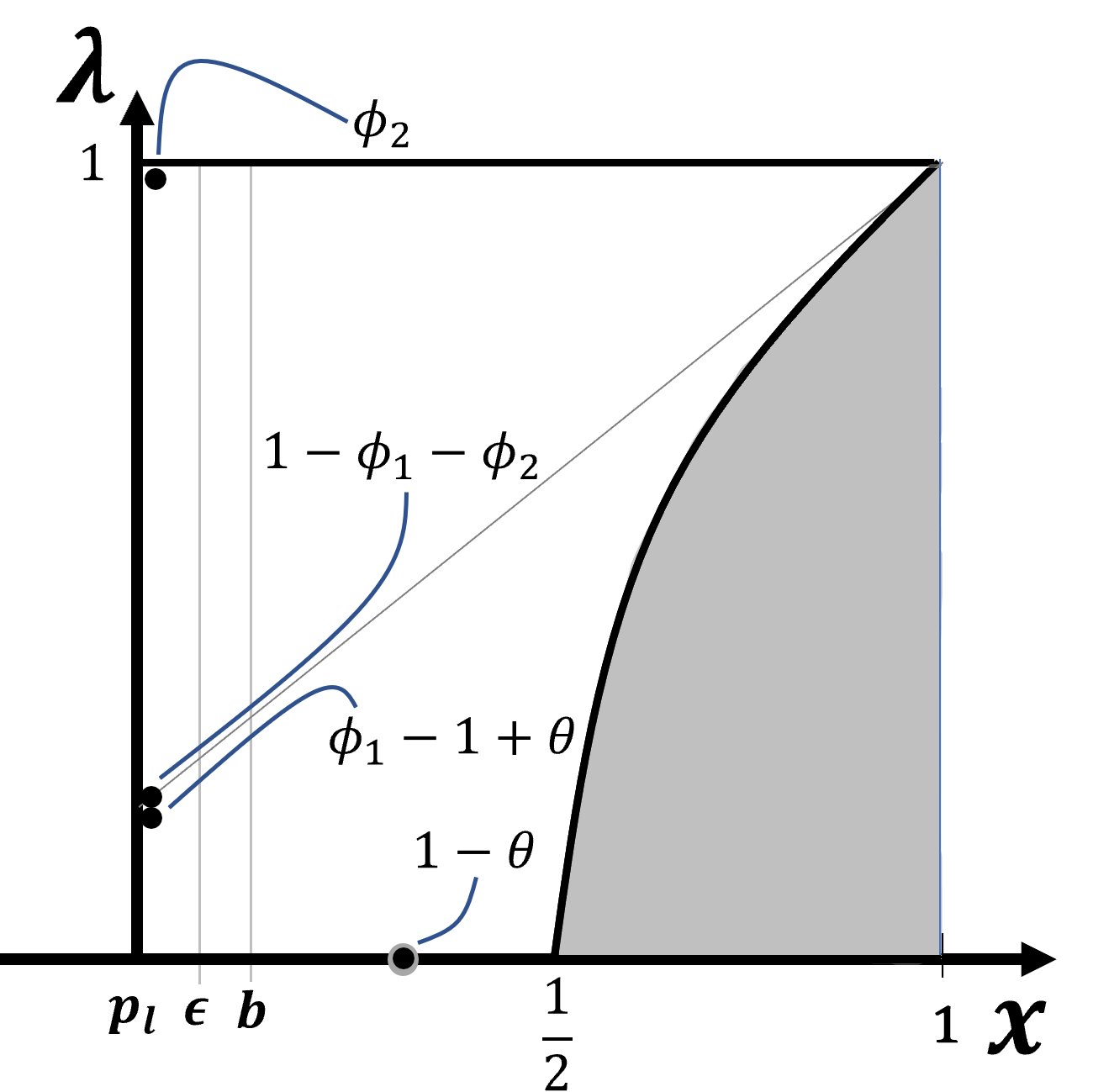}
		\caption{{\footnotesize $\,\,\phi_1\geqslant 1-\theta\,\,$ }}
		\label{fig:fig_CBIsoln_plnoConsecFailsPhi2ltTheta_Phi1gt1minTheta}
	\end{subfigure}
	\begin{subfigure}[]{0.3\linewidth}
		\centering
			\includegraphics[width=1.0\linewidth]{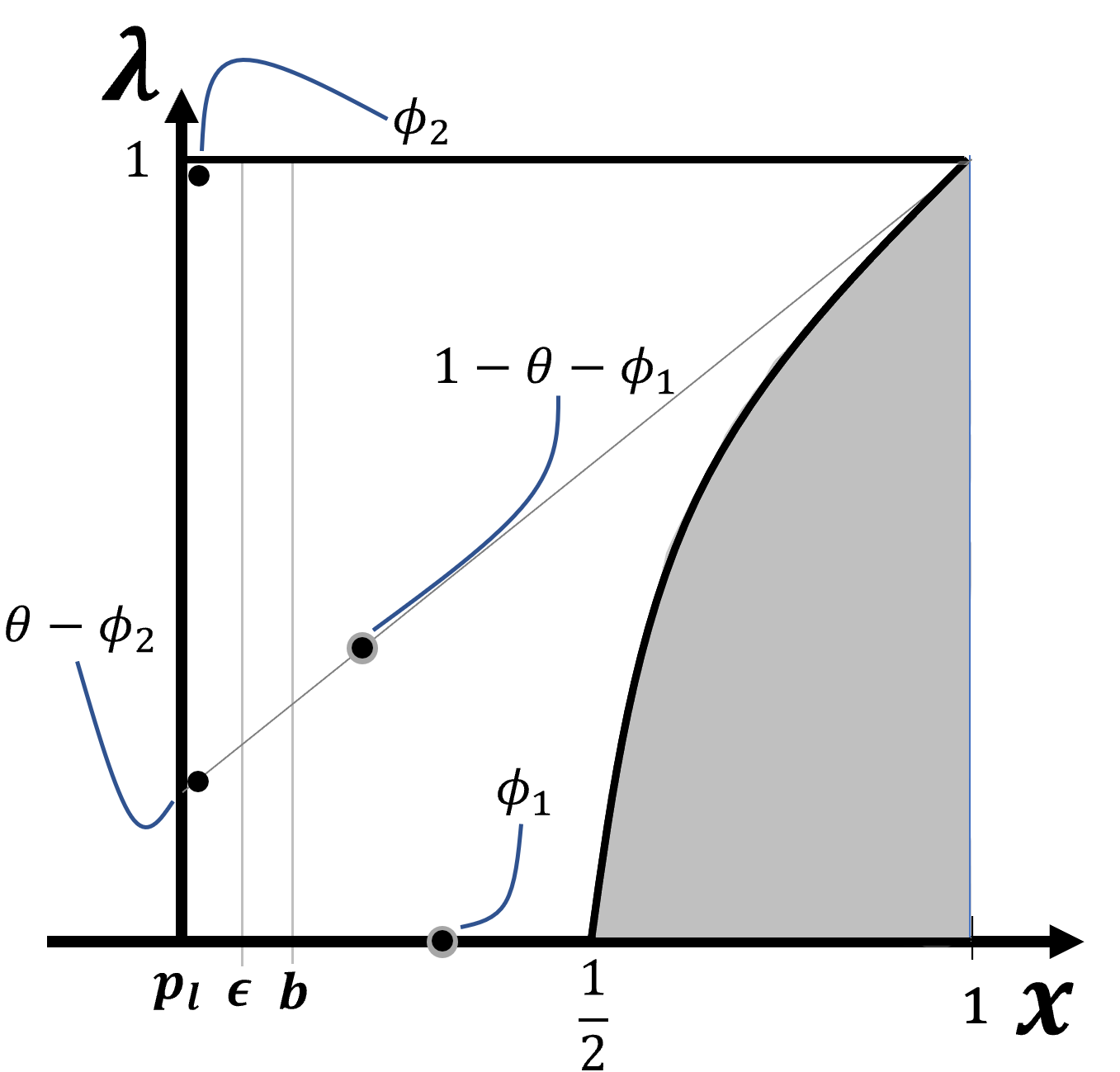}
		\caption{{\footnotesize $\,\,\phi_1\leqslant 1-\theta\,\,$ and $\,\,\phi_2\leqslant\theta$ }}
		\label{fig:fig_CBIsoln_plnoConsecFailsPhi2ltTheta_Phi1lt1minTheta}
	\end{subfigure}
	\caption[CBI priors with $p_l$ lower bound]{{\footnotesize Some examples of worst-case priors that give nonzero infima for the optimisation in Theorem \ref{theorem_CBI_withFailures_and_pl}; please see Figs.~\ref{fig_CBIsoln_withFails_Phi2glt1minustheta}, \ref{fig_CBIsoln_withnoconsFails_Phi1glt1minustheta}, \ref{fig_CBIsoln_Fails_zeroconf} of the supplementary material for all of the priors. When the software executes with {\bf some isolated and consecutive failures} (i.e. $r>0$), the worst-case priors can look like those in subfig.s \ref{fig:fig_CBIsoln_plgeneralFailsPhi2gt1minTheta} and \ref{fig:fig_CBIsoln_plgeneralFailsPhi2lt1minTheta_Phi1ltTheta}. And, if the executions contain some {\bf isolated failures, but no consecutive failures} (i.e. $r=0$), worst-case priors can look like subfig.s \ref{fig:fig_CBIsoln_plnoConsecFailsPhi2ltTheta_Phi1gt1minTheta} and \ref{fig:fig_CBIsoln_plnoConsecFailsPhi2ltTheta_Phi1lt1minTheta} instead. The exact locations of the support of these distributions (i.e. the black dots) depend on the values of the exponents in the likelihood function, and whether the 1st execution is a failure or not. 
			\normalsize}}
	\label{fig_CBIsoln_pl}
\end{figure}

\begin{example}[AVs] 
	\label{exp_av}
	Consider an assessor's confidence in an AV being as safe as the average human driver\footnote{The exact statistic in the U.S. (2013) was $1.09e{\!-\!8}$, as used by \cite{kalra_driving_2016}. For simplicity we round this to $10^{-8}$ in the examples.} in terms of \textit{pfm} (so $b=10^{-8}$), after a fleet of AVs have driven millions of miles. Using PK parameter values from \cite{zhao_assessing_2020}, we have: the engineering goal, $\epsilon=10^{-10}$, is 2 orders of magnitude safer than the pfm for human drivers; and the lower bound on pfm is $p_l=10^{-15}$. Fig.~\ref{fig_av_with_failures_example3} shows confidence in $b$ under different values of $s$ (number of failures), $r$ (consecutive failures), $\phi_1$, and $\phi_2$.
\end{example}

\begin{figure}[h!]
	\centering
	\includegraphics[width=0.45\linewidth]{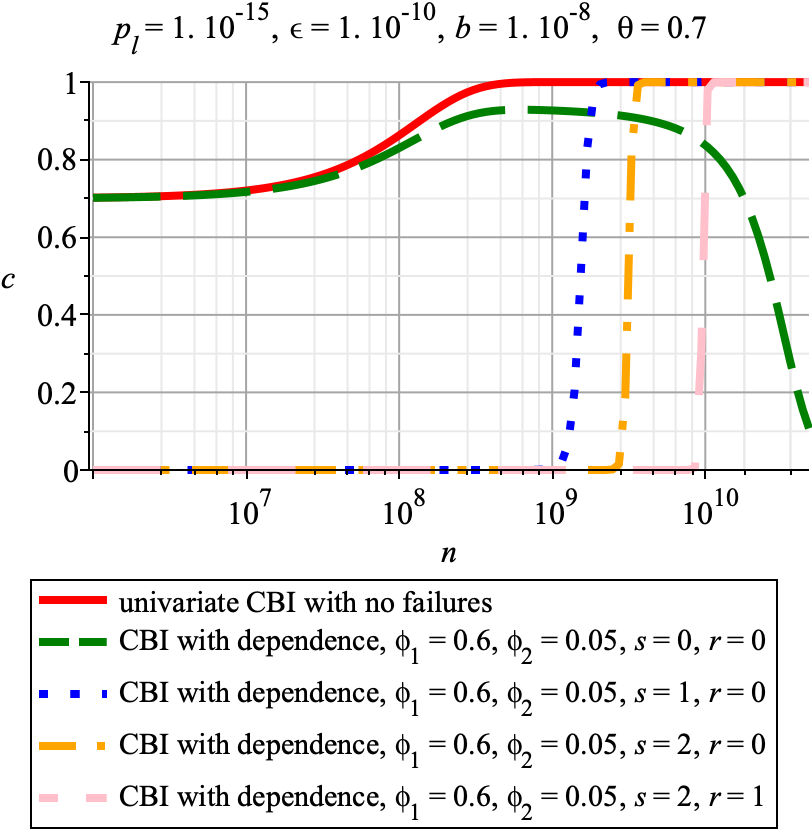}
	\caption{(Example \ref{exp_av}) $c\times 100\%$ posterior confidence in $[X\leqslant 10^{-8}]$ upon seeing $s$ failures ($r$ of them consecutive) in $n$ tests, from CBI models with different PKs.
		\xingyu{Empirically, there seems to be two equivalence relations: i. Univariate CBI = CBI with dependence but with phi1=phi2=0; ii. Univariate CBI = CBI with dependence when r=0 (including failure-free case); Not sure if the second one is true analytically... Maybe we should place a Remark somewhere regarding this two (effectively) equivalence relations; Double check this later... } \kizito{(i) failures allow confidence to grow to 1 asymptotically, because failures-free operation due to ``easy demands'' becomes unlikely. Otherwise failure-free evidence eventually harm the confidence; this is not possible in the univaraite case where failure-free evidence is always an indication of very reliable systems\xingyu{yes, I add in this point in insight 2 which is the first time we saw this result}; this observation is only part of the picture, more general results might be obtained by using more general posterior measures --measures that take into account both pfd and positive correlation (i.e. no failures). (ii) Seeing more failures requires more testing for confidence to grow. (iii) Seeing consecutive failures requires even more testing for confidence to grow.}}
	\label{fig_av_with_failures_example3}
\end{figure}

Three observations from Fig.~\ref{fig_av_with_failures_example3}: {\bf i)} like previous scenarios, the dashed curve shows that doubts in independent executions eventually undermine confidence in $b$ when no failures occur (\emph{cf.} Fig.s~\ref{fig_example1_baseline},\ref{fig_example2_nuclear}). However, the other curves in Fig.~\ref{fig_av_with_failures_example3}  (except the solid curve) show that failures allow confidence $c$ to eventually approach 1. This is explained in Insight \ref{insight_failures_allow_conf_1}; {\bf ii)} unsurprisingly, more failures  requires more testing for confidence in $b$ to grow (i.e., the dash-dot curve lies to the right of the dotted curve); and {\bf iii)} more consecutive failures requires even more testing (i.e., the space-dash curve lies to the right of the dash-dot curve). Section \ref{sec_senstivity_analysis}'s sensitivity analysis explores this further.

\begin{table*}[h!]
\small
\centering
\caption{A list of all prior distributions used in the sensitivity analysis of Section \ref{sec_senstivity_analysis}. See the supplementary material for any listed figures not included in the paper.} 
\label{tab_wcp_sa}
\resizebox{\linewidth}{!}{
	\begin{tabular}{|cccccccc|}
		\hline
		\hline
		\multicolumn{8}{|c|}{Sensitivity Analysis in Section \ref{sec_senstivity_analysis}} \\ \hline
		\multicolumn{2}{|c|}{Scenarios with various PKs} &
		\multicolumn{1}{c|}{red solid} &
		\multicolumn{1}{c|}{green dashed} &
		\multicolumn{1}{c|}{blue dotted} &
		\multicolumn{1}{c|}{orange dashdotted} &
		\multicolumn{1}{c|}{\begin{tabular}[c]{@{}c@{}} pink \\ spacedashed\end{tabular}} &
		\begin{tabular}[c]{@{}c@{}} black \\ spacedotted\end{tabular}\\ \hline
		\multicolumn{1}{|c|}{\multirow{3}{*}{\begin{tabular}[c]{@{}c@{}}Nuclear \\ reactor\\ protection \\ system\\ (Fig.~\ref{fig_sa_nuclear})\end{tabular}}} &
		\multicolumn{1}{c|}{\begin{tabular}[c]{@{}c@{}}Varying PK\ref{cons_engineering_goal} \\and $b$  (Fig.~\ref{fig_sa_nuclear_vary_PK2})\end{tabular}} &
		\multicolumn{1}{c|}{Fig.~\ref{fig_unicbi_wc_prior}} &
		\multicolumn{1}{c|}{Fig.~\ref{fig:fig_CBInoFails_maintxt_Phi1gtTheta}} &
		\multicolumn{1}{c|}{Fig.~\ref{fig:fig_CBInoFails_maintxt_Phi1ltThetaPhi2gt1minusTheta}} &
		\multicolumn{1}{c|}{Fig.~\ref{fig:fig_CBInoFails_maintxt_Phi1gtTheta}} &
		\multicolumn{1}{c|}{Fig.~\ref{fig:fig_CBInoFails_maintxt_Phi1gtTheta}} & n/a
		\\ \cline{2-8} 
		\multicolumn{1}{|c|}{} &
		\multicolumn{1}{c|}{\begin{tabular}[c]{@{}c@{}}Varying PK\ref{cons_negative_dependence}\\ (Fig.~\ref{fig_sa_nuclear_vary_PK3})\end{tabular}} &
		\multicolumn{1}{c|}{Fig.~\ref{fig_unicbi_wc_prior}} &
		\multicolumn{1}{c|}{Fig.~\ref{fig:fig_CBInoFails_maintxt_Phi1gtTheta}} &
		\multicolumn{1}{c|}{Fig.~\ref{fig:fig_CBInoFails_maintxt_Phi1ltThetaPhi2gt1minusTheta}} &
		\multicolumn{1}{c|}{Fig.~\ref{fig:fig_CBInoFails_maintxt_Phi1ltThetaPhi2gt1minusTheta}} &
		\multicolumn{1}{c|}{Fig.~\ref{fig:fig_CBInoFails_maintxt_Phi1gtTheta}} & n/a
		\\ \cline{2-8} 
		\multicolumn{1}{|c|}{} &
		\multicolumn{1}{c|}{\begin{tabular}[c]{@{}c@{}}Varying PK\ref{cons_positive_dependence}\\ (Fig.~\ref{fig_sa_nuclear_vary_PK4})\end{tabular}} &
		\multicolumn{1}{c|}{Fig.~\ref{fig_unicbi_wc_prior}} &
		\multicolumn{1}{c|}{Fig.~\ref{fig:fig_CBInoFails_maintxt_Phi1ltThetaPhi2gt1minusTheta}} &
		\multicolumn{1}{c|}{Fig.~\ref{fig:fig_CBInoFails_maintxt_Phi1ltThetaPhi2gt1minusTheta}} &
		\multicolumn{1}{c|}{Fig.s~\ref{fig:fig_CBInoFails_maintxt_Phi1ltThetaPhi2gt1minusTheta}/\ref{fig:fig_CBInoFails_maintxt_Phi1ltThetaPhi2lt1minusTheta}} &
		\multicolumn{1}{c|}{Fig.~\ref{fig:fig_CBInoFails_maintxt_Phi1ltThetaPhi2gt1minusTheta}} & n/a
		\\ \hline
		\multicolumn{1}{|c|}{\multirow{3}{*}{\begin{tabular}[c]{@{}c@{}}AVs with no \\consecutive \\ failures\\ (Fig.~\ref{fig_sa_av_r0})\end{tabular}}} &
		\multicolumn{1}{c|}{\begin{tabular}[c]{@{}c@{}}Varying PK\ref{cons_pl_lowerbound}, \\ PK\ref{cons_engineering_goal} and $b$ \\ (Fig.~\ref{fig_sa_av_r0_vary_pk12})\end{tabular}} &
		\multicolumn{1}{c|}{\begin{tabular}[c]{@{}c@{}} Fig.s~\ref{fig:fig_CBIsoln_withnoconsFails_Phi1gt1minustheta_1},\\ \ref{fig:fig_CBIsoln_withnoconsFails_Phi1gt1minustheta_2}, \ref{fig:fig_CBIsoln_withnoconsFails_Phi1gt1minustheta_3},\\ \ref{fig:fig_CBIsoln_withnoconsFails_Phi1gt1minustheta_4} as $n$ \\ increases \end{tabular} } &
		\multicolumn{1}{c|}{\begin{tabular}[c]{@{}c@{}} Fig.s~\ref{fig:fig_CBIsoln_withnoconsFails_Phi1gt1minustheta_1},\\ \ref{fig:fig_CBIsoln_withnoconsFails_Phi1gt1minustheta_2}, \ref{fig:fig_CBIsoln_withnoconsFails_Phi1gt1minustheta_3},\\ \ref{fig:fig_CBIsoln_withnoconsFails_Phi1gt1minustheta_4} as $n$ \\ increases \end{tabular}} &
		\multicolumn{1}{c|}{\begin{tabular}[c]{@{}c@{}} Fig.s~\ref{fig:fig_CBIsoln_withnoconsFails_Phi1gt1minustheta_1},\\ \ref{fig:fig_CBIsoln_withnoconsFails_Phi1gt1minustheta_2}, \ref{fig:fig_CBIsoln_withnoconsFails_Phi1gt1minustheta_3},\\ \ref{fig:fig_CBIsoln_withnoconsFails_Phi1gt1minustheta_4} as $n$ \\ increases \end{tabular}} &
		\multicolumn{1}{c|}{\begin{tabular}[c]{@{}c@{}} Fig.s~\ref{fig:fig_CBIsoln_withnoconsFails_Phi1gt1minustheta_1},\\ \ref{fig:fig_CBIsoln_withnoconsFails_Phi1gt1minustheta_2}, \ref{fig:fig_CBIsoln_withnoconsFails_Phi1gt1minustheta_3},\\ \ref{fig:fig_CBIsoln_withnoconsFails_Phi1gt1minustheta_4} as $n$ \\ increases \end{tabular}} &
		\multicolumn{1}{c|}{\begin{tabular}[c]{@{}c@{}} Fig.s~\ref{fig:fig_CBIsoln_withnoconsFails_Phi1gt1minustheta_1},\\ \ref{fig:fig_CBIsoln_withnoconsFails_Phi1gt1minustheta_2}, \ref{fig:fig_CBIsoln_withnoconsFails_Phi1gt1minustheta_3},\\ \ref{fig:fig_CBIsoln_withnoconsFails_Phi1gt1minustheta_4} as $n$ \\ increases \end{tabular}} & \begin{tabular}[c]{@{}c@{}} Fig.s~\ref{fig:fig_CBIsoln_withnoconsFails_Phi1gt1minustheta_1},\\ \ref{fig:fig_CBIsoln_withnoconsFails_Phi1gt1minustheta_2}, \ref{fig:fig_CBIsoln_withnoconsFails_Phi1gt1minustheta_3},\\ \ref{fig:fig_CBIsoln_withnoconsFails_Phi1gt1minustheta_4} as $n$ \\ increases \end{tabular}
		\\ \cline{2-8} 
		\multicolumn{1}{|c|}{} &
		\multicolumn{1}{c|}{\begin{tabular}[c]{@{}c@{}}Varying PK\ref{cons_negative_dependence}\\  (Fig.~\ref{fig_sa_av_r0_vary_PK3})\end{tabular}} &
		\multicolumn{1}{c|}{\begin{tabular}[c]{@{}c@{}} Fig.s~\ref{fig:fig_CBIsoln_withnoconsFails_Phi1gt1minustheta_1},\\ \ref{fig:fig_CBIsoln_withnoconsFails_Phi1gt1minustheta_2}, \ref{fig:fig_CBIsoln_withnoconsFails_Phi1gt1minustheta_3},\\ \ref{fig:fig_CBIsoln_withnoconsFails_Phi1gt1minustheta_4} as $n$ \\ increases \end{tabular}} &
		\multicolumn{1}{c|}{\begin{tabular}[c]{@{}c@{}} Fig.s~\ref{fig:fig_CBIsoln_withnoconsFails_Phi1gt1minustheta_1},\\ \ref{fig:fig_CBIsoln_withnoconsFails_Phi1gt1minustheta_2}, \ref{fig:fig_CBIsoln_withnoconsFails_Phi1gt1minustheta_3},\\ \ref{fig:fig_CBIsoln_withnoconsFails_Phi1gt1minustheta_4} as $n$ \\ increases \end{tabular}} &
		\multicolumn{1}{c|}{\begin{tabular}[c]{@{}c@{}} Fig.s~\ref{fig:fig_CBIsoln_withnoconsFails_Phi1lt1minustheta_1},\\ \ref{fig:fig_CBIsoln_withnoconsFails_Phi1lt1minustheta_2}, \ref{fig:fig_CBIsoln_withnoconsFails_Phi1lt1minustheta_3},\\ \ref{fig:fig_CBIsoln_withnoconsFails_Phi1lt1minustheta_4} as $n$ \\ increases* \end{tabular}} &
		\multicolumn{1}{c|}{\begin{tabular}[c]{@{}c@{}} Fig.s~\ref{fig:fig_CBIsoln_withnoconsFails_Phi1gt1minustheta_1},\\ \ref{fig:fig_CBIsoln_withnoconsFails_Phi1gt1minustheta_2}, \ref{fig:fig_CBIsoln_withnoconsFails_Phi1gt1minustheta_3},\\ \ref{fig:fig_CBIsoln_withnoconsFails_Phi1gt1minustheta_4} as $n$ \\ increases \end{tabular}} &
		\multicolumn{1}{c|}{\begin{tabular}[c]{@{}c@{}} Fig.s~\ref{fig:fig_CBIsoln_withnoconsFails_Phi1gt1minustheta_1},\\ \ref{fig:fig_CBIsoln_withnoconsFails_Phi1gt1minustheta_2}, \ref{fig:fig_CBIsoln_withnoconsFails_Phi1gt1minustheta_3},\\ \ref{fig:fig_CBIsoln_withnoconsFails_Phi1gt1minustheta_4} as $n$ \\ increases \end{tabular}} & \begin{tabular}[c]{@{}c@{}} Fig.s~\ref{fig:fig_CBIsoln_withnoconsFails_Phi1lt1minustheta_1},\\ \ref{fig:fig_CBIsoln_withnoconsFails_Phi1lt1minustheta_2}, \ref{fig:fig_CBIsoln_withnoconsFails_Phi1lt1minustheta_3},\\ \ref{fig:fig_CBIsoln_withnoconsFails_Phi1lt1minustheta_4} as $n$ \\ increases \end{tabular}
		\\ \cline{2-8} 
		\multicolumn{1}{|c|}{} &
		\multicolumn{1}{c|}{\begin{tabular}[c]{@{}c@{}}Varying PK\ref{cons_positive_dependence}\\ (Fig.~\ref{fig_sa_av_r0_vary_PK4})\end{tabular}} &
		\multicolumn{1}{c|}{\begin{tabular}[c]{@{}c@{}} Fig.s~\ref{fig:fig_CBIsoln_withnoconsFails_Phi1lt1minustheta_1},\\ \ref{fig:fig_CBIsoln_withnoconsFails_Phi1lt1minustheta_2}, \ref{fig:fig_CBIsoln_withnoconsFails_Phi1lt1minustheta_3},\\ \ref{fig:fig_CBIsoln_withnoconsFails_Phi1lt1minustheta_4} as $n$ \\ increases* \end{tabular}} &
		\multicolumn{1}{c|}{\begin{tabular}[c]{@{}c@{}} Fig.s~\ref{fig:fig_CBIsoln_withnoconsFails_Phi1lt1minustheta_1},\\ \ref{fig:fig_CBIsoln_withnoconsFails_Phi1lt1minustheta_2}, \ref{fig:fig_CBIsoln_withnoconsFails_Phi1lt1minustheta_3},\\ \ref{fig:fig_CBIsoln_withnoconsFails_Phi1lt1minustheta_4} as $n$ \\ increases* \end{tabular}} &
		\multicolumn{1}{c|}{\begin{tabular}[c]{@{}c@{}} Fig.s~\ref{fig:fig_CBIsoln_withnoconsFails_Phi1lt1minustheta_1},\\ \ref{fig:fig_CBIsoln_withnoconsFails_Phi1lt1minustheta_2}, \ref{fig:fig_CBIsoln_withnoconsFails_Phi1lt1minustheta_3},\\ \ref{fig:fig_CBIsoln_withnoconsFails_Phi1lt1minustheta_4} as $n$ \\ increases* \end{tabular}} &
		\multicolumn{1}{c|}{\begin{tabular}[c]{@{}c@{}} Fig.s~\ref{fig:fig_CBIsoln_withnoconsFails_Phi1lt1minustheta_1},\\ \ref{fig:fig_CBIsoln_withnoconsFails_Phi1lt1minustheta_2}, \ref{fig:fig_CBIsoln_withnoconsFails_Phi1lt1minustheta_3},\\ \ref{fig:fig_CBIsoln_withnoconsFails_Phi1lt1minustheta_4} as $n$ \\ increases* \end{tabular}} &
		\multicolumn{1}{c|}{\begin{tabular}[c]{@{}c@{}} Fig.s~\ref{fig:fig_CBIsoln_withnoconsFails_Phi1lt1minustheta_1},\\ \ref{fig:fig_CBIsoln_withnoconsFails_Phi1lt1minustheta_2}, \ref{fig:fig_CBIsoln_withnoconsFails_Phi1lt1minustheta_3},\\ \ref{fig:fig_CBIsoln_withnoconsFails_Phi1lt1minustheta_4} as $n$ \\ increases \end{tabular}} & \begin{tabular}[c]{@{}c@{}} Fig.s~\ref{fig:fig_CBIsoln_withnoconsFails_zeroconf_1},\\ \ref{fig:fig_CBIsoln_withnoconsFails_zeroconf_2}, \ref{fig:fig_CBIsoln_withnoconsFails_zeroconf_3},\\ \ref{fig:fig_CBIsoln_withnoconsFails_zeroconf_4} as $n$ \\ increases \end{tabular}
		\\ \hline
		\multicolumn{1}{|c|}{\multirow{3}{*}{\begin{tabular}[c]{@{}c@{}}AVs with \\consecutive \\ failures\\ (Fig.~\ref{fig_sa_av_rnot0})\end{tabular}}} &
		\multicolumn{1}{c|}{\begin{tabular}[c]{@{}c@{}}Varying PK\ref{cons_pl_lowerbound}, \\ PK\ref{cons_engineering_goal} and $b$ \\ (Fig.~\ref{fig_sa_av_rnot0_vary_pk12})\end{tabular}} &
		\multicolumn{1}{c|}{\begin{tabular}[c]{@{}c@{}} Fig.s~\ref{fig:fig_CBIsoln_withFails_Phi2lt1minustheta_1},\\ \ref{fig:fig_CBIsoln_withFails_Phi2lt1minustheta_2}, \ref{fig:fig_CBIsoln_withFails_Phi2lt1minustheta_3},\\ \ref{fig:fig_CBIsoln_withFails_Phi2lt1minustheta_4} as $n$ \\ increases \end{tabular}} &
		\multicolumn{1}{c|}{\begin{tabular}[c]{@{}c@{}} Fig.s~\ref{fig:fig_CBIsoln_withFails_Phi2lt1minustheta_1},\\ \ref{fig:fig_CBIsoln_withFails_Phi2lt1minustheta_2}, \ref{fig:fig_CBIsoln_withFails_Phi2lt1minustheta_3},\\ \ref{fig:fig_CBIsoln_withFails_Phi2lt1minustheta_4} as $n$ \\ increases \end{tabular}} &
		\multicolumn{1}{c|}{\begin{tabular}[c]{@{}c@{}} Fig.s~\ref{fig:fig_CBIsoln_withFails_Phi2lt1minustheta_1},\\ \ref{fig:fig_CBIsoln_withFails_Phi2lt1minustheta_2}, \ref{fig:fig_CBIsoln_withFails_Phi2lt1minustheta_3},\\ \ref{fig:fig_CBIsoln_withFails_Phi2lt1minustheta_4} as $n$ \\ increases \end{tabular}} &
		\multicolumn{1}{c|}{\begin{tabular}[c]{@{}c@{}} Fig.s~\ref{fig:fig_CBIsoln_withFails_Phi2lt1minustheta_1},\\ \ref{fig:fig_CBIsoln_withFails_Phi2lt1minustheta_2}, \ref{fig:fig_CBIsoln_withFails_Phi2lt1minustheta_3},\\ \ref{fig:fig_CBIsoln_withFails_Phi2lt1minustheta_4} as $n$ \\ increases \end{tabular}} &
		\multicolumn{1}{c|}{\begin{tabular}[c]{@{}c@{}} Fig.s~\ref{fig:fig_CBIsoln_withFails_Phi2lt1minustheta_1},\\ \ref{fig:fig_CBIsoln_withFails_Phi2lt1minustheta_2}, \ref{fig:fig_CBIsoln_withFails_Phi2lt1minustheta_3},\\ \ref{fig:fig_CBIsoln_withFails_Phi2lt1minustheta_4} as $n$ \\ increases \end{tabular}} &\begin{tabular}[c]{@{}c@{}} Fig.s~\ref{fig:fig_CBIsoln_withFails_Phi2lt1minustheta_1},\\ \ref{fig:fig_CBIsoln_withFails_Phi2lt1minustheta_2}, \ref{fig:fig_CBIsoln_withFails_Phi2lt1minustheta_3},\\ \ref{fig:fig_CBIsoln_withFails_Phi2lt1minustheta_4} as $n$ \\ increases \end{tabular}
		\\ \cline{2-8} 
		\multicolumn{1}{|c|}{} &
		\multicolumn{1}{c|}{\begin{tabular}[c]{@{}c@{}}Varying PK\ref{cons_negative_dependence}   \\ (Fig.~\ref{fig_sa_av_rnot0_vary_PK3})\end{tabular}} &
		\multicolumn{1}{c|}{\begin{tabular}[c]{@{}c@{}} Fig.s~\ref{fig:fig_CBIsoln_withFails_Phi2lt1minustheta_1},\\ \ref{fig:fig_CBIsoln_withFails_Phi2lt1minustheta_2}, \ref{fig:fig_CBIsoln_withFails_Phi2lt1minustheta_3},\\ \ref{fig:fig_CBIsoln_withFails_Phi2lt1minustheta_4} as $n$ \\ increases \end{tabular}} &
		\multicolumn{1}{c|}{\begin{tabular}[c]{@{}c@{}} Fig.s~\ref{fig:fig_CBIsoln_withFails_zeroconf_1},\\ \ref{fig:fig_CBIsoln_withFails_zeroconf_2}, \ref{fig:fig_CBIsoln_withFails_zeroconf_3},\\ \ref{fig:fig_CBIsoln_withFails_zeroconf_4} as $n$ \\ increases \end{tabular}} &
		\multicolumn{1}{c|}{\begin{tabular}[c]{@{}c@{}} Fig.s~\ref{fig:fig_CBIsoln_withFails_Phi2lt1minustheta_1},\\ \ref{fig:fig_CBIsoln_withFails_Phi2lt1minustheta_2}, \ref{fig:fig_CBIsoln_withFails_Phi2lt1minustheta_3},\\ \ref{fig:fig_CBIsoln_withFails_Phi2lt1minustheta_4} as $n$ \\ increases \end{tabular}} &
		\multicolumn{1}{c|}{\begin{tabular}[c]{@{}c@{}} Fig.s~\ref{fig:fig_CBIsoln_withFails_Phi2lt1minustheta_1},\\ \ref{fig:fig_CBIsoln_withFails_Phi2lt1minustheta_2}, \ref{fig:fig_CBIsoln_withFails_Phi2lt1minustheta_3},\\ \ref{fig:fig_CBIsoln_withFails_Phi2lt1minustheta_4} as $n$ \\ increases \end{tabular}} &
		\multicolumn{1}{c|}{\begin{tabular}[c]{@{}c@{}} Fig.s~\ref{fig:fig_CBIsoln_withFails_zeroconf_1},\\ \ref{fig:fig_CBIsoln_withFails_zeroconf_2}, \ref{fig:fig_CBIsoln_withFails_zeroconf_3},\\ \ref{fig:fig_CBIsoln_withFails_zeroconf_4} as $n$ \\ increases \end{tabular}} & \begin{tabular}[c]{@{}c@{}} Fig.s~\ref{fig:fig_CBIsoln_withFails_Phi2lt1minustheta_1},\\ \ref{fig:fig_CBIsoln_withFails_Phi2lt1minustheta_2}, \ref{fig:fig_CBIsoln_withFails_Phi2lt1minustheta_3},\\ \ref{fig:fig_CBIsoln_withFails_Phi2lt1minustheta_4} as $n$ \\ increases \end{tabular}
		\\ \cline{2-8} 
		\multicolumn{1}{|c|}{} &
		\multicolumn{1}{c|}{\begin{tabular}[c]{@{}c@{}}Varying PK\ref{cons_positive_dependence}   \\ (Fig.~\ref{fig_sa_av_rnot0_vary_PK4})\end{tabular}} &
		\multicolumn{1}{c|}{\begin{tabular}[c]{@{}c@{}} Fig.s~\ref{fig:fig_CBIsoln_withFails_Phi2lt1minustheta_1},\\ \ref{fig:fig_CBIsoln_withFails_Phi2lt1minustheta_2}, \ref{fig:fig_CBIsoln_withFails_Phi2lt1minustheta_3},\\ \ref{fig:fig_CBIsoln_withFails_Phi2lt1minustheta_4} as $n$ \\ increases \end{tabular}} &
		\multicolumn{1}{c|}{\begin{tabular}[c]{@{}c@{}} Fig.s~\ref{fig:fig_CBIsoln_withFails_Phi2gt1minustheta_1},\\ \ref{fig:fig_CBIsoln_withFails_Phi2gt1minustheta_2}, \ref{fig:fig_CBIsoln_withFails_Phi2gt1minustheta_3},\\ \ref{fig:fig_CBIsoln_withFails_Phi2gt1minustheta_4} as $n$ \\ increases \end{tabular}} &
		\multicolumn{1}{c|}{\begin{tabular}[c]{@{}c@{}} Fig.s~\ref{fig:fig_CBIsoln_withFails_Phi2gt1minustheta_1},\\ \ref{fig:fig_CBIsoln_withFails_Phi2gt1minustheta_2}, \ref{fig:fig_CBIsoln_withFails_Phi2gt1minustheta_3},\\ \ref{fig:fig_CBIsoln_withFails_Phi2gt1minustheta_4} as $n$ \\ increases \end{tabular}} &
		\multicolumn{1}{c|}{\begin{tabular}[c]{@{}c@{}} Fig.s~\ref{fig:fig_CBIsoln_withFails_Phi2lt1minustheta_1},\\ \ref{fig:fig_CBIsoln_withFails_Phi2lt1minustheta_2}, \ref{fig:fig_CBIsoln_withFails_Phi2lt1minustheta_3},\\ \ref{fig:fig_CBIsoln_withFails_Phi2lt1minustheta_4} as $n$ \\ increases \end{tabular}} &
		\multicolumn{1}{c|}{\begin{tabular}[c]{@{}c@{}} Fig.s~\ref{fig:fig_CBIsoln_withFails_Phi2gt1minustheta_1},\\ \ref{fig:fig_CBIsoln_withFails_Phi2gt1minustheta_2}, \ref{fig:fig_CBIsoln_withFails_Phi2gt1minustheta_3},\\ \ref{fig:fig_CBIsoln_withFails_Phi2gt1minustheta_4} as $n$ \\ increases \end{tabular}} &\begin{tabular}[c]{@{}c@{}} Fig.s~\ref{fig:fig_CBIsoln_withFails_Phi2gt1minustheta_1},\\ \ref{fig:fig_CBIsoln_withFails_Phi2gt1minustheta_2}, \ref{fig:fig_CBIsoln_withFails_Phi2gt1minustheta_3},\\ \ref{fig:fig_CBIsoln_withFails_Phi2gt1minustheta_4} as $n$ \\ increases \end{tabular}
		\\ \hline \hline
	\end{tabular}}
	\\
	{\footnotesize\raggedright *The set of worst-case priors can be replaced by Fig.s~\ref{fig:fig_CBIsoln_withnoconsFails_Phi1gt1minustheta_1},  \ref{fig:fig_CBIsoln_withnoconsFails_Phi1gt1minustheta_2}, \ref{fig:fig_CBIsoln_withnoconsFails_Phi1gt1minustheta_3},  \ref{fig:fig_CBIsoln_withnoconsFails_Phi1gt1minustheta_4} as $n$ increases, since the parameters chosen cover these cases. \par}
\end{table*}

\begin{insight}[failures can allow confidence in $b$ to grow to 1]
	\label{insight_failures_allow_conf_1}
	Consider the following two possibilities when failures occur: {\bf i)} \emph{no consecutive failures} (so $s>r=0$), \emph{and prior evidence weakly supports positively correlated executions} ($\phi_2\leqslant\theta$). Then initially, execution outcomes are evidence of negative correlation (possibly from a system with \emph{pfm} larger than $b$). However, as the number of successful executions increases (with no more failures), it becomes less likely that the executions are negatively (or positively) correlated; otherwise, more failures should have been observed. Instead, it's more likely that the successes are occurring because the \emph{pfm} is smaller than $b$; {\bf ii)} \emph{consecutive failures} (so $s>r>0$), \emph{and prior evidence weakly supports negatively correlated executions} ($\phi_1\leqslant\theta$). Again, initially,  correlated executions (possibly from a system with \emph{pfm} larger than $b$) are most likely. And like the previous case, as the successful executions increase (with no more failures), it becomes less likely that the executions \emph{are} correlated, and more likely that the successes are due to the \emph{pfm} being smaller than $b$.
\end{insight}


\begin{figure*}[h]
	\centering
	\begin{subfigure}[b]{0.33\textwidth}
		\centering
		\includegraphics[width=\textwidth]{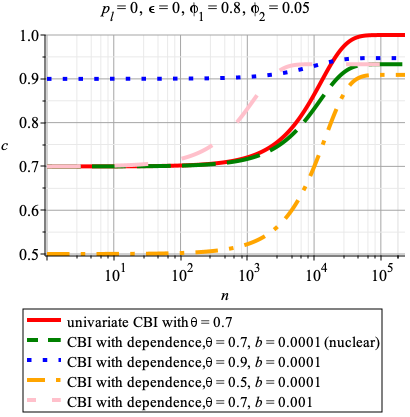}
		\caption{Varying PK\ref{cons_engineering_goal} and $b$}
		\label{fig_sa_nuclear_vary_PK2}
	\end{subfigure}
	\hfill
	\begin{subfigure}[b]{0.33\textwidth}
		\centering
		\includegraphics[width=\textwidth]{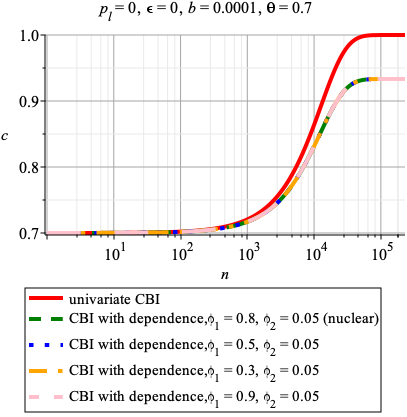}
		\caption{Varying PK\ref{cons_negative_dependence}}
		\label{fig_sa_nuclear_vary_PK3}
	\end{subfigure}
	\hfill
	\begin{subfigure}[b]{0.33\textwidth}
		\centering
		\includegraphics[width=\textwidth]{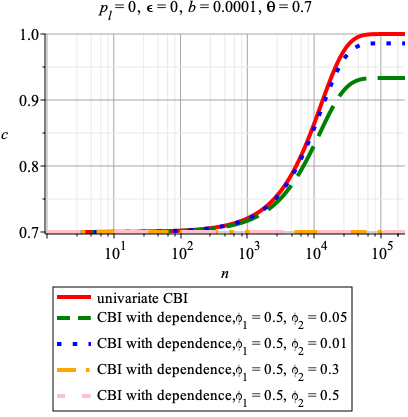}
		\caption{Varying PK\ref{cons_positive_dependence}}
		\label{fig_sa_nuclear_vary_PK4}
	\end{subfigure}
	\caption{Sensitivity analysis varying PKs for the nuclear reactor safety protection system scenario.}
	\label{fig_sa_nuclear}
\end{figure*}

\begin{figure*}[h]
	\centering
	\begin{subfigure}[b]{0.33\textwidth}
		\centering
		\includegraphics[width=\textwidth]{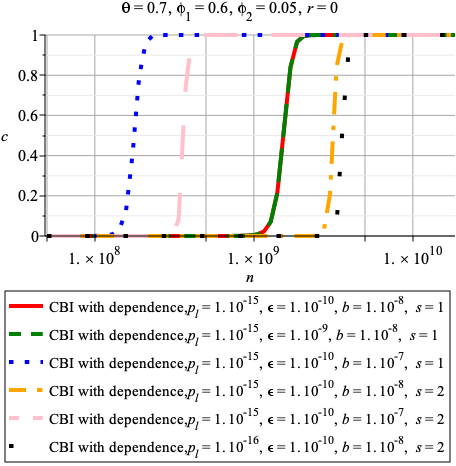}
		\caption{Varying PK\ref{cons_pl_lowerbound}, PK\ref{cons_engineering_goal} and $b$}
		\label{fig_sa_av_r0_vary_pk12}
	\end{subfigure}
	\hfill
	\begin{subfigure}[b]{0.33\textwidth}
		\centering
		\includegraphics[width=\textwidth]{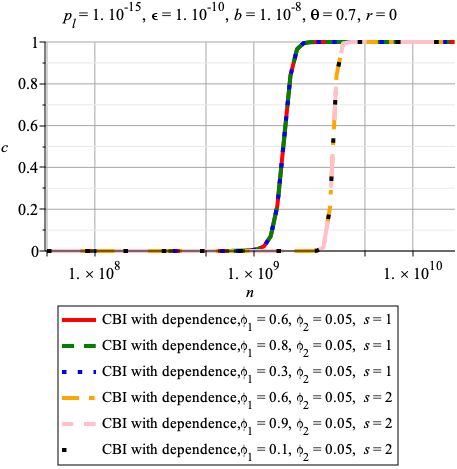}
		\caption{Varying PK\ref{cons_negative_dependence}}
		\label{fig_sa_av_r0_vary_PK3}
	\end{subfigure}
	\hfill
	\begin{subfigure}[b]{0.33\textwidth}
		\centering
		\includegraphics[width=\textwidth]{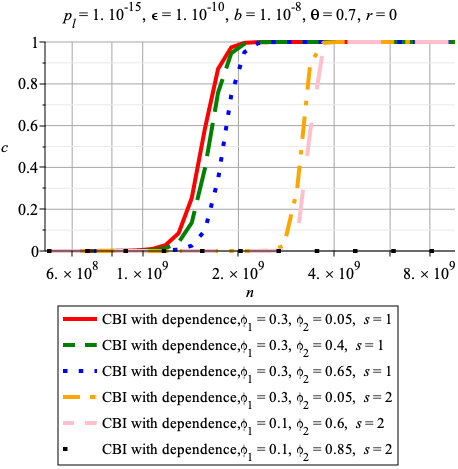}
		\caption{Varying PK\ref{cons_positive_dependence}}
		\label{fig_sa_av_r0_vary_PK4}
	\end{subfigure}
	\caption{Sensitivity analysis varying PKs for the AV-safety scenario with no consecutive failures.}
	\label{fig_sa_av_r0}
\end{figure*}

\begin{figure*}[h]
	\centering
	\begin{subfigure}[b]{0.33\textwidth}
		\centering
		\includegraphics[width=\textwidth]{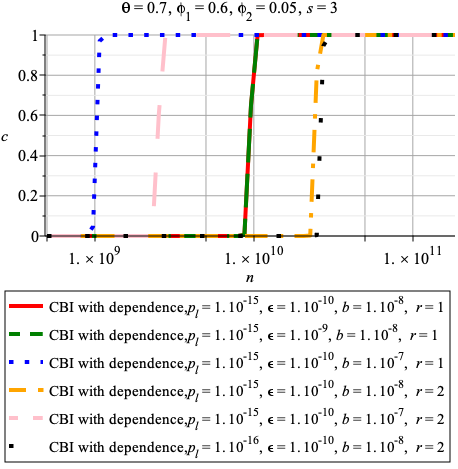}
		\caption{Varying PK\ref{cons_pl_lowerbound}, PK\ref{cons_engineering_goal} and $b$}
		\label{fig_sa_av_rnot0_vary_pk12}
	\end{subfigure}
	\hfill
	\begin{subfigure}[b]{0.33\textwidth}
		\centering
		\includegraphics[width=\textwidth]{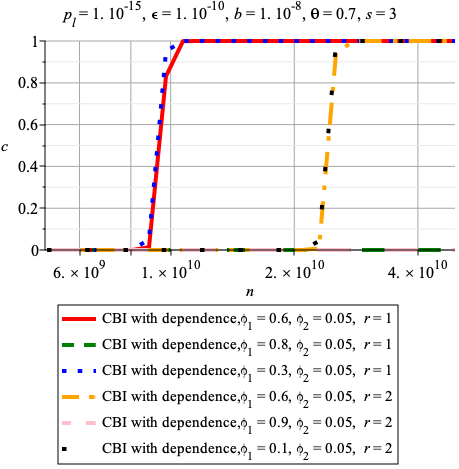}
		\caption{Varying PK\ref{cons_negative_dependence}}
		\label{fig_sa_av_rnot0_vary_PK3}
	\end{subfigure}
	\hfill
	\begin{subfigure}[b]{0.33\textwidth}
		\centering
		\includegraphics[width=\textwidth]{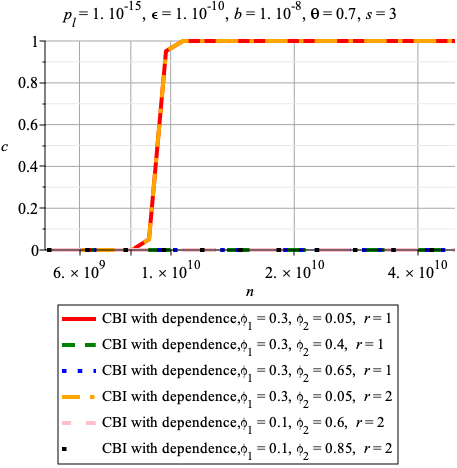}
		\caption{Varying PK\ref{cons_positive_dependence}}
		\label{fig_sa_av_rnot0_vary_PK4}
	\end{subfigure}
	\caption{Sensitivity analysis varying PKs for the AV-safety scenario with some consecutive failures (note, when $\phi_2\geqslant 1-\theta$, the plots are effectively 0, but not 0).}
\label{fig_sa_av_rnot0}
\end{figure*}

\section{The Sensitivity of Confidence Bounds to Changes in Prior Knowledge and Evidence}
\label{sec_senstivity_analysis}
For the assessor that is uncertain about PK values, this section illustrates how to check the sensitivity/robustness of confidence in $b$ to changes in PK values. The analyses also gives insight into how confidence responds to a ``strengthening'' of prior reliability evidence. We systematically vary PK parameters, the bound $b$, and the execution outcomes. The prior distributions used in the numerical analyses are summarised in Table \ref{tab_wcp_sa}.

\subsection{The Nuclear Safety Protection System Scenario}
Fig.~\ref{fig_sa_nuclear} provides 3 sub-figures highlighting the effects of changes in PK parameters (except PK\ref{cons_pl_lowerbound}).

In Fig.~\ref{fig_sa_nuclear_vary_PK2}, ``CBI with dependence'' curves are asymptotically more conservative than the univariate CBI solid curve. However, because the system could be fault-free, the curves show that confidence in $b$ always increases with increasing failure-free operational evidence. Also, the smaller $\phi_2$ becomes, or the bigger $b$ or $\theta$ are, the greater confidence in $b$ can become.

Fig.~\ref{fig_sa_nuclear_vary_PK3} illustrates how changes in $\phi_1$ have no apparent impact on confidence in $b$. However, Fig.~\ref{fig_sa_nuclear_vary_PK4} shows that the smaller $\phi_2$ becomes, the closer the ``CBI with dependence'' curve gets to the univariate CBI curve. While, for $\phi_2\geqslant 1-\theta$, the confidence in $b$ is the constant horizontal line $c=\frac{\theta}{\theta+(1-b)(1-\theta)}$ (see Fig.~\ref{fig:fig_CBInoFails_maintxt_Phi1ltThetaPhi2gt1minusTheta} with $p_l=\epsilon=0$).


\xingyu{all lines are asymptotic to the analytical value mentioned in section 3.2; the two straight horizontal lines in the right say they are not functions of n, which can be seen from analytical results}

\subsection{The AV Scenario}
Recall that, unlike the baseline and nuclear protection system scenarios, failures occur (albeit rarely\footnote{Due to the severe negative impact failures can have in SCSs, we only consider operational campaigns with no more than a few failures.}) during sufficiently long road testing campaigns (so $s>0$). The confidence in bound $b$ from Theorem \ref{theorem_CBI_withFailures_and_pl} is dependent on whether some of these failures are consecutive ($r>0$) or not ($r=0$). We conduct two sets of analyses\footnote{Note, Fig.~\ref{fig_av_with_failures_example3} shows the result of $r$ becoming nonzero for fixed $s$.} along these possibilities.

\subsubsection{With No Consecutive Failures $(s>r=0)$}
In Fig.~\ref{fig_sa_av_r0_vary_pk12} confidence in $b$ changes in response to increases in $s$ and $b$. The increase in $\epsilon$ has no noticeable impact -- the solid and dashed curves overlap. However, when $b$ is increased, the required number of executions $n$ to support a given confidence level reduces by a similar order of magnitude. In contrast, an additional failure increases $n$ significantly -- so the dotted and solid curves lie to the left of the spacedashed and  dashdotted curves, respectively. This is consistent with the findings of \cite{littlewood_conservative_1997}.


Changes in $\phi_1$ have no noticeable effect on confidence in $b$ (see Fig.~\ref{fig_sa_av_r0_vary_PK3}). Perhaps because there is very little operational evidence to support negative correlations -- i.e. only few instances of ``switching'' between failure and success. 


On the other hand, changes in $\phi_2$ have a clear impact, as shown in Fig.~\ref{fig_sa_av_r0_vary_PK4}. An increase in $\phi_2$ requires an increase in $n$. Moreover, when $\phi_2\geqslant\theta$, confidence in $b$ becomes 0 \emph{for all} $n$ (see the prior distributions for $r=0$ in Fig.~\ref{fig_CBIsoln_Fails_zeroconf}).

\xingyu{The right one covers all those 3 sets of priors in KS' blackboard picture. It seems, overall, not quite sensitive to phi1 and phi2, unless phi2 is bigger than theta...(and this actually holds for both previous cases...);\kizito{for the left: green and red are different for large $n$. Black line requires more failure-free test for confidence to grow, because conservatively we need to claim something more stringent (i.e. (pl,pl)) in this case; For the middle: the biggest impact comes from the failures--each extra failure requires significantly more testing (we should cross reference to earlier publications). the curves are not identical but the difference is negligible. to check when s is big, change phi1 will separate the curves..; For the right one, increasing phi2 is undesirable; when phi2 is bigger than theta,it becomes 0 posterior confidence no matter how many tests}}

\subsubsection{With Some Consecutive Failures $(s>r>0)$}
Despite consecutive failures, Figs.~\ref{fig_sa_av_rnot0_vary_pk12} and \ref{fig_sa_av_r0_vary_pk12} broadly give the same insights.

When prior confidence in negative correlations is strong (i.e. $\phi_1\geqslant\theta$), Fig.~\ref{fig_sa_av_rnot0_vary_PK3} shows confidence in $b$ is 0 \emph{for all} $n$, just like the condition on $\phi_2$ in the $r=0$ scenario (see the left column of prior distributions in Fig.~\ref{fig_CBIsoln_Fails_zeroconf}). While a large $\phi_2$ (i.e. $\phi_2\geqslant 1- \theta$) gives practically 0 confidence in $b$ in Fig.~\ref{fig_sa_av_rnot0_vary_PK4}.\kizito{I can't see how increasing $r$ in the final plot requires more $n$} 

\xingyu{for the left, all curves say the same lesson as before; in summary, bigger phi1, bigger phi2 and bigger r require more tests.}

We now summarise the results of the sensitivity analysis:
\begin{insight}[results of sensitivity analysis]
Confidence in $b$ from ``CBI with dependence'' is insensitive to changes in $\phi_1$, although confidence is 0 for $\phi_1\geqslant\theta, r>0$. Confidence in $b$ is sensitive to $\phi_2$. Both the number of failures $s$ and the number of consecutive failures $r$ have a significant effect on confidence in $b$. When there are no failed executions $p_l$ is irrelevant, with some effect when failures occur. The impact of the  engineering goal, $\epsilon$, is consistent with previous CBI models, e.g. \cite{zhao_assessing_2020} -- smaller $\epsilon$ gives greater confidence in $b$ when no failures are observed, but has no significant impact when there are failures. 
\end{insight}

\section{Discussion}
\label{sec_discussion}
\begin{figure*}[h!]
\centering
\includegraphics[width=0.9\linewidth]{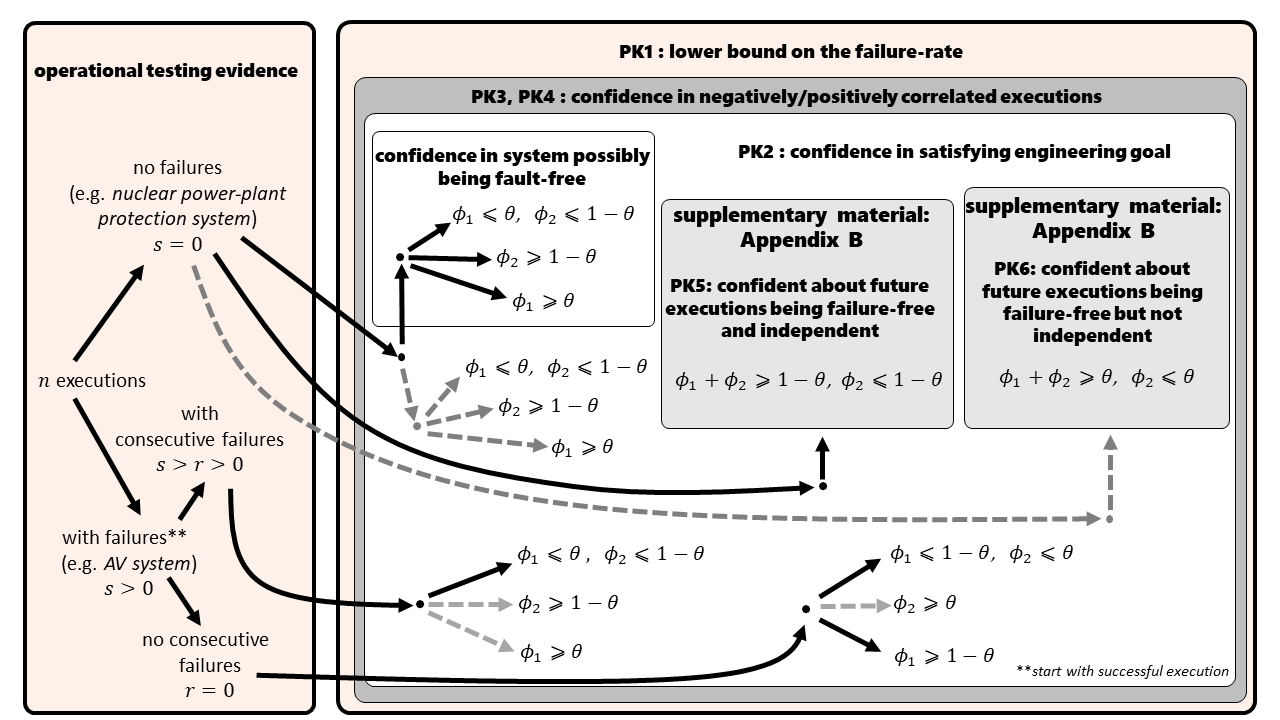}
\caption{{\footnotesize An updated overview of the various assessment scenarios analysed in this paper (\emph{cf.} Fig.~\ref{fig_theBigPicture_intro}), indicating the possible testing outcomes and prior confidence an assessor could have. The dashed paths are scenarios where either testing is futile in supporting reliability claims, or the scenarios are of little practical interest. \normalsize}}
\label{fig_theBigPicture_discussion}
\end{figure*}


\subsection{Incorporating Doubts about Model Assumptions into Reliability Assessments}
\label{sec_includingDoubtsinDepClaims}
In Bayesian assessments, one should always question the properties of the statistical model (such as independent executions), and whether these properties are appropriate in a given real-world context. Our use of CBI illustrates a general, formal method for incorporating such doubts -- about the suitability of any statistical model properties -- directly into the assessment. Since the results of CBI are guaranteed to be conservative (see Insight \ref{insight_gist_cbi}), this is a conservative version of a Bayesian approach first proposed by \cite{Draper_1995}. Draper suggests that if one is uncertain about the suitability of model properties, one should perform the inference with an ``expanded'' model that weakens the properties in question and has the original model as a special case. In this sense, our choice of the Klotz model is not arbitrary: it is the simplest model we know that exhibits dependent, stationary Bernoulli trials, and that has ``independent executions'' as a special case. This approach is incremental; if one is doubtful of the Klotz model, a model expansion of the Klotz model can be used for assessment.

Using this approach in \cite{SalakoZhao_TSE_2023}, we illustrated how to formally check the impact of one's doubts (about i.i.d. execution outcomes) on reliability claims, where such claims are based on Bayesian inference with operational testing data. The results suggested that, for on-demand systems, using the i.i.d. assumption in assessments could result in extremely optimistic claims, but not always. By weakening the prior knowledge an assessor must justify before commencing operational testing, the current paper extends this previous work to a much wider class of scenarios. For example, there are assessment scenarios for which the i.i.d. assumption supports claims ``close to being'' conservative. ``How close'' will depend on the strength of reliability evidence available (e.g., Fig.~\ref{fig_av_with_failures_example3} shows how accumulating operational evidence can make i.i.d.-based claims less optimistic).             


\subsection{Which Prior Beliefs give Conservative Confidence Bounds?}
Only certain prior beliefs about the \emph{pfe}, $X$, and the dependence between executions, $\Lambda$, will ensure an assessor's posterior confidence in a \emph{pfe} upper bound $b$ is conservative. The CBI solutions of section \ref{sec_conf_bound_in_reliability} make clear what these beliefs are -- i.e., these beliefs are encoded in the prior distributions that solve Theorems \ref{theorem_thm1_baseline} and \ref{theorem_CBI_withFailures_and_pl}
. Specifically, the beliefs are encoded as those $(x,\lambda)$ locations in region $\mathcal R$ that each of these distributions assign nonzero probability to (e.g., see Fig.s~\ref{fig_CBInoFails_maintxt}, \ref{fig_CBIsoln_pl}, \ref{fig_beliefsAboutIndependence}). Four main factors determine such beliefs: {\bf i)} the execution outcomes during operational testing (i.e., the successes/failures); {\bf ii)} prior knowledge, e.g. evidence strongly suggests the executions are positively correlated (i.e., large $\phi_2$), not negatively correlated (i.e., small $\phi_1$); {\bf iii)} which beliefs about $(X, \Lambda)$ are \emph{least likely} to have produced the testing outcomes, if the \emph{pfe} is less than $b$; and {\bf iv)} which beliefs are \emph{most likely} to have produced the outcomes, if the \emph{pfe} is larger than $b$. Here, ``least likely'' and ``most likely'' are determined by the Klotz likelihood. 

In addition to ensuring an assessor's confidence is conservative, here are two more reasons for why such beliefs are important. Firstly, they are consistent with the available evidence, since the beliefs are encoded in prior distributions that are (the limits of sequences of prior distributions that are) consistent with the evidence. So, the assessor cannot ``rule out'' these beliefs without extra evidence, and the consequences of these beliefs should be taken seriously. Secondly, when reliability evidence is ``weak'', these beliefs can make operational testing futile: the more one observes successful executions, the more doubtful one becomes about the \emph{pfe}. Assessors should have enough evidence before embarking on operational testing; we elaborate on these points below.

For example, consider when all of the executions are successful (e.g. Example~\ref{exp_baseline_scenario} of section~\ref{sec_conf_bound_in_reliability}). Superficially, this suggests the executions are strongly positively correlated, or the system's \emph{pfe} is low. However, the assessor can take the more conservative view that these successful executions are evidence the system is not quite good enough. The assessor does this by holding the following beliefs: {\bf i)} if the \emph{pfe} is larger than $b$, then successful tests are most likely if the following two beliefs are true: the executions are ``perfectly positively correlated'' (i.e. $\lambda=1$), and the \emph{pfe} is ``as small as possible, but no smaller than $b$''. In terms of the Klotz likelihood, these beliefs are encoded as the location\footnote{$(b,1)$ is a \emph{limit point} -- the limit of a sequence of $\mathcal R$ locations that represent increasingly pessimistic beliefs. See supplementary material, \ref{sec_app_B}.} $(b, 1)$ in $\mathcal R$; {\bf ii)} if instead, the \emph{pfe} \emph{is} smaller than $b$, then successful tests are least likely if the following two beliefs are true: the executions are as ``negatively correlated'' as possible, and the \emph{pfe} is ``as big as possible, but no bigger than $b$''. These beliefs are encoded as the location $(b, 0)$. 

PKs refine these beliefs, giving the prior distributions shown in Fig.s~\ref{fig_CBInoFails_maintxt} and \ref{fig_CBIsoln_noFails}. These beliefs imply that as failure-free executions increase without bound, in order for the assessor to be conservative, their confidence in the bound must diminish (e.g., see dotted curve in Fig.~\ref{fig_example2_nuclear}). Because the increasing number of successful executions makes all other beliefs unlikely, except the beliefs that the executions are ``perfectly positively correlated'' and the \emph{pfe} is ``as small as possible, but no smaller than $b$''. The Klotz likelihood tends to zero at all of the nonzero probability locations in Fig.~\ref{fig_CBInoFails_maintxt}, except at the point $(b,1)$ where the likelihood has the constant value $(1-b)$ for all $n$. Failure-free executions eventually undermine confidence. 

It is possible that the successful executions are occurring because the \emph{pfe} is very small. The problem with this possibility is that any very small \emph{pfe} eventually becomes too unlikely to have produced a sufficiently large number of successful executions. Moreover, there are more pessimistic reasons (not disallowed by the evidence) for runs of successful executions. For example, all of the test inputs may have been ``easy'' for the software to respond to. This could happen by chance: e.g., the operational environment just happens to be submitting a sequence of easy inputs. Overcoming such problems of chance may require an infeasible amount of testing. Another pessimistic reason for the successful executions could be an error in the test-case generation procedure, which systematically fails to generate inputs that lie in the system's failure region. Or an error in the test oracle, which fails to indicate true failures. Such possibilities are consistent with previous works that show how ``favourable'' operational evidence can undermine confidence during assessments; e.g., \cite{littlewood_use_2007,SalakoZhao_TSE_2023}. To make progress with using failure-free evidence, an assessor must use appropriate additional evidence to rule out pessimistic reasons for not observing any failures.   

If the system \emph{could} be fault-free -- in modelling terms, the engineering goal $\epsilon$ is zero -- then a fault-free system \emph{could} produce the increasing number of successful executions. This possibility allows the assessor's confidence in the bound $b$ to grow, because the confidence an assessor has in the bound being satisfied is never smaller than their confidence in the system being fault-free. The confidence in the system being fault-free increases as the number of successful executions increases, thus increasing confidence in $b$ too. See Fig.~\ref{fig_example2_nuclear}, where confidence in $b$ increases from $\theta$ to $\frac{\theta}{\theta+(1-b)\phi_2}$ along the dotted curve. Contrarily, the successes could also be produced by more pessimistic reasons, so confidence in these more pessimistic reasons also increases. For example, the dotted curve in Fig.~\ref{fig_example2_nuclear} and the prior distribution in Fig.~\ref{fig:fig_CBInoFails_maintxt_Phi1gtTheta} that produced this curve, together imply that this confidence increases from $\frac{(1-b)\phi_2}{\theta+(1-b)(1-\theta)}$ to $\frac{(1-b)\phi_2}{\theta+(1-b)\phi_2}$, since $\phi_2\leqslant 1-\theta$ is assumed. So, the assessor will always be uncertain about whether the \emph{pfe} is better than $b$.

So far we have discussed conservative beliefs when only failure-free executions are observed. Ironically, a few failures during testing can help overcome the pessimism we have highlighted. As the number of successful executions increases, any initial failures become evidence that the executions cannot be ``perfectly positively correlated''. Otherwise, if the executions \emph{were} this strongly correlated, either no successes or no failures would have occurred! The previously pessimistic belief implied by location $(b,1)$ -- where $\lambda=1$ -- must now move to a different pessimistic location where $x<\lambda<1$ (indicating less strong positive correlation). Now, as successful executions increase, it eventually becomes most likely that the successes are being produced by a system with a \emph{pfe} smaller than $b$. Consequently, the assessor's confidence in $b$ eventually approaches certainty, although this can require considerable amounts of successful executions because of the few failures (e.g. Fig.~\ref{fig_av_with_failures_example3}). 

Like the ``failure-free'' scenario, there are also situations where testing is futile when some failures occur. The futility here is even more extreme than before: confidence in the bound is now identically zero, no matter how many more successful tests are observed. For instance, note the zero confidence in Fig.s~\ref{fig_sa_av_r0} and \ref{fig_sa_av_rnot0} from priors such as those in Fig.~\ref{fig_CBIsoln_Fails_zeroconf}. If the failures during testing are isolated and few, then strong confidence in the executions being positively correlated (i.e., $\phi_2\geqslant\theta$) undermines confidence in the bound $b$. If the failures are clustered and few, then strong confidence in the executions being negatively correlated (i.e., $\phi_1\geqslant\theta$) undermines confidence in $b$. In both situations, the assessor's prior confidence in the system being ``very reliable'' is simply not strong enough to rule out pessimistic causes for the testing outcomes.  

The various assessment scenarios are summarised in Fig.~\ref{fig_theBigPicture_discussion} (an updated version of Fig.~\ref{fig_theBigPicture_intro}). Each path through the figure -- starting from the far-left at the ``$n$ executions'' node -- gives the operational evidence and the prior confidence (i.e., PK parameter ranges) an assessor could have. The paths containing dashed lines are paths to be weary of; paths that our analyses reveal to be asymptotically futile or of little practical interest. Assessor's on such paths should seek stronger evidence of the system being sufficiently reliable prior to commencing testing.

\subsection{Is assuming Independence always very Optimistic?}
\begin{figure}[htbp!]
\centering
\includegraphics[width=0.45\linewidth]{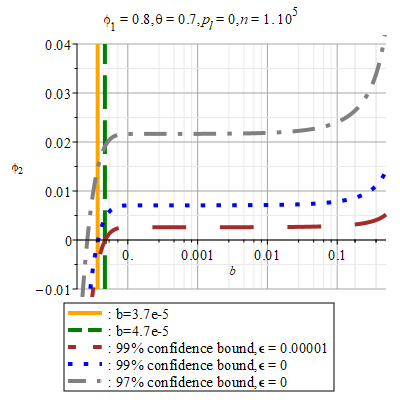}
\caption{{\footnotesize The relationship between prior confidence $\phi_2$ in positively correlated executions and the $(1-\alpha)\times100\%$ confidence bound $b$, when the system is subjected to $10^5$ tests without failure. The curves are obtained from posterior confidence given by the prior in Fig.~\ref{fig:fig_CBInoFails_maintxt_Phi1gtTheta}. The values of $\alpha$ plotted here are curves for $\alpha=0.01, 0.03$. The smaller $\phi_2$ becomes, the smaller the \emph{pfe} upper bound $b$ that can be supported at a given level of confidence. \kizito{legend needs to be fixed}\normalsize}}
\label{fig_confidencebounds_and_phi2}
\end{figure}
Concerning the impact of assuming independent executions, \cite{chen_binary_1996} observe the following when using classical inference with their Markov model: for a given number of successes/failures during testing, {\bf i)} positively correlated executions give bigger confidence bounds (i.e., worse values for \emph{pfe} estimates), compared with using Thayer \emph{et al.}'s independent executions model; {\bf ii)} a ``\emph{not so big}'' positive correlation gives confidence bounds that are comparable to those given by Thayer \emph{et al.}'s model. This is consistent with our findings. 

Indeed, consider the examples in Fig.~\ref{fig_confidencebounds_and_phi2}, where all $10^5$ executions are successful, and the system could be ``fault-free''. Then, for instance, the $99\%$ confidence bound from univariate CBI (i.e., CBI that assumes independence) is $3.7\times 10^{-5}$ -- the smallest $b$-value from the $99\%$ confidence bound (dotted) curve, precisely when $\phi_2=0$. All other $99\%$ confidence bounds from this curve -- i.e., all $99\%$ confidence bounds from CBI using the Klotz model, $\epsilon=0$ and $\phi_2>0$ -- are larger. Moreover, as $\phi_2$ increases, the confidence bounds $b$ increase. While, the smaller the assessor's prior confidence in positive correlation (i.e., $\phi_2$ decreases), the closer the Klotz confidence bound becomes to the confidence bound under independent executions (i.e., the intersection of the curves with the horizontal axes). 

If the system cannot be fault-free, CBI with the Klotz model is significantly more conservative. The long-dashed curve gives the $99\%$ confidence bounds when $\epsilon=10^{-5}$. Here, the univariate CBI $99\%$ confidence bound is $4.7\times 10^{-5}$ (at $\phi_2=0$) -- ``4 orders of magnitude'' smaller than the $99\%$ confidence bound $0.1$ from the Klotz model with $\phi_2=3.5\times 10^{-3}$. 

Interestingly, unlike Chen and Mill's other observation (that a reduction in confidence bounds accompanies an increase in negative correlation), section \ref{sec_senstivity_analysis} suggests that confidence bounds from failure-free testing are insensitive to prior confidence in negative correlation $\phi_1$. In fact, since $\phi_1\geqslant\theta$ in  Fig.~\ref{fig_confidencebounds_and_phi2}, the relevant posterior confidence is given by the prior in Fig.~\ref{fig:fig_CBInoFails_maintxt_Phi1gtTheta} as $\frac{(1-\epsilon)\left(1-\sfrac{\epsilon}{1-\epsilon}\right)^{n-1}\theta}{(1-\epsilon)\left(1-\sfrac{\epsilon}{1-\epsilon}\right)^{n-1}\theta + (1-\theta-\phi_2)(1-b)^n + (1-b)\phi_2}$, which doesn't depend on $\phi_1$. ``No failures'' supports confidence in positive correlation $\phi_2$, and undermines confidence in negative correlation $\phi_1$. 

So, the plots illustrate how conservative confidence bounds can be very sensitive to confidence in positive correlation -- i.e., small changes in $\phi_2$ can result in ``orders of magnitude'' changes in confidence bounds. If evidence strongly supports the executions being positively correlated, the confidence bounds obtained under assuming independence can be quite optimistic. On the other hand, \ref{sec_post_conf_in_independence} and Fig.~\ref{fig_doubtInIndependence_noperfection} of the supplementary material show how being skeptical about independent executions can be initially optimistic; giving smaller confidence in $b$ than the confidence from univariate CBI. While strongly believing in independence can be conservative initially. As successes accumulate, these roles between ``skepticism '' and ``strong belief'' are reversed, with ``skepticism'' in independence eventually giving conservative confidence and ``strong belief'' giving optimistic confidence.


\subsection{Limitations, Generalisations and Future Work}

The following Klotz model limitations, first highlighted in \cite{SalakoZhao_TSE_2023}, remain.

Using a relatively simple model of dependent Bernoulli trials -- i.e., the Klotz model -- we have illustrated how one might account for dependent executions in conservative reliability assessments. 
Of course, there is scope for studying the implications of more expressive failure-models. For instance, many systems experience different types of failure, some of which may be considered benign. Some Markov models that capture this include the models of \cite{Csenki_RecoveryBlocks_1993,goseva_popstojanova_failure_2000,Bondavalli_1999}. In contrast, the Klotz model treats all failure-types (and all successes-types) identically, in terms of how likely they are to occur, and the model ignores variations in the impact different failure-types have on system stakeholders when they occur. 

Another Klotz model limitation is that positive correlation in both of its forms -- i.e., whether failures are likely to follow previous failures or successes follow previous successes -- are captured by the size of parameter $\lambda$ relative to $x$ (see Remark \ref{remark_3_caes_of_dependance}). A further limitation is that the type of dependence -- i.e., whether executions are positively or negatively correlated, or independent -- is fixed for the duration of the system's operation. Certainly, there are practical situations where dependence can vary significantly over time; e.g., a change in the system's internal state makes failures much more/less likely. Or dependence can exist between several executions separated in time. Or, the sequence of executions could be halted whenever a failure occurs and the software could be fixed, before the software is allowed to resume executing -- thus altering the faults the software contains and the dependence among execution outcomes. Accounting for such dependence variation requires a failure-model that explicitly captures time-dependent correlations. These scenarios justify a weakening of the conditional independence in the Klotz model: in the model, $T_i$ is conditionally independent of $T_{i-2}, T_{i-3}, \ldots, T_1$ given $T_{i-1}$. In future work, it will be interesting to consider longer dependence structures (over several ``time steps'' into the past) -- e.g., $T_i$ being dependent on the last ``$i-1$'' execution outcomes. By applying the general conservative approach illustrated in this paper, an assessor can check the robustness of assessment claims based on models with more general dependence structures.

The Klotz model is a 1st-order stationary stochastic process (see \cite{klotz_statistical_1973,SalakoZhao_TSE_2023}). \emph{pfe}s used in reliability assessment make sense when the failure process is stationary. Because then, the probability of the system failing its $n$-th execution is the same for all $n$, and it equals the \emph{pfe}. This, despite the conditional probability of failing the very next execution being dependent on, say, the success/failure of the last execution. Consequently, upper confidence bounds on such \emph{pfe}s are useful measures of reliability in those practical scenarios characterised by a stationary failure process. But when failure probabilities are time-dependent, one should forego using \emph{pfe}s in assessment claims and opt for more suitable reliability measures, such as the probability of failure-free operation in the future (see \cite{strigini_testing_1996}).


Even with a 1st-order stationary model, it's still worth studying the impact of the independence assumption on reliability measures like the probability of future failure-free operation. Previous CBI studies have shown that an assessor's justifications for a conservative claim are often different for different measures -- even if the justifications are ultimately based on the same PKs. Also, some measures may be more sensitive to PK changes than other measures.


\section{Concluding Remarks}
\label{sec_conclusion}
Statistically independent software executions are often assumed when assessing software reliability. If inappropriate, this assumption can result in (dangerously) optimistic reliability claims. By formalising informal notions of ``doubting'' the independence assumption, and by employing conservative Bayesian methods, this work demonstrates how such doubts can be accounted for in assessments.

This paper contains analyses of various assessment scenarios. This involved the constrained mathematical optimisation of an assessor's confidence in an upper bound on the probability of failure per execution (\emph{pfe}), after observing the system in operation. 
The work highlights a number of practical considerations. For example, a system exhibiting no failures during operation can give \emph{less} confidence in a \emph{pfe} bound, compared with if the system \emph{had} exhibited failures. Or confidence can be very sensitive to failures; each additional failure means significantly more failure-free operation is needed for confidence to grow.

The scope of the results makes clear that a nuanced answer is required to the question of whether assuming independence undermines assessments. The answer depends, often sensitively, on various factors outlined in the paper. So that sometimes, the independence assumption has ``\emph{little to no}'' impact on conservatism. And sometimes, the impact is simply too great to ignore. A ``case-by-case'' approach to estimating this impact in practice is advised, and the methods and many solutions in this paper provide assessors/practitioners with the means to do this. 


\section*{Acknowledgements} This work was partly funded by the European Union's Horizon 2020 Research and Innovation Programme under grant agreement No 956123, and by the UK EPSRC through the End-to-End Conceptual Guarding of Neural Architectures [EP/T026995/1]. We thank Bev Littlewood for comments on earlier versions of the paper.

\bibliography{ref}

\bigskip

\begin{biography}{Kizito Salako} is a Lecturer at the department of computer science, City, University of London. He holds a double honours (1st-class) degree in mathematics and statistics from the University of Lagos; a Master of advanced study in mathematics degree from the University of Cambridge (where he was both a Shell Centenary scholar and a Commonwealth scholar); and a PhD in computer science from City, university of London. Kizito is passionate about applications of probability theory, Bayesian statistics, geometry and machine-learning, when simulating, assessing and forecasting the (failure) behaviour of software-based systems. His research produces statistical techniques that support conservative dependability claims for safety-critical systems. He also builds simulations of large-scale, complex, interdependent, critical infrastructure, in order to forecast the occurrence and impact of cascading failures, cyber-attacks and the efficacy of mitigation strategies.
\bigskip

\end{biography}

\begin{biography}{Xingyu Zhao} is a Assistant Professor in Safety-Critical Systems at WMG, University of Warwick. After receiving Bachelor and Masters degrees from the Beihang University, he earned a PhD in computer science at the Centre for Software Reliability, City, University of London in 2017. His research expertise covers probabilistic verification of autonomous systems, Bayesian inference with partial/vague prior knowledge, reliability assessment and safety assurance, and trustworthy AI. He has published 40+ papers in interdisciplinary fields of software engineering, AI, and system safety and reliability. Beyond publications, he has secured funding from UK EPSRC, UK DSTL and Innovate UK as a co-Investigator.
\bigskip
\end{biography}

\newpage
\appendix

\section{Conservative Bayesian Assessment}
\label{sec_app_B}
\subsection{Klotz Failure-Model Likelihood Function}
\label{sec_LikeFnAltForms}
The Klotz failure-model is naturally expressed in coordinates $(x, \lambda)$, where $x$ is the Bernoulli frequency parameter and $\lambda$ is the model's correlation parameter. If $\lambda^{\ast}$ denotes $\max\left\{0,\frac{2x-1}{x}\right\}$, the Klotz model likelihood function is well-defined\footnote{Strictly speaking, the likelihood function in \eqref{eqn_xKlotzlklhdFn_appndx} is not well-defined at the point $(1,1)$ except for appropriate ranges of the greek exponents. The likelihood's value at this point is $0$, and thus well-defined, if either the demand executions begin with a success, or $\delta>0$ (so failure is not an absorbing state). These sufficient conditions are satisfied in all practical scenarios of interest.} over the region $\mathcal R$ defined by $0\leqslant x \leqslant 1$, $\lambda^{\ast}\leqslant\lambda\leqslant 1$ (depicted in Fig.s~\ref{fig:fig_RegionR} and \ref{fig:fig_RegionRProbMasses}). $\mathcal R$ ensures the likelihood's magnitude is no greater than $1$. The likelihood has two primary forms:
\begin{equation}
L(x,\lambda;\alpha,\beta,\gamma,\delta) = \left\{ \begin{array}{lr}
	x\left(\frac{(1-\lambda)x}{1-x}\right)^{\alpha}\left(1-\frac{(1-\lambda)x}{1-x}\right)^{\beta}\lambda^{\gamma}(1-\lambda)^{\delta}; & \\
	\mbox{when the 1st execution is a failure} & \\ &
	\\
	(1-x)\left(\frac{(1-\lambda)x}{1-x}\right)^{\alpha}\left(1-\frac{(1-\lambda)x}{1-x}\right)^{\beta}\lambda^{\gamma}(1-\lambda)^{\delta}; &\\
	\mbox{when the 1st execution is a success}&\end{array}\right.
\label{eqn_xKlotzlklhdFn_appndx}
\end{equation}
where the greek exponents in the likelihood are fixed by the outcomes of $n$ system executions and satisfy $\alpha, \beta, \gamma, \delta \geqslant 0$.  

In practice, the Klotz likelihood is used as follows. Consider a system executing $n$ demands, with $s$ failed executions, and $r$ of these failures being \emph{consecutive failures} -- i.e. when a failure immediately follows a previous failure. Then, the Klotz likelihood becomes one of the following four expressions, depending on the particular values of the greek exponents in \eqref{eqn_xKlotzlklhdFn_appndx}. That is, depending on whether the $n$ executions, 

\noindent \emph{1) begin with a failure and end with a failure}: 	
\begin{align}
x\left(\frac{(1-\lambda)x}{1-x}\right)^{s-r-1}\lambda^r(1-\lambda)^{s-r-1}\left(1-\frac{(1-\lambda)x}{1-x}\right)^{n-2s+r+1}  &
\label{eqn_failfail}	
\end{align}
when $\alpha=s-r-1$, $\beta=n-2s+r+1$, $\gamma=r$ and $\delta=s-r-1$;

\noindent \emph{2) begin with a success and end with a failure}:	
\begin{align}
(1-x)\left(\frac{(1-\lambda)x}{1-x}\right)^{s-r}\lambda^r(1-\lambda)^{s-r-1}\left(1-\frac{(1-\lambda)x}{1-x}\right)^{n-2s+r}  &
\label{eqn_succfail}
\end{align}
when $\alpha=s-r$, $\beta=n-2s+r$, $\gamma=r$ and $\delta=s-r-1$;

\noindent	\emph{3) begin with a success and end with a success}:	
\begin{align}
(1-x)\left(\frac{(1-\lambda)x}{1-x}\right)^{s-r}\lambda^r(1-\lambda)^{s-r}\left(1-\frac{(1-\lambda)x}{1-x}\right)^{n-2s+r-1}  &
\label{eqn_succsucc}
\end{align}
when $\alpha=s-r$, $\beta=n-2s+r-1$, $\gamma=r$ and $\delta=s-r$;

\noindent	\emph{4) begin with a failure and end with a success}:	
\begin{align}
x\left(\frac{(1-\lambda)x}{1-x}\right)^{s-r-1}\lambda^r(1-\lambda)^{s-r}\left(1-\frac{(1-\lambda)x}{1-x}\right)^{n-2s+r}  &
\label{eqn_failsucc}
\end{align}
when $\alpha=s-r-1$, $\beta=n-2s+r$, $\gamma=r$ and $\delta=s-r$.

In alternative coordinates $(\lambda,y)$, defined by the transformation $\lambda=\lambda$ and $y=\frac{(1-\lambda)x}{1-x}$ from $(x,\lambda)$ coordinates, the Klotz likelihood \eqref{eqn_xKlotzlklhdFn_appndx} has the equivalent forms:
\begin{equation}
L(\lambda,y;\alpha,\beta,\gamma,\delta) = \left\{ \begin{array}{lr}
	\frac{y}{y+1-\lambda}y^{\alpha}(1-y)^{\beta}\lambda^{\gamma}(1-\lambda)^{\delta}; & \\
	\mbox{when the 1st execution is a failure} & \\ &
	\\
	\frac{1-\lambda}{y+1-\lambda}y^{\alpha}(1-y)^{\beta}\lambda^{\gamma}(1-\lambda)^{\delta}; &\\
	\mbox{when the 1st execution is a success}&\end{array}\right.
\label{eqn_yKlotzlklhdFn}
\end{equation}	
In what follows, we will use the likelihood function in either $(x,\lambda)$ or $(\lambda, y)$ coordinates (whichever is more convenient to work with), for suitable $\alpha, \beta, \gamma$ and $\delta$.

\begin{figure}[h!]
	\captionsetup[figure]{format=hang}
	\begin{subfigure}[]{1.0\linewidth}
		\centering
		\includegraphics[width=0.4\linewidth]{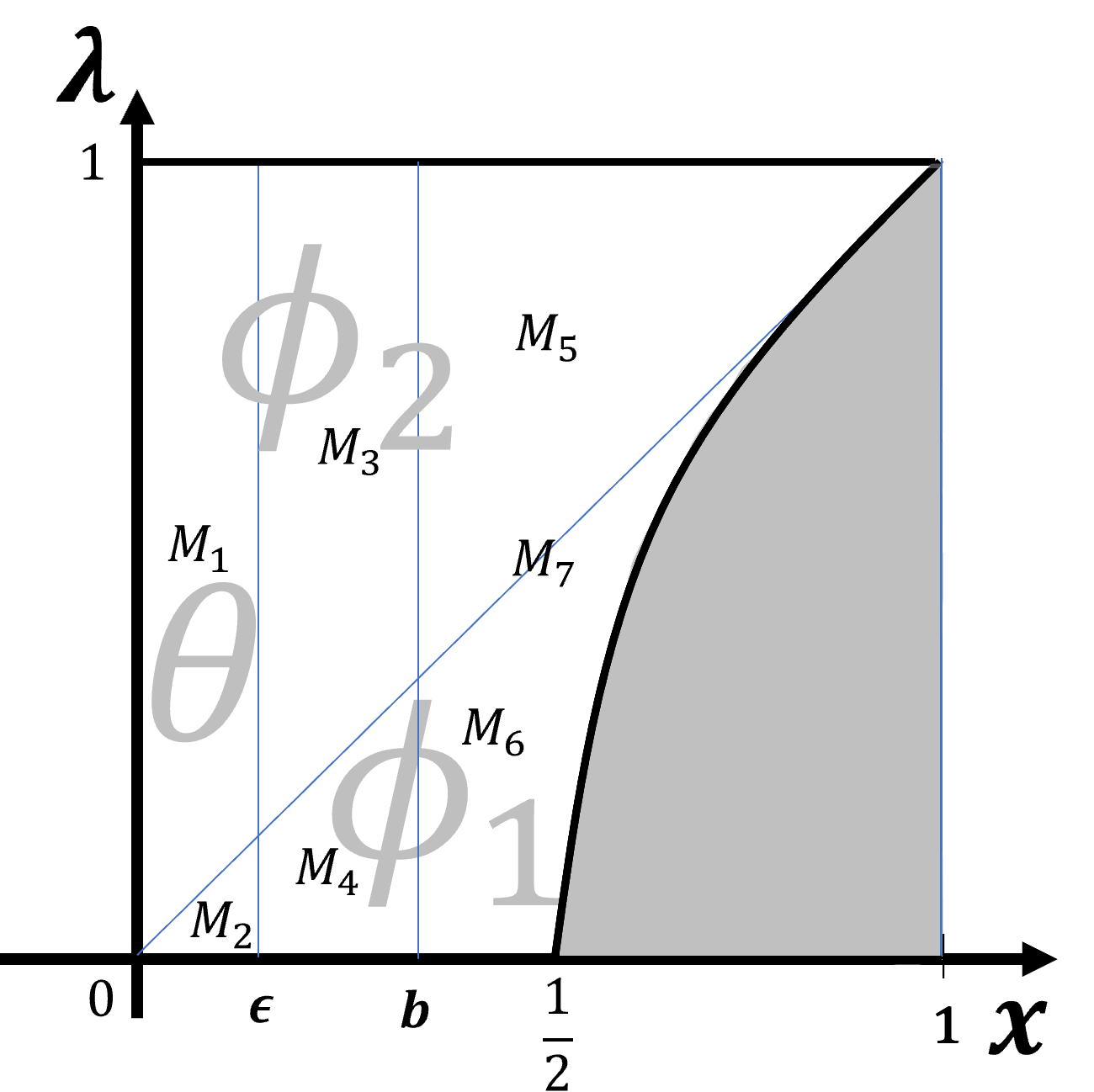}
		\caption[Region $\mathcal R$s probability masses]{{\small A feasible joint distribution of $(X, \Lambda)$, over the region $\mathcal R$, allocates probability masses $\{M_i\}$ over $\mathcal R$ subsets. The masses must satisfy $M_1+M_3+M_5=\Phi_2$, $M_2+M_4+M_6=\Phi_1$, and probability $M_7=1-\sum_{i=1}^{6}M_i$ is allocated to the diagonal.   \normalsize}}
		\label{fig:fig_RegionRProbMasses}
	\end{subfigure}
	\begin{subfigure}[]{1.0\linewidth}
		\centering
		\includegraphics[width=0.4\linewidth]{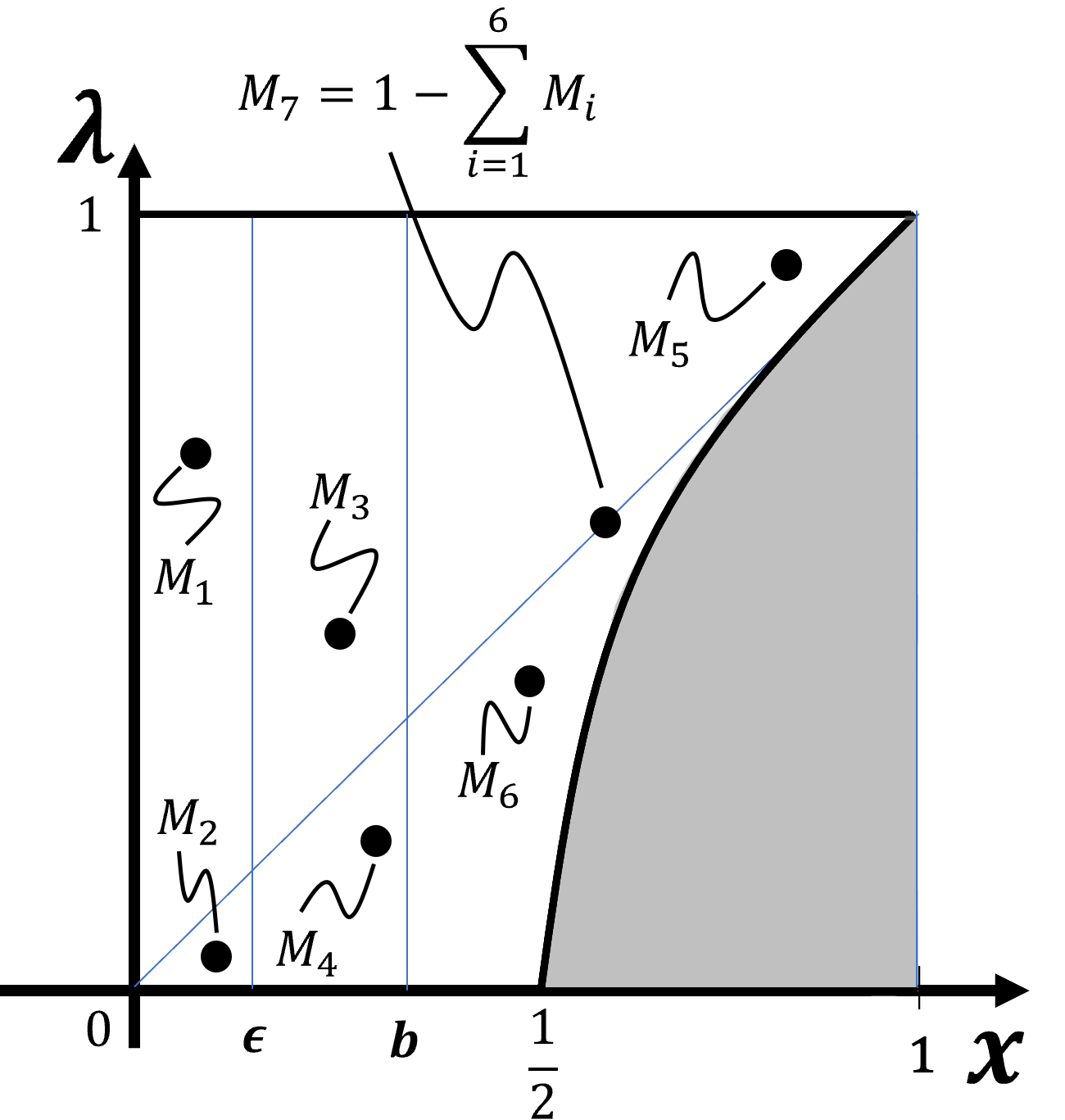}
		\caption[Discrete Priors over $\mathcal R$]{{\small A discrete joint prior distribution over $\mathcal R$. It allocates probability masses $M_i$ to single points within each subset and along the diagonal.   \normalsize}}
		\label{fig:fig_RregionDiscretePriors}
	\end{subfigure}
	\caption{}
\end{figure}

\subsection{Posterior Confidence in a Bound on the Probability of Failure per Execution}
\label{sec_proofThrm1}	
We prove the following theorem.	
\begin{theorem}
Let $\mathcal D$ be the set of all prior distributions over the region $\mathcal R$ and let $0\leqslant\epsilon<b<\frac{1}{2}$ (see Fig.~\ref{fig:fig_RegionRProbMasses}). The optimisation problem
\begin{align*}
	&\qquad\inf\limits_{\mathcal D} P(\,X\leqslant b \mid \mbox{outcomes of }n\mbox{ executions} ) \\
	&\mbox{s.t.} \,\,\,\,\,P(\Lambda< X) = \phi_1,\,\,\,\,P(\Lambda> X) = \phi_2,\,\,\,\,P(X\leqslant\epsilon) = \theta 
\end{align*}
is solved by prior distributions such as those in Fig.s~\ref{fig_CBIsoln_withFails_Phi2glt1minustheta} --  \ref{fig_CBIsoln_noFails}, since $P(\, X<b \mid \mbox{\it outcomes of }n \mbox{\it\ executions } )$ from these priors (for $p_l=0$) equals the infimum.
\label{theorem_CBIwithdependence}
\end{theorem}

\begin{proof}
The optimisation constraints can be used to partition $\mathcal R$ into 7 disjoint subsets (with one of the subsets being the $45^{\circ}$ diagonal). Each prior distribution $F\in\mathcal D$ must assign 7 probabilities $\{M_i\}^{7}_{i=1}$ to these subsets, in such a way as to satisfy the constraints of the optimisation problem (see Fig.~\ref{fig:fig_RegionRProbMasses}).

The proof progresses in 6 stages:
\begin{enumerate}
	\item restrict the optimisation from $\mathcal D$ to its subset  ${\mathcal D}^{\prime}$ of discrete prior distributions. An arbitrary discrete prior assigns its probabilities $M_i$ to $7$ arbitrary points $\{(x_i,\lambda_i)\}_{i=1}^{7}$ within $\mathcal R$. Hence, the objective function becomes a rational function of the ``$x_i$''s, ``$\lambda_i$''s and ``$M_i$''s;
	\item show that the gradient of this objective function is determined by the gradient of the Klotz model likelihood;
	\item show the likelihood is unimodal along vertical and horizontal lines in $\mathcal R$, as well as along the $45^{\circ}$ diagonal line;
	\item show that the likelihood is also unimodal over all of  $\mathcal R$, and it attains its maximum either at a stationary point in the interior of $\mathcal R$, or along the boundary of $\mathcal R$;
	\item the previous steps in the proof imply the following: starting from any $F\in{\mathcal D}^{\prime}$, and the probabilities $\{M_i\}$ assigned by $F$, we can construct a new prior that gives a smaller value for the objective function (compared with $F$'s objective function value). We simply use the gradient of the likelihood to determine new locations within each of the $7$ $\mathcal R$-subsets, and reassign the ``$M_i$''s to these new locations. This reassignment produces a new prior distribution, which in turn can have \emph{its} probabilities reassigned to new points (and so on, indefinitely). In the limit, depending on the values of $\alpha, \beta, \gamma$ and $\delta$, the sequence of new points obtained by successive reassignments will converge to \emph{limit points}\footnote{Definition: for a given topology (e.g., the ``open balls'' topology associated with 2D Euclidean space), a \emph{limit point} of a subset of the plane is a point that is arbitrarily well-approximated by sequences of points within the subset (see \cite{rudin1976principles,copson_1968}).} in each $\mathcal R$-subset. That is, the objective function values converge in a monotonically decreasing manner, as the sequence of ``reassigned'' priors converge to a limiting distribution with support at, no more than, 7 \emph{limit points};
	\item finally, determine the values for the ``$M_i$''s that a limiting distribution should assign to \emph{limit points} -- to ensure that the related sequence of objective function values converge to the infimum. Determining these worst-case ``$M_i$''s is a constrained \emph{linear fractional programming} problem. One may solve this either numerically, or by a logical allocation of probability masses to the relevant \emph{limit points} in $\mathcal R$. For the CBI solutions in this paper, we use the latter approach. These final forms of limiting distribution (illustrated in Fig.s~\ref{fig_CBIsoln_withFails_Phi2glt1minustheta} -- \ref{fig_CBIsoln_noFails}) are worst-case prior distributions; so-called because $P(\,X< b \mid \mbox{\emph{outcomes of } }n \mbox{ \emph{executions} } )$ for these distributions equals the infimum we seek.   
\end{enumerate}

\noindent Let us proceed with the proof:

\emph{stage 1)} By definition, for any $F\in \mathcal D$,
\begin{align*}
	P(\,X\leqslant b \mid \mbox{\it outcomes of }n\mbox{\it\ executions}  ) =\frac{\MyExp[L(X,\Lambda;\alpha,\beta,\gamma,\delta){\bf 1}_{X\leqslant b}]}{\MyExp[L(X,\Lambda;\alpha,\beta,\gamma,\delta)]} 
	=\frac{\int_{[0,b]\times[\lambda^\ast,1]} L(x,\lambda;\alpha,\beta,\gamma,\delta) \,\mathtt{d}F(x,\lambda)}{\int_{[0,1]\times[\lambda^\ast,1]} L(x,\lambda;\alpha,\beta,\gamma,\delta)\,\mathtt{d}F(x,\lambda)}
\end{align*}
However, the set {$\mathcal D$} can be restricted to the subset {$\mathcal D^{\prime}$} of discrete joint distributions -- i.e., to those distributions that assign their ``$M_i$''s to single points within each $\mathcal R$ subset (e.g. see Fig.~\ref{fig:fig_RregionDiscretePriors}), see \cite{Moreno_1991}. 
So that, for any $F\in \mathcal D^{\prime}$, the objective function of the optimisation becomes
\begin{align*}
	P(\,X\leqslant b \mid \mbox{\it outcomes of }n\mbox{\it\ executions}  ) &=\frac{\int_{[0,b]\times[\lambda^\ast,1]} L(x,\lambda;\alpha,\beta,\gamma,\delta) \,\mathtt{d}F(x,\lambda)}{\int_{[0,1]\times[\lambda^\ast,1]} L(x,\lambda;\alpha,\beta,\gamma,\delta)\,\mathtt{d}F(x,\lambda)} 
	\\
	&= \frac{\sum_{i=1}^{7} L(x_i,\lambda_i;\alpha,\beta,\gamma,\delta){\bf 1}_{x_i\leqslant b} \,M_i}{\sum_{i=1}^{7} L(x_i,\lambda_i; \alpha,\beta,\gamma,\delta)\,M_i} 
	\\
	&=\frac{Num}{Denum}
\end{align*}
Consequently the objective function has become $\frac{Num}{Denum}$; a rational function of the ``$x_i$''s, ``$\lambda_i$''s and ``$M_i$''s.

\emph{stage 2)} Consider how this objective function changes when restricted to a vertical line in the subset of $\mathcal R$ where $x\leqslant \frac{1}{2}$. The rate of change of $\frac{Num}{Denum}$ with respect to $\lambda$ is then
\begin{align}
	&\frac{\partial}{\partial \lambda}\left(\frac{Num}{Denum}\right)=\frac{\frac{\partial}{\partial \lambda}Denum}{Denum}\left(\frac{\frac{\partial}{\partial \lambda}Num}{\frac{\partial}{\partial \lambda}Denum}-\frac{Num}{Denum}\right)
	\label{eq_deriv_obj_lambda}
\end{align}
Since $\frac{Num}{Denum}$ is a rational function of $\lambda$, it is smooth (except where $Denum=0$). Consequently, the sign of $\frac{\partial}{\partial \lambda}\left(\frac{Num}{Denum}\right)$ indicates how to move the location of the ``$M_i$''s along vertical lines in each $\mathcal R$ subset, in order to minimise $\frac{Num}{Denum}$. 

The following argument shows how the sign of $\frac{\partial}{\partial \lambda}\left(\frac{Num}{Denum}\right)$ is determined by the sign of $\frac{\partial}{\partial \lambda}L(x,\lambda;\alpha, \beta,\gamma,\delta)$ and the size of $x$. Observe that $\frac{\partial}{\partial \lambda}\left(\frac{Num}{Denum}\right)\geqslant 0$ \emph{if, and only if}, $\frac{\partial}{\partial \lambda}Denum$ and $\left(\frac{\frac{\partial}{\partial \lambda}Num}{\frac{\partial}{\partial \lambda}Denum}-\frac{Num}{Denum}\right)$ have the same sign. When the $\mathcal R$ region satisfies $x\leqslant b$, this implies  $\frac{\frac{\partial}{\partial \lambda}Num}{\frac{\partial}{\partial \lambda}Denum}=1$ for $\lambda$ from that region. Substituting $1$ for $\frac{\frac{\partial}{\partial \lambda}Num}{\frac{\partial}{\partial \lambda}Denum}$ in  \eqref{eq_deriv_obj_lambda}, and noting that the objective function $\frac{Num}{Denum}$ is a probability (hence must lie between $0$ and $1$),
we have that $\frac{\partial}{\partial \lambda}\left(\frac{Num}{Denum}\right)$ and $\frac{\partial}{\partial \lambda}Denum$ share the same sign when $x\leqslant b$. When $x\geqslant b$ instead, $\frac{\partial}{\partial \lambda}\left(\frac{Num}{Denum}\right)$ and $\frac{\partial}{\partial \lambda}Denum$ have opposite signs\footnote{These statements exclude the unimportant edge cases when $\frac{Num}{Denum}=0,1$.}. So, the sign of $\frac{\partial}{\partial \lambda}Denum$ (i.e. the sign of $\frac{\partial}{\partial \lambda}L(x,\lambda;\alpha, \beta,\gamma,\delta)$), and the value of $x$, together determine the sign of $\frac{\partial}{\partial \lambda}\left(\frac{Num}{Denum}\right)$. A similar argument shows that along horizontal lines in $\mathcal R$, $\frac{\partial}{\partial x}L(x,\lambda;\alpha, \beta,\gamma,\delta)$ and $x$'s value determine the sign of $\frac{\partial}{\partial x}\left(\frac{Num}{Denum}\right)$. And thus, they determine where the ``$M_i$''s should be allocated to minimize $\frac{Num}{Denum}$ along that line.

\emph{stage 3)} Along any vertical line in $\mathcal R$ where $x\leqslant \frac{1}{2}$ is satisfied, $L(x,\lambda;\alpha,\beta,\gamma,\delta)$ is a non-negative unimodal function of $\lambda$. To see this, note that $\frac{\partial}{\partial \lambda}L(x,\lambda;\alpha, \beta,\gamma,\delta)=0$ has non-trivial solutions at $\lambda$ values where two quadratic functions of $\lambda$ intersect. That is, solutions to 
\begin{align}
	\left(1-\frac{(1-\lambda)}{1-x}x\right)(\gamma-\lambda(\gamma+\delta))=\lambda\left(\alpha-(\alpha+\beta)\frac{(1-\lambda)}{1-x}x\right)
	\label{eq_lambdaMode}
\end{align}

An illustration of these two functions is given in Fig.~\ref{fig_LiklhdstatpointverticalLine}. One function has two roots of opposite sign (at $\lambda = \frac{2x-1}{x}, \frac{\gamma}{\gamma+\delta}$) and a maximum, while the other function has a root at $\lambda=0$ and a minimum. This means at least one solution to \eqref{eq_lambdaMode} cannot lie within $\mathcal R$ -- it must be non-positive. And the other solution must be positive and represents a maximum turning point. Because the l.h.s of \eqref{eq_lambdaMode} is bigger than the r.h.s. for $\lambda$ values slightly smaller than the positive solution, and the l.h.s. is smaller than the r.h.s. for all $\lambda$ values bigger than the positive solution.


\begin{figure}[h!]
	\centering
	\begin{subfigure}[]{0.35\linewidth}
		\centering
			\includegraphics[width=1.0\linewidth]{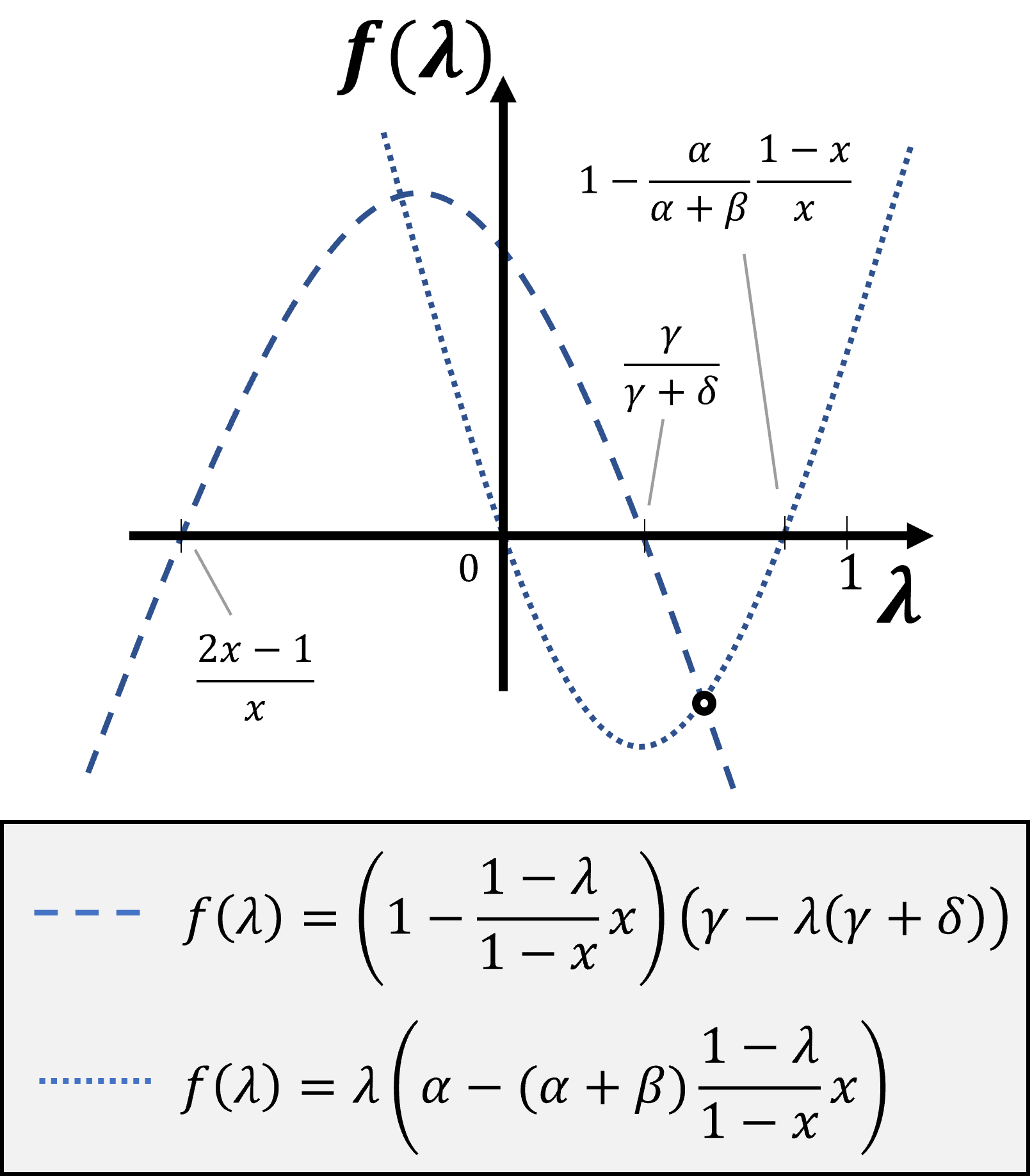}
		\caption{{\footnotesize $0\leqslant\frac{\alpha(1-x)}{(\alpha+\beta)x}\leqslant 1$}}
		\label{fig:fig_LiklhdstatpointverticalLine_1}
	\end{subfigure}
	\begin{subfigure}[]{0.35\linewidth}
		\centering
			\includegraphics[width=1.0\linewidth]{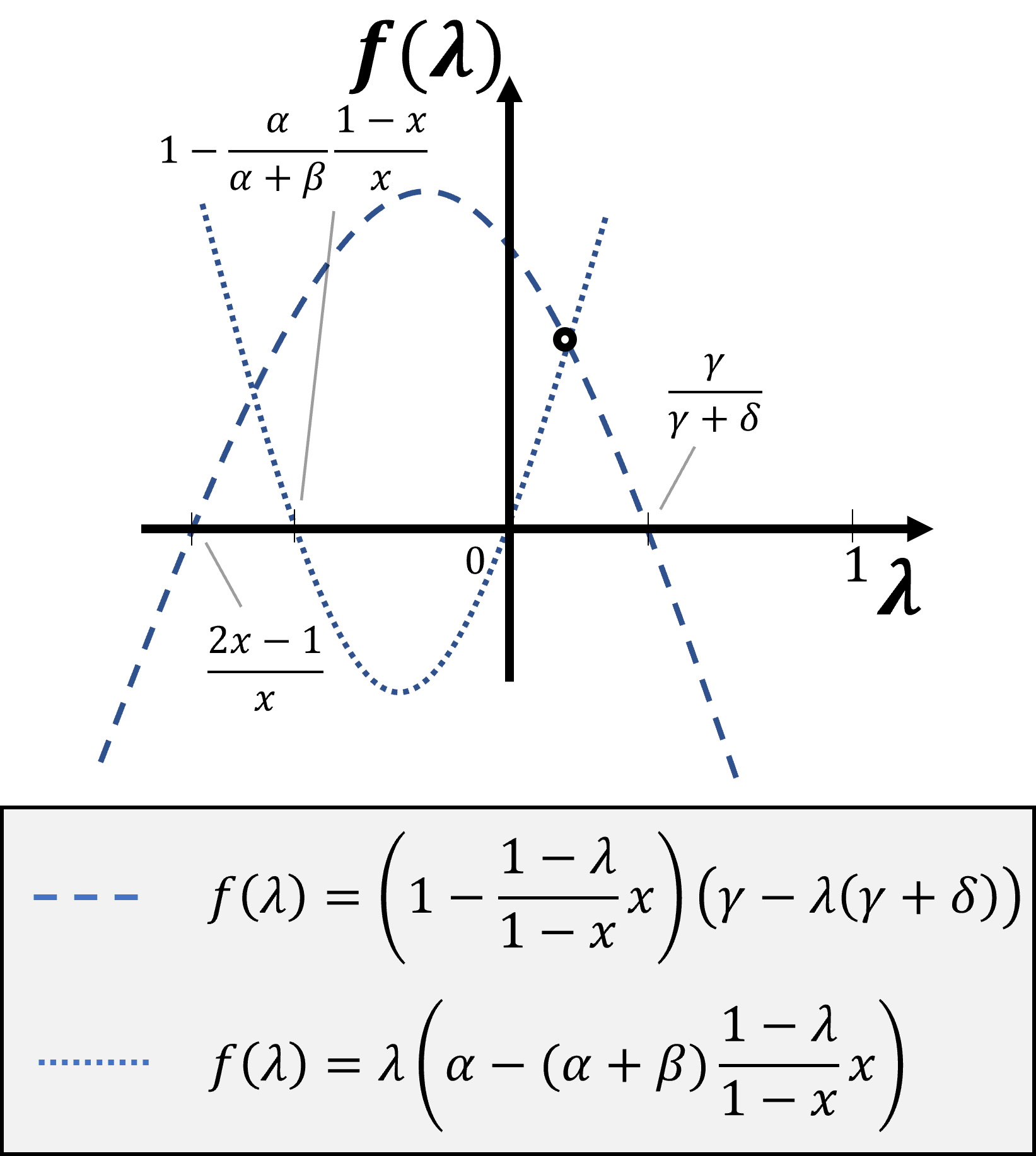}
		\caption{{\footnotesize $1\leqslant\frac{\alpha(1-x)}{(\alpha+\beta)x}$}}
		\label{fig:fig_LiklhdstatpointverticalLine_2}
	\end{subfigure}
	\caption[Prior Constraints]{{\small Two illustrations of two quadratic functions of $\lambda$ having, at most, one intersection over the range $0<\lambda<1$. For fixed $x$, this geometric fact implies $L(x,\lambda;\alpha,\beta,\gamma,\delta)$ is unimodal over any vertical line in $\mathcal R$ such that $x\leqslant\frac{1}{2}$.  \normalsize}}
	\label{fig_LiklhdstatpointverticalLine}
\end{figure}

Thus, as $\lambda$ grows from $0$ to $1$ along any vertical line in $\mathcal R$ (where $x\leqslant \frac{1}{2}$), there is (at most) one stationary point at which the likelihood is maximum. The likelihood is monotonic on either side of this maximum along the vertical line.

$L(x,\lambda;\alpha,\beta,\gamma,\delta)$ is also unimodal along the $45^{\circ}$ diagonal (i.e., when $x=\lambda$). Because it has only one non-trivial stationary point\footnote{This stationary point is located at either $x=\frac{1+\alpha+\gamma}{1+\alpha+\gamma+\beta+\delta}$ or $x=\frac{\alpha+\gamma}{1+\alpha+\gamma+\beta+\delta}$, depending on whether executions begin with a failure or success respectively.}, and this must be a maximum since the likelihood is non-negative with value $0$ at the endpoints of the diagonal.

\begin{figure}[h!]
	\centering
	\begin{subfigure}[]{0.35\linewidth}
		\centering
			\includegraphics[width=1.0\linewidth]{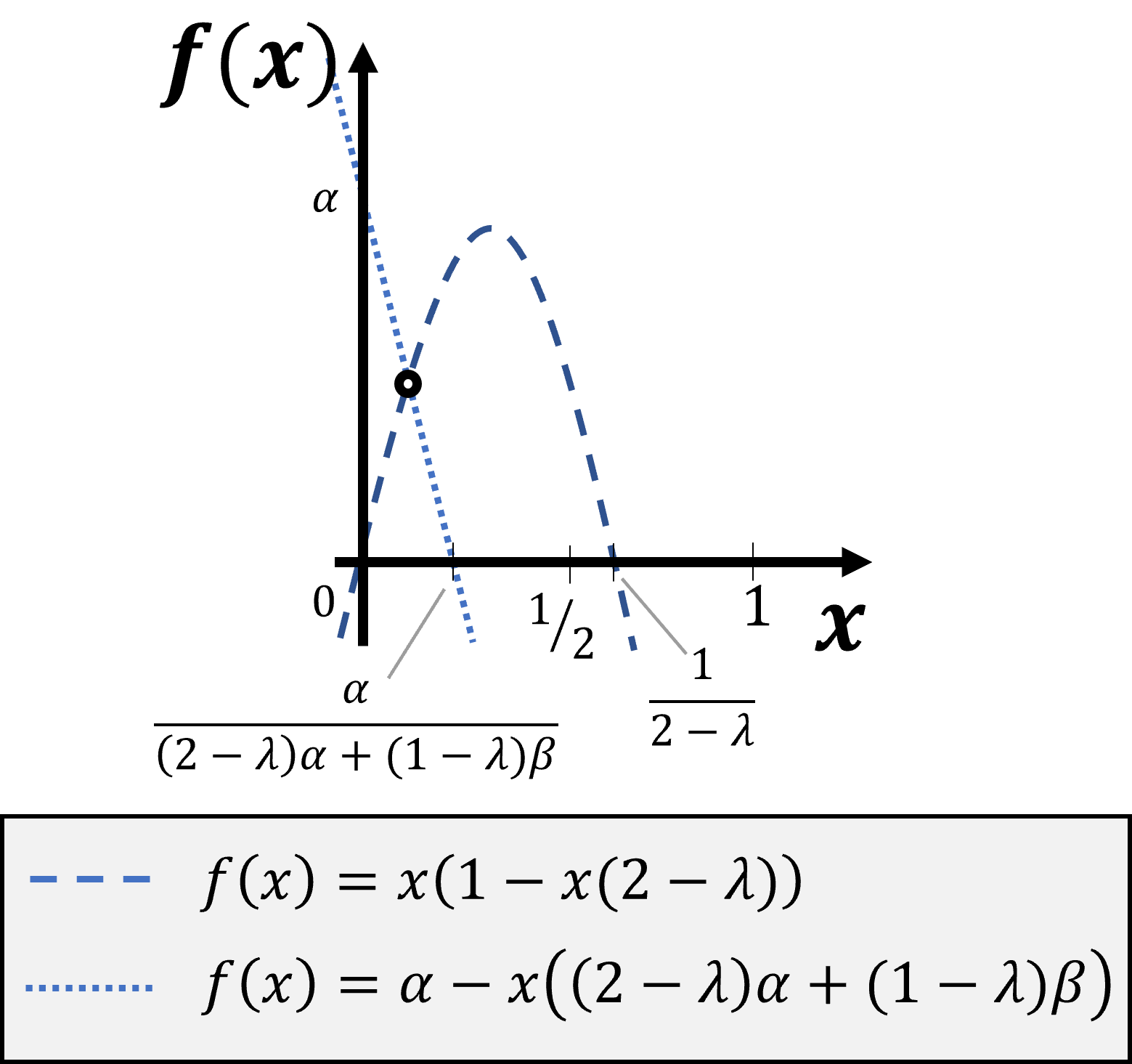}
		\caption{{\footnotesize From likelihood with a successful 1st execution}}
		\label{fig:fig_LiklhdstatpointhorizontalLine_1}
	\end{subfigure}
	\begin{subfigure}[]{0.35\linewidth}
		\centering
			\includegraphics[width=1.0\linewidth]{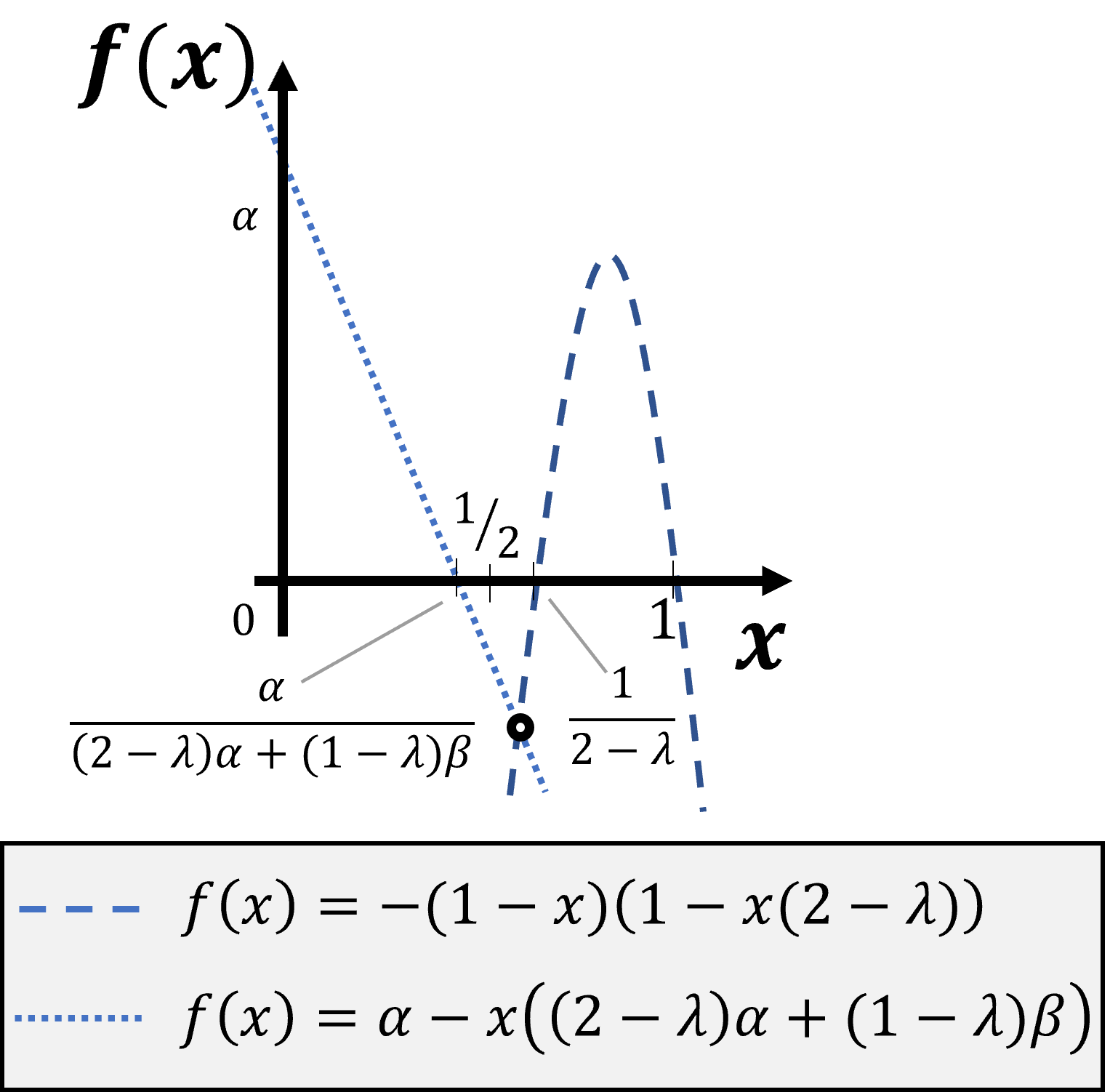}
		\caption{{\footnotesize From likelihood with a failed 1st execution}\vspace{0.2cm}}
		\label{fig:fig_LiklhdstatpointhorizontalLine_2}
	\end{subfigure}
	\caption[Prior Constraints]{{\small Two illustrations of quadratic and linear functions of $x$ having, at most, one intersection over the range $0<x<\frac{1}{2}$. For fixed $\lambda$, this geometric fact implies $L(x,\lambda;\alpha,\beta,\gamma,\delta)$ is unimodal over any horizontal line in $\mathcal R$. \normalsize}}
	\label{fig_LiklhdstatpointhorizontalLine}
\end{figure}

Analogously, $L(x,\lambda;\alpha,\beta,\gamma,\delta)$ is unimodal along any horizontal line within $\mathcal R$, since it has (at most) one stationary point at which it attains a maximum. The stationary point solves $\frac{\partial}{\partial x}L(x,\lambda;\alpha, \beta,\gamma,\delta)=0$ non-trivially. Equivalently, the stationary point satisfies the leftmost intersection between a straight line and a quadratic function in $x$ (see Fig.~\ref{fig_LiklhdstatpointhorizontalLine}). This leftmost intersection must occur to the left of $x=\frac{1}{2-\lambda}$, ensuring that the stationary point lies in $\mathcal R$. The fact that the line lies above the quadratic before this intersection, and then below the quadratic immediately after, ensures that, as $x$ increases from $0$, the $\frac{\partial}{\partial x}L(x,\lambda;\alpha, \beta,\gamma,\delta)$ transitions from being positive to being negative. That is, the stationary point is a maximum.  

\begin{figure}[h!]
	\captionsetup[figure]{format=hang}
	\centering
	\begin{subfigure}[]{0.35\linewidth}
		\centering
			\includegraphics[width=1.0\linewidth]{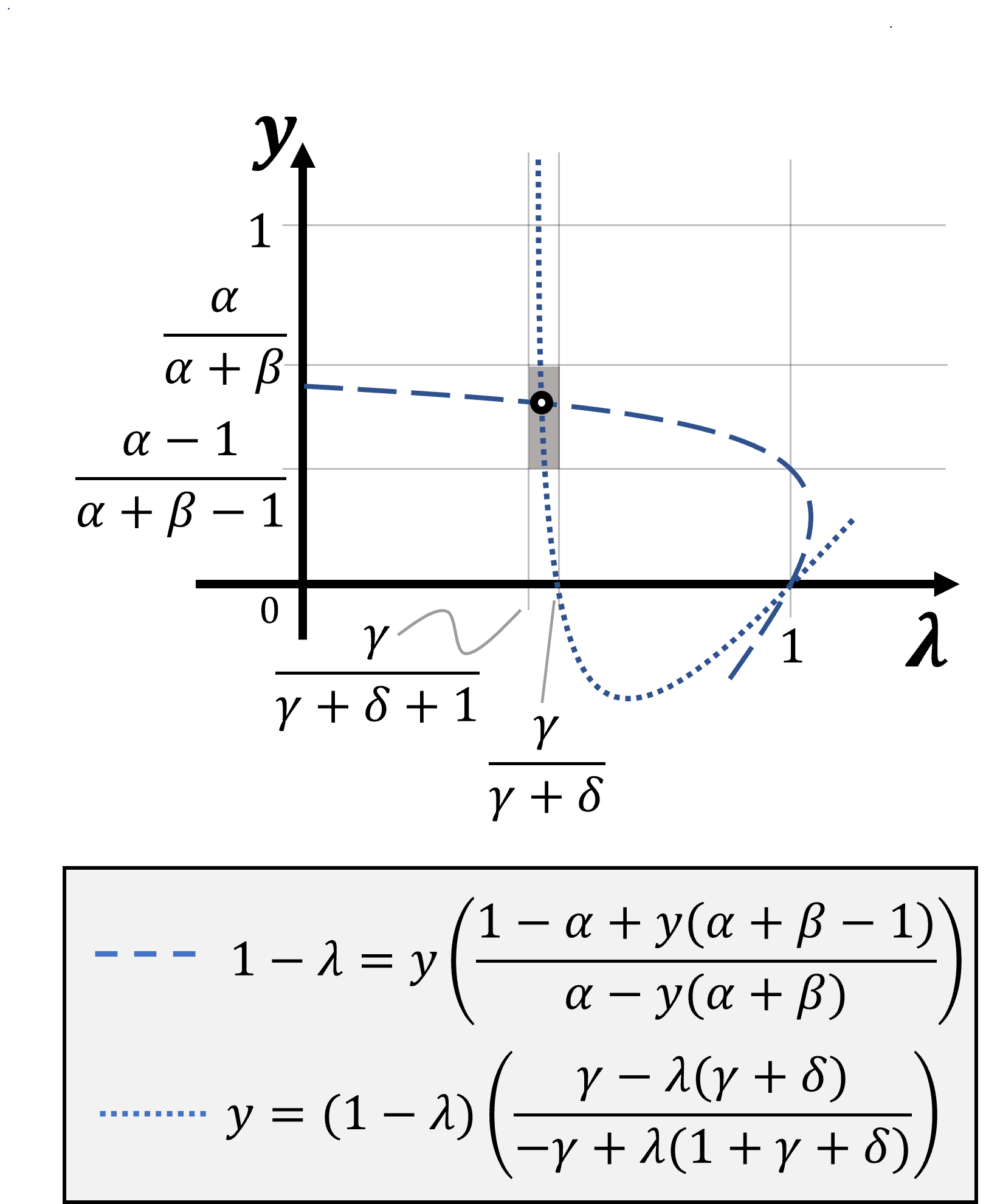}
		\caption{{\footnotesize From likelihood with a successful 1st execution, $\gamma+\delta \neq 0$ and $\alpha+\beta\neq 1,0$ }}
		\label{fig:fig_SingleStationary Point_ffns}
	\end{subfigure}
	\begin{subfigure}[]{0.35\linewidth}
		\centering\vspace{0.75cm}
			\includegraphics[width=1.0\linewidth]{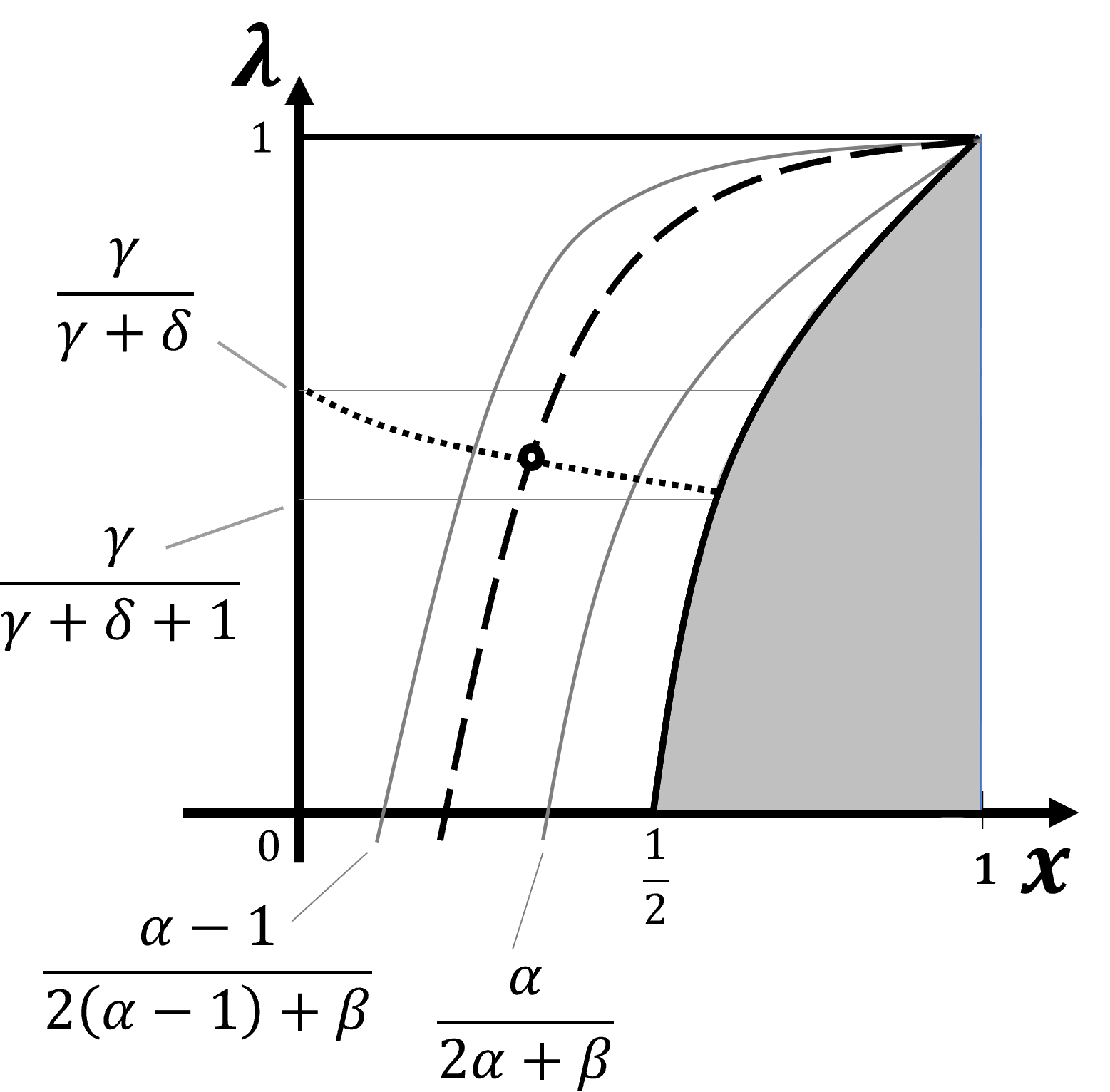}
		\caption{{\footnotesize The stationary curves of Fig.~\ref{fig:fig_SingleStationary Point_ffns} depicted in $\mathcal R$. }}
		\label{fig:fig_RregionSingleStationary Point_ffns}
	\end{subfigure}
	\begin{subfigure}[]{0.35\linewidth}
		\centering
			\includegraphics[width=1.0\linewidth]{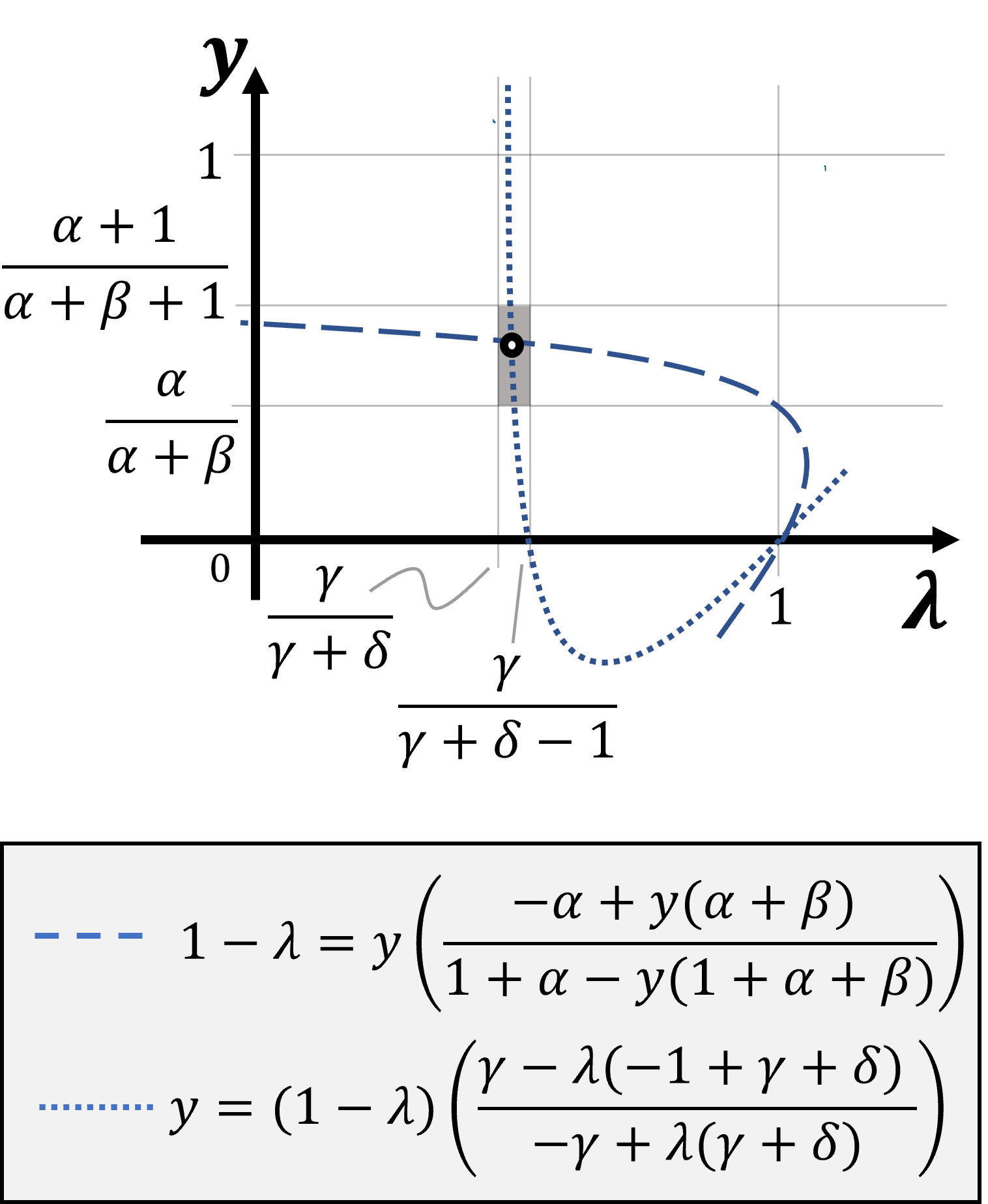}
		\caption{{\footnotesize From likelihood with a failed 1st execution, $\alpha+\beta\neq 0$ and $\gamma +\delta\neq 1,0$ }}
		\label{fig:fig_SingleStationary Point_gfns}
	\end{subfigure}
	\begin{subfigure}[]{0.35\linewidth}
		\centering\vspace{0.75cm}
			\includegraphics[width=1.0\linewidth]{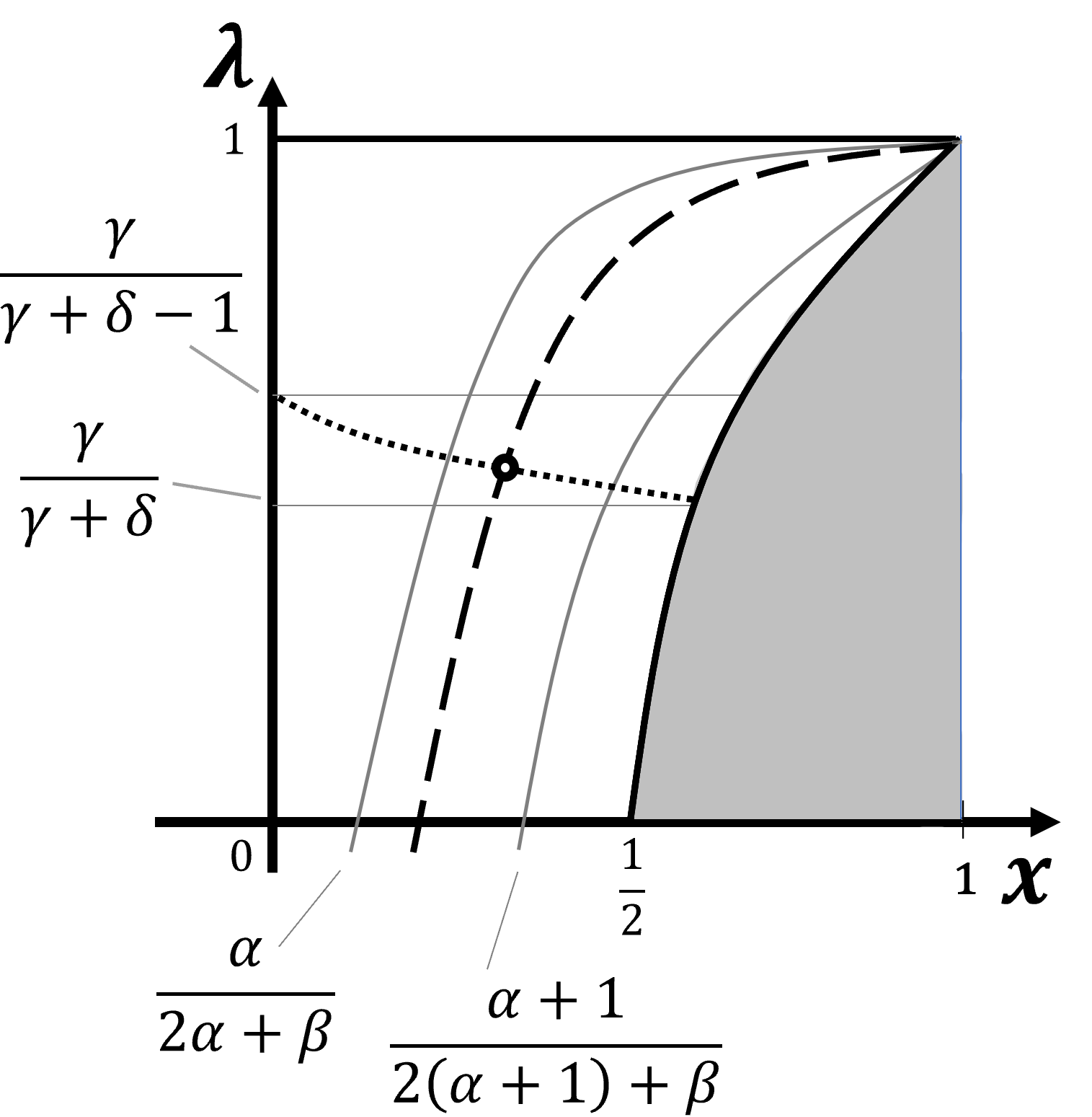}
		\caption{{\footnotesize The stationary curves of Fig.~\ref{fig:fig_SingleStationary Point_gfns} depicted in $\mathcal R$. }}
		\label{fig:fig_RregionSingleStationary Point_gfns}
	\end{subfigure}
	\caption[single global stationary point]{{\small Two illustrations of pairs of curves intersecting, at most, once, over the range $0<\lambda<1, 0<y<1$. This geometric fact implies $L(x,\lambda;\alpha,\beta,\gamma,\delta)$ is unimodal over $\mathcal R$ with, at most, one stationary point in the interior of $\mathcal R$ .  \normalsize}}
	\label{fig_singleStationarypoint}
\end{figure}

\emph{stage 4)} Finally, the likelihood either has a single stationary point within $\mathcal R$ at which it attains a maximum value over $\mathcal R$, or it attains its maximum value over $\mathcal R$ on the boundary of $\mathcal R$. When the single stationary point lies within $\mathcal R$, it can be determined by solving $\frac{\partial}{\partial \lambda}L(\lambda, y;\alpha,\beta,\gamma,\delta) = 0$ and $\frac{\partial}{\partial y}L(\lambda, y;\alpha,\beta,\gamma,\delta) = 0$ simultaneously for $\lambda$ and $y$. The non-trivial solutions for this simultaneous system of equations are given by the intersections of pairs of curves, as illustrated in Fig.~\ref{fig_singleStationarypoint}.

If the stationary point lies within $\mathcal R$ then it must be a maximum; because the stationary curves in Fig.~\ref{fig_singleStationarypoint} imply that, from any point along the boundary of $\mathcal R$, we can always move away from that point along an appropriate path within $\mathcal R$ to increase the likelihood's value.    

\emph{stage 5)} stages 1-4 of this proof demonstrate the existence and uniqueness of locations in $\mathcal R$ that are local or global maxima, as exemplified in Fig.~\ref{fig_allStationarypoints}. For the region $x\leqslant b$ in $\mathcal R$, the ``further away'' from maxima the locations a prior assigns probabilities to, the smaller the objective function. For $x>b$, the ``closer'' the nonzero probability locations are to the maxima, the smaller the objective function. Here, ``further away'' and ``closer'' are in terms of the Klotz likelihood's gradients.
\begin{figure}[h!]
	\captionsetup[figure]{format=hang}
	\centering
	\begin{subfigure}[]{0.35\linewidth}
		\centering
			\includegraphics[width=1.0\linewidth]{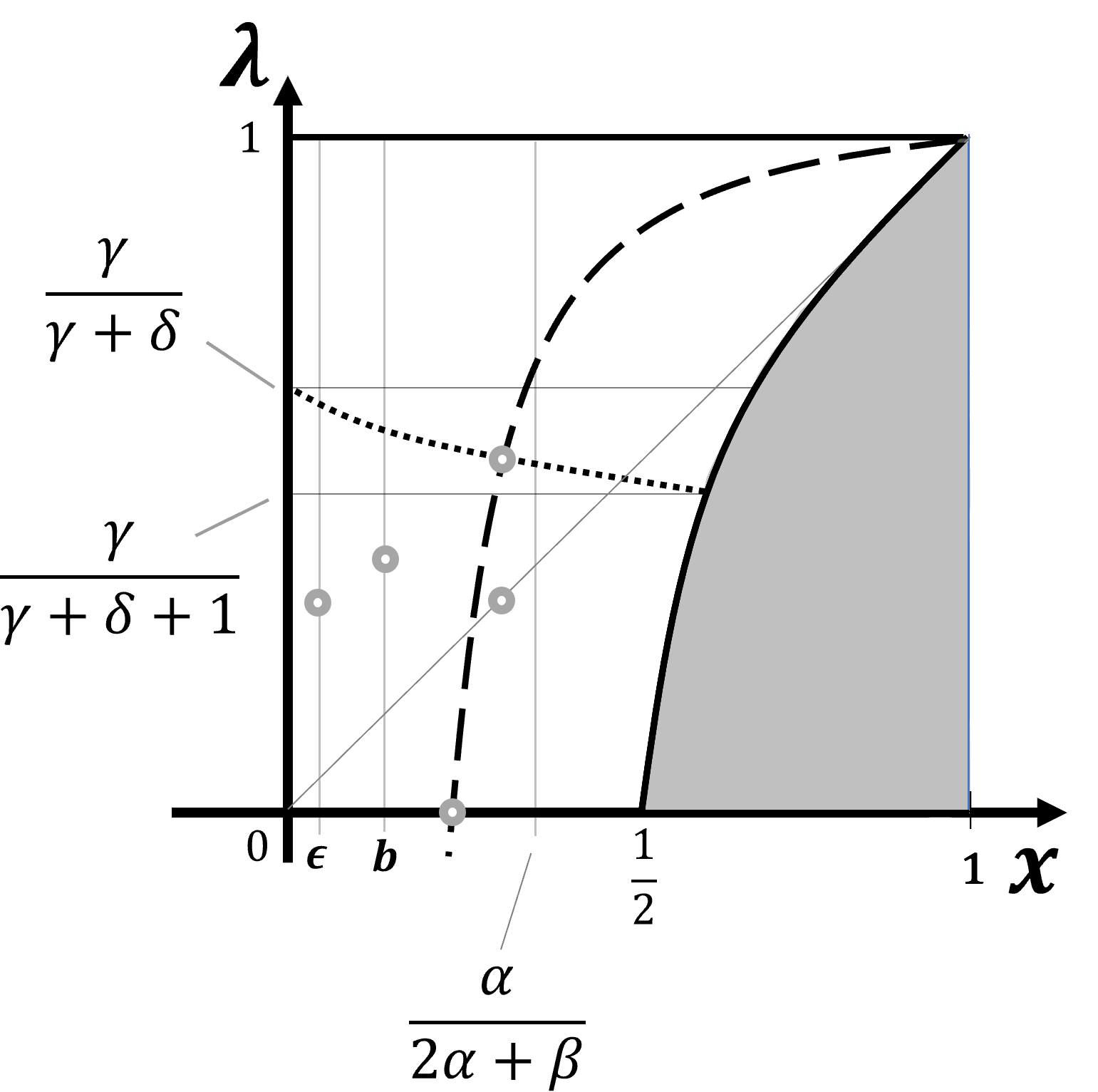}
		\caption{{\footnotesize A successful 1st execution  }}
		\label{fig:fig_RregionAllStationary Points_ffns}
	\end{subfigure}
	\begin{subfigure}[]{0.35\linewidth}
		\centering
			\includegraphics[width=1.0\linewidth]{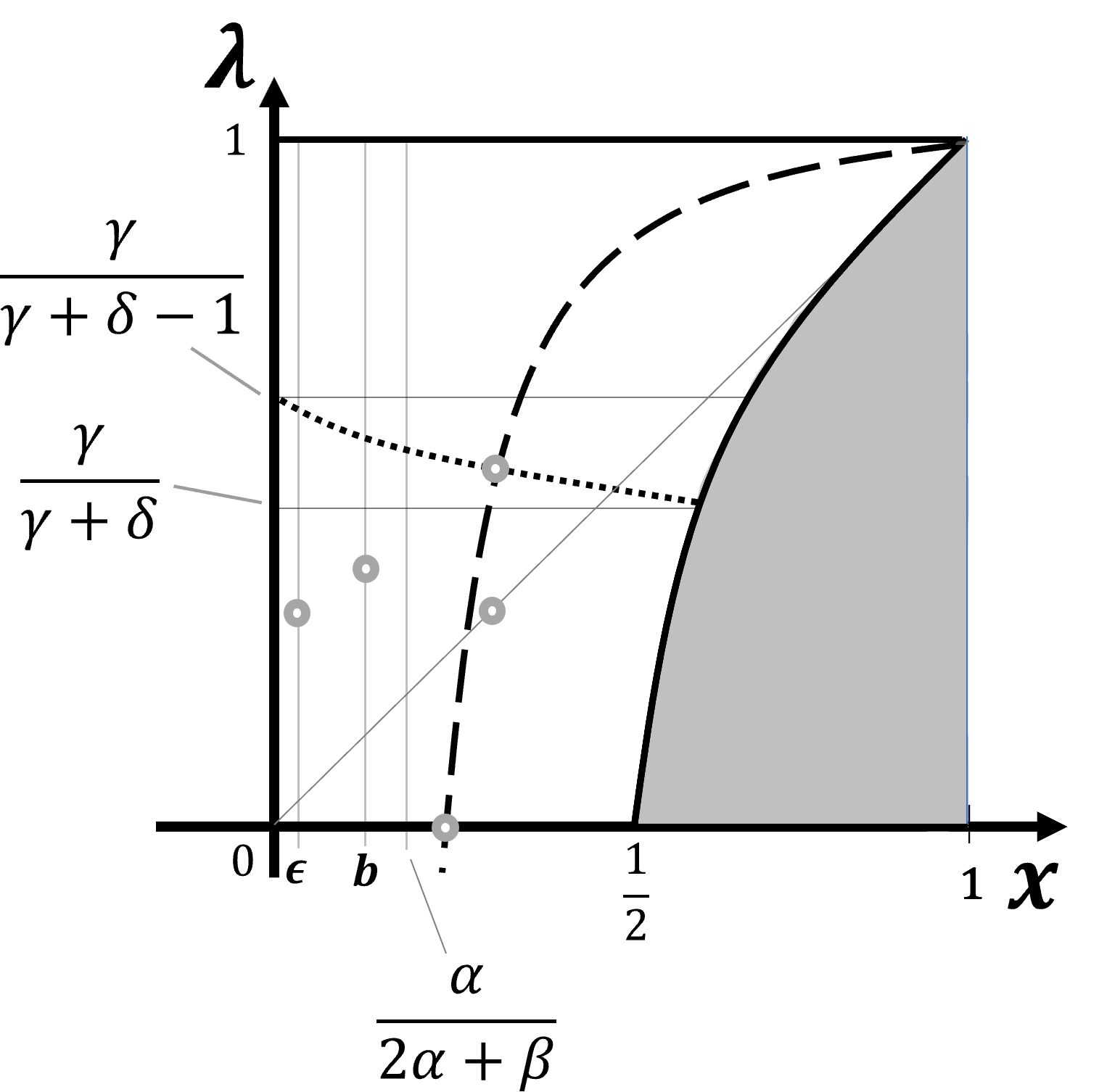}
		\caption{{\footnotesize A failed 1st execution }}
		\label{fig:fig_RregionAllStationary Points_gfns}
	\end{subfigure}
	\caption[stationary points in $\mathcal R$]{{\small Examples of locations in $\mathcal R$ (indicated by grey circles) at which local and global maxima of the Klotz likelihood occur. Here, $\alpha, \beta, \gamma, \delta \geqslant 1$. \normalsize}}
	\label{fig_allStationarypoints}
\end{figure}

That is, given any $F\in {\mathcal D}^{\prime}$, we can reassign the probabilities $\{M_i\}$ that $F$ allocates to points in $\mathcal R$, to new points suggested by the likelihood's gradients -- resulting in a new prior with a smaller objective function value. Such reassignments can be carried out indefinitely, creating a sequence of priors with an associated, monotonically decreasing sequence of objective function values. And the \emph{completeness of the real numbers} guarantees that this sequence of objective function values converge\footnote{Note that the objective function is a probability, and is therefore bounded.}. Being discrete distributions, it is also clear that these reassigned priors, themselves, converge to some limiting discrete distribution. Examples of limiting distributions converged to in this manner are illustrated in Fig.~\ref{fig_limitingDists}. The points in each subfigure indicated by black dots are the limits of the sequences of new points chosen for reassignments -- so-called \emph{limit points}.

\emph{stage 6)} So, the limiting distributions assign probabilities only to certain limit points of the 7 $\mathcal R$-subsets. The exact values of the probabilities will depend on which initial prior $F$ (with its probabilities $\{M_i\}$) was chosen to create the ``reassigned'' priors sequence. To determine those values for the ``$M_i$''s that ensure the sequence of objective function values converges to the infimum, one can systematically allocate probability masses to the limit points. We will now illustrate this, and show how the priors (when $p_l=0$) in Fig.s~\ref{fig:fig_CBIsoln_withFails_Phi2gt1minustheta_1}, \ref{fig:fig_CBIsoln_withFails_Phi2lt1minustheta_1}, \ref{fig:fig_CBIsoln_withnoconsFails_Phi1lt1minustheta_1}, \ref{fig:fig_CBIsoln_withnoconsFails_Phi1gt1minustheta_1}, \ref{fig:fig_CBIsoln_withFails_zeroconf_1}, \ref{fig:fig_CBIsoln_withnoconsFails_zeroconf_1} and \ref{fig_CBIsoln_noFails} can be obtained from the limiting distribution forms in Fig.~\ref{fig_limitingDists}. All of the priors in Fig.s~\ref{fig_CBIsoln_withFails_Phi2glt1minustheta} -- \ref{fig_CBIsoln_noFails} (when $p_l>0$) are similarly obtained.

\begin{figure}[h!]
	\captionsetup[figure]{format=hang}
	\begin{subfigure}[]{0.48\linewidth}
		\centering
			\includegraphics[width=0.65\linewidth]{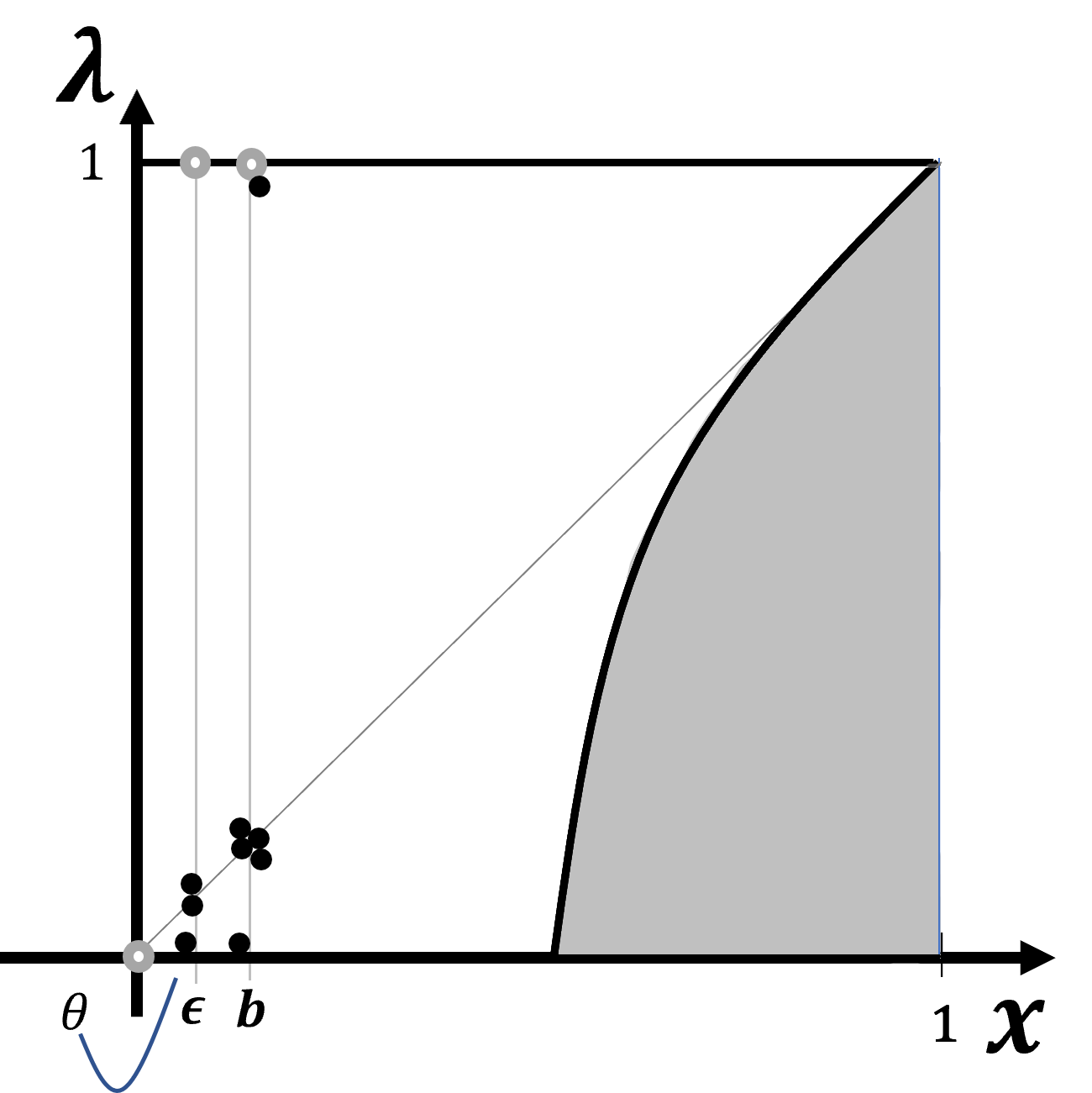}
		\caption{{\footnotesize For $x\leqslant\epsilon$, assign $\theta$ to limit point where likelihood is smallest, i.e. $(\epsilon, 0)$. }}
		\label{fig:figB8_proof_step1}
	\end{subfigure}
	\begin{subfigure}[]{0.48\linewidth}
		\centering
			\includegraphics[width=0.65\linewidth]{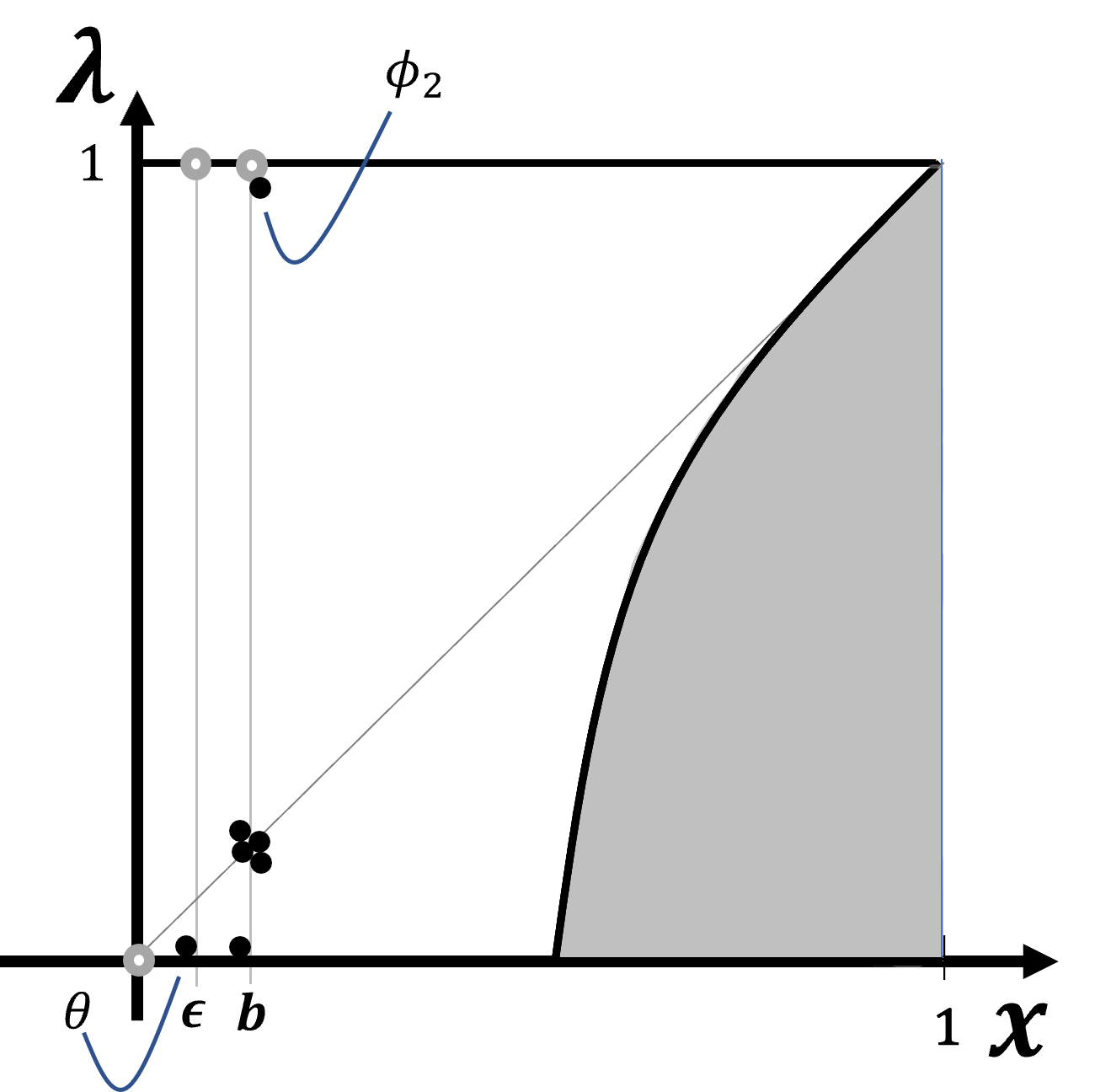}
		\caption{{\footnotesize Assign $\phi_2$ to limit point where $x\geqslant b$, i.e. $(b, 1)$.  }}
		\label{fig:figB8_proof_step2}
	\end{subfigure}
	\begin{subfigure}[]{0.48\linewidth}
		\centering
			\includegraphics[width=0.65\linewidth]{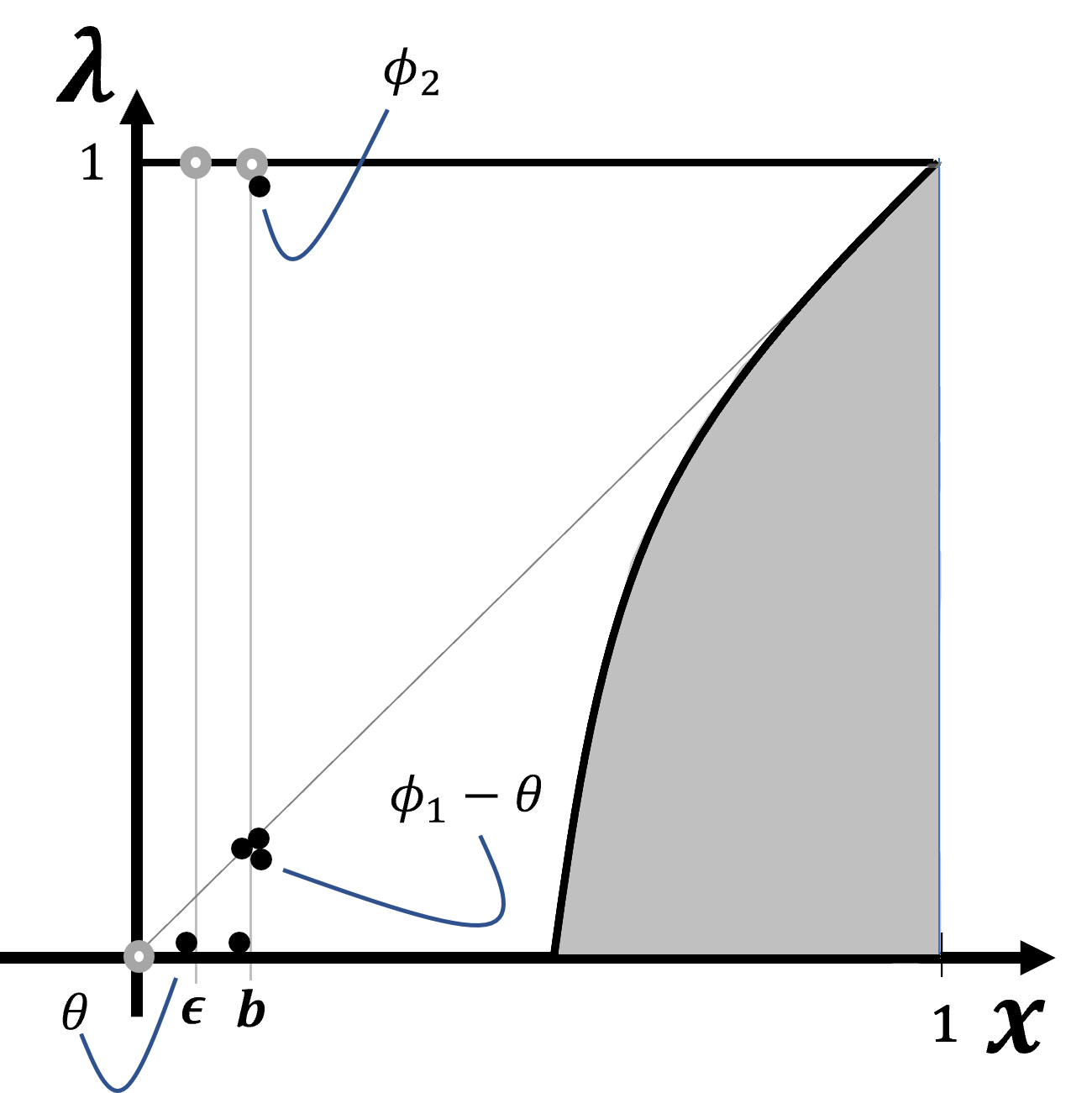}
		\caption{{\footnotesize Assign $\phi_1-\theta$ to limit point where $x\geqslant b$, i.e. $(b, b)$. }}
		\label{fig:figB8_proof_step3}
	\end{subfigure}
	\begin{subfigure}[]{0.48\linewidth}
		\centering
			\includegraphics[width=0.65\linewidth]{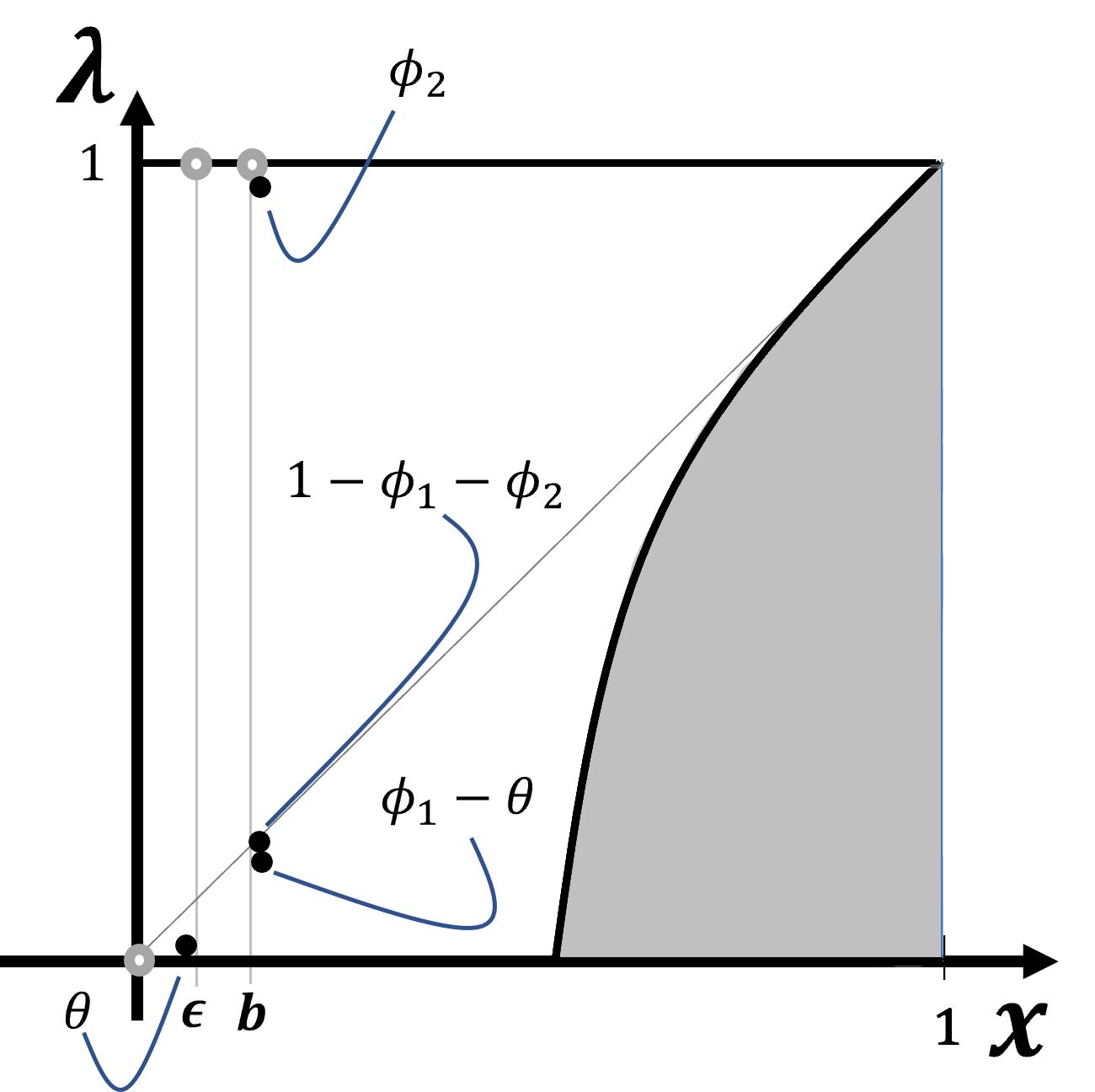}
		\caption{{\footnotesize Assign $1-\phi_1-\phi_2$ to limit point on diagonal where $x\geqslant b$, i.e. $(b, b)$. }}
		\label{fig:figB8_proof_step4}
	\end{subfigure}
	\caption[Example derivation 1 of worst-case prior]{{\small A systematic allocation of probabilities to limit points, that demonstrates how Fig.~\ref{fig:fig_CBIsoln_noFails_Phi1gtTheta} is obtained from Fig.\ref{fig:fig_limitingDist_gammadeltaalpha0}.  \normalsize}}
	\label{figB8_proof}
\end{figure}

\begin{figure*}[h!]
	\captionsetup[figure]{format=hang}
	\begin{subfigure}[]{0.33\linewidth}
		\centering
			\includegraphics[width=0.8\linewidth]{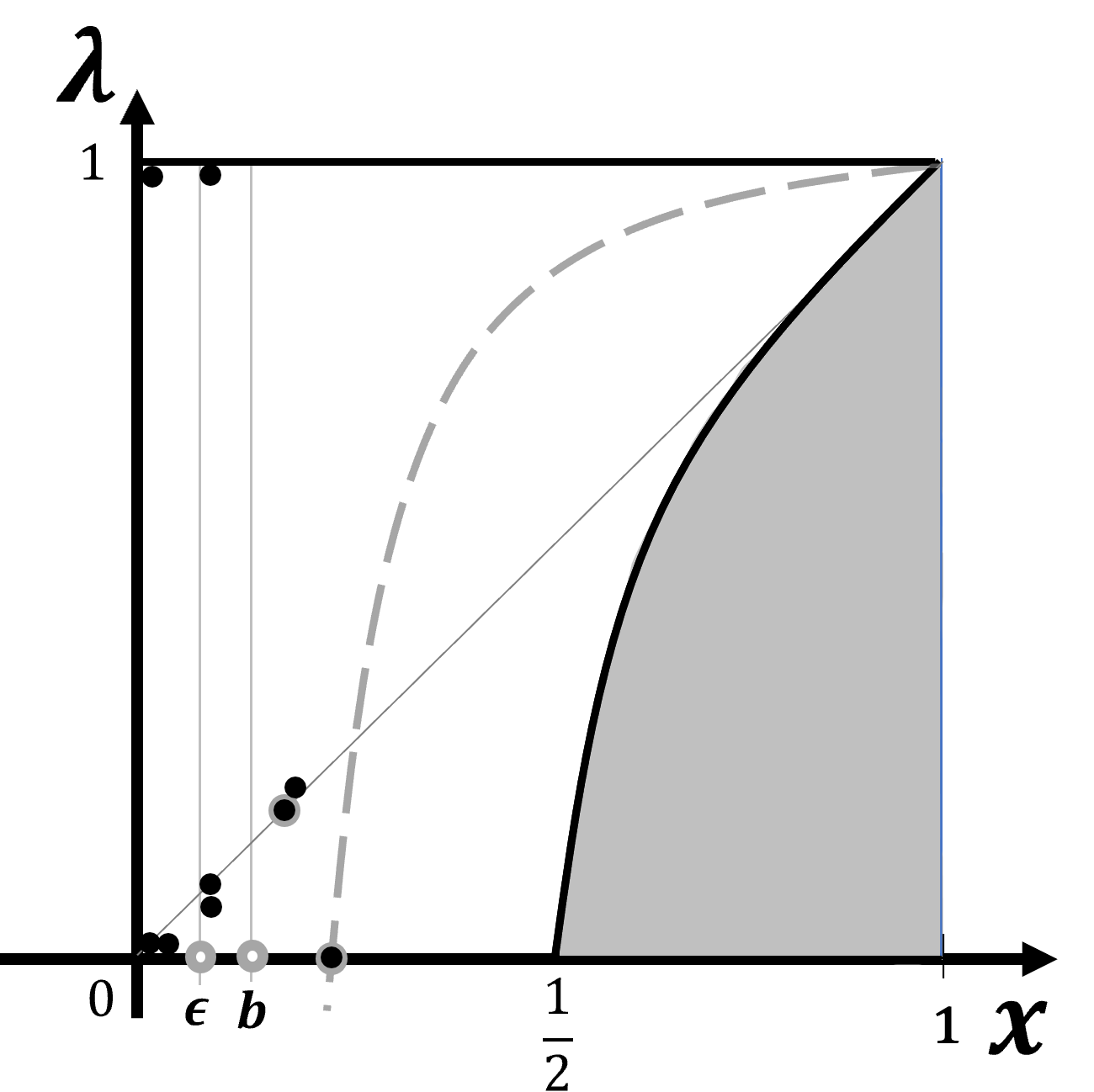}
		\caption{{\footnotesize No consecutive failures (so $\gamma=0$ and $\alpha,\beta,\delta\geqslant 1$) }}
		\label{fig:fig_limitingDist_gamma0}
	\end{subfigure}
	\begin{subfigure}[]{0.33\linewidth}
		\centering
			\includegraphics[width=0.8\linewidth]{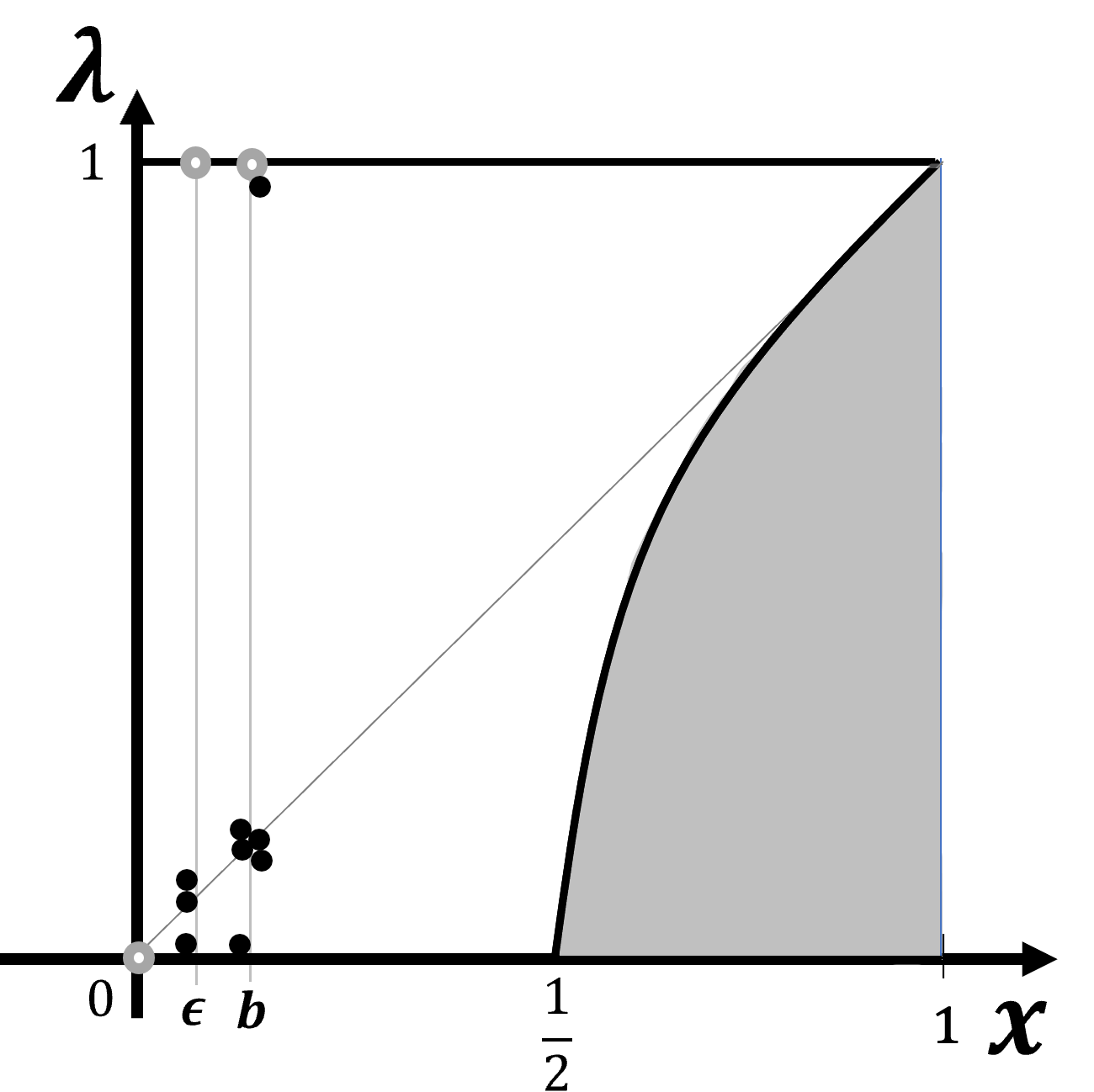}
		\caption{{\footnotesize No failures (so $\alpha, \gamma, \delta =0$)  }}
		\label{fig:fig_limitingDist_gammadeltaalpha0}
	\end{subfigure}
	\centering
	\begin{subfigure}[]{0.33\linewidth}
		\centering
			\includegraphics[width=0.8\linewidth]{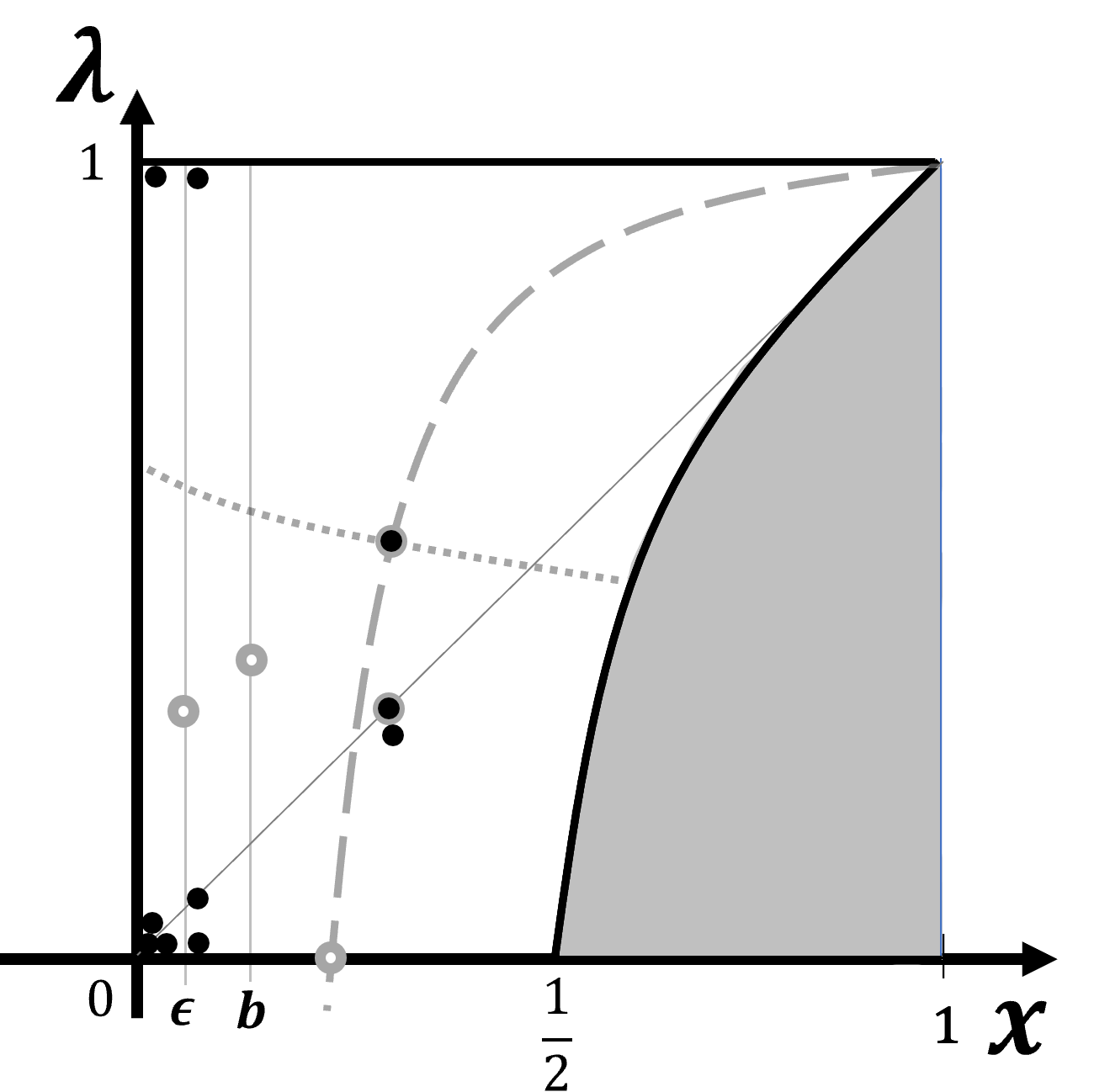}
		\caption{{\footnotesize No restrictions on failures (so $\alpha,\beta,\gamma, \delta\geqslant 1$) }}
		\label{fig:fig_limitingDist_betagammadeltaalphaPositive}
	\end{subfigure}
	\caption[Example limiting distributions]{{\small Examples of 3 limiting distributions for sequences of prior distributions (in ${\mathcal D}^{\prime}$) that give progressively smaller posterior confidence in the failure-rate bound $b$. These distributions must allocate mass only at certain limit points of each $\mathcal R$-subset, as indicated by the black circles. Some relevant stationary points in $\mathcal R$ are also indicated as grey circles.  \normalsize}}
	\label{fig_limitingDists}
\end{figure*}

\begin{figure}[htbp!]
	\captionsetup[figure]{format=hang}
	\centering
	\begin{subfigure}[]{0.3\linewidth}
		\centering
			\includegraphics[width=1.0\linewidth]{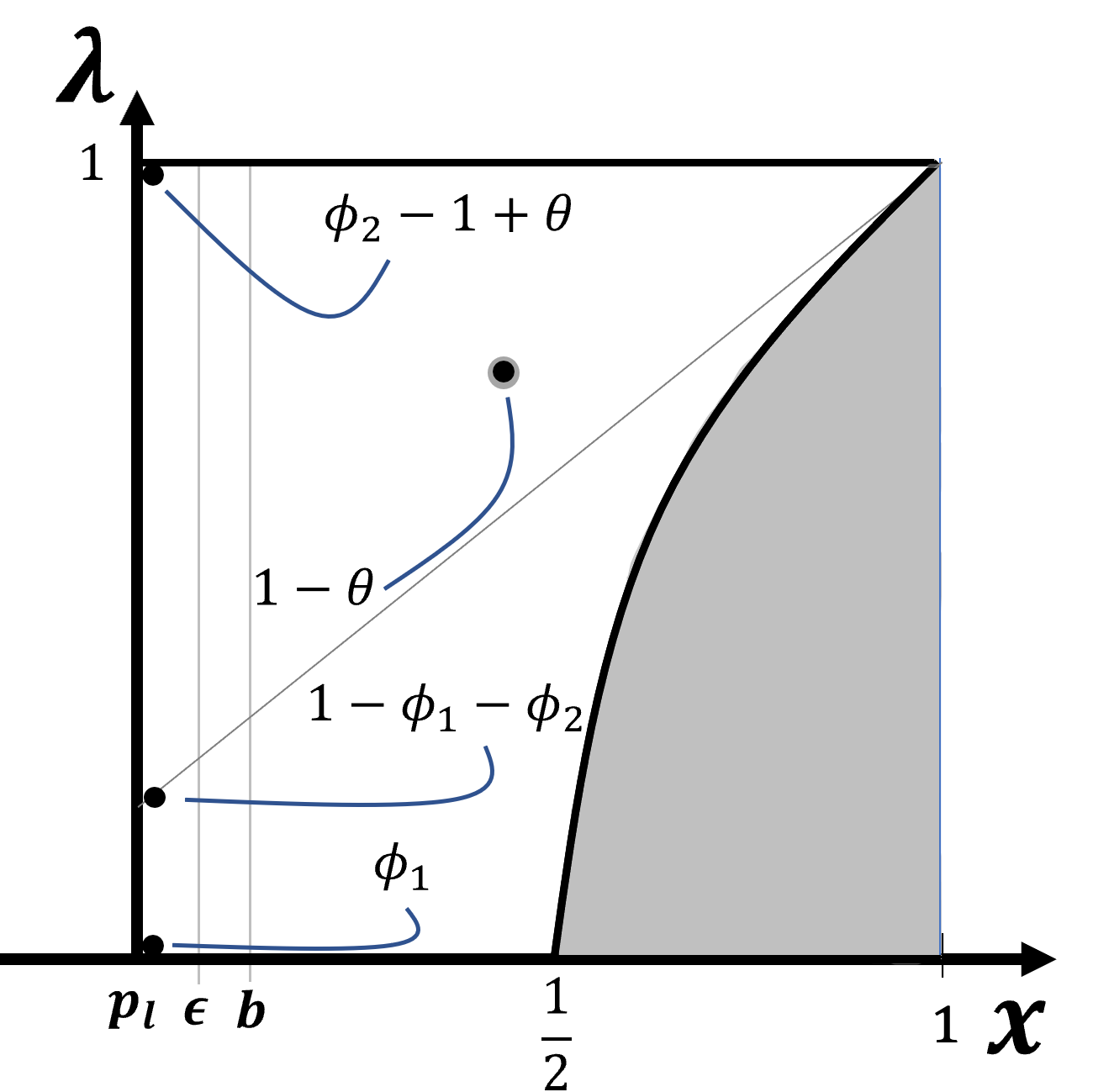}
		\caption{{\footnotesize $\phi_2\geqslant 1-\theta$ }}
		\label{fig:fig_CBIsoln_withFails_Phi2gt1minustheta_1}
	\end{subfigure}
	\begin{subfigure}[]{0.3\linewidth}
		\centering
			\includegraphics[width=1.0\linewidth]{Images/fig_CBIsoln_withFails_Phi2lt1minustheta_1}
		\caption{{\footnotesize $\phi_2\leqslant 1-\theta\ $ and $\ \phi_1\leqslant \theta$ }}
		\label{fig:fig_CBIsoln_withFails_Phi2lt1minustheta_1}
	\end{subfigure}
	\begin{subfigure}[]{0.3\linewidth}
		\centering
			\includegraphics[width=1.0\linewidth]{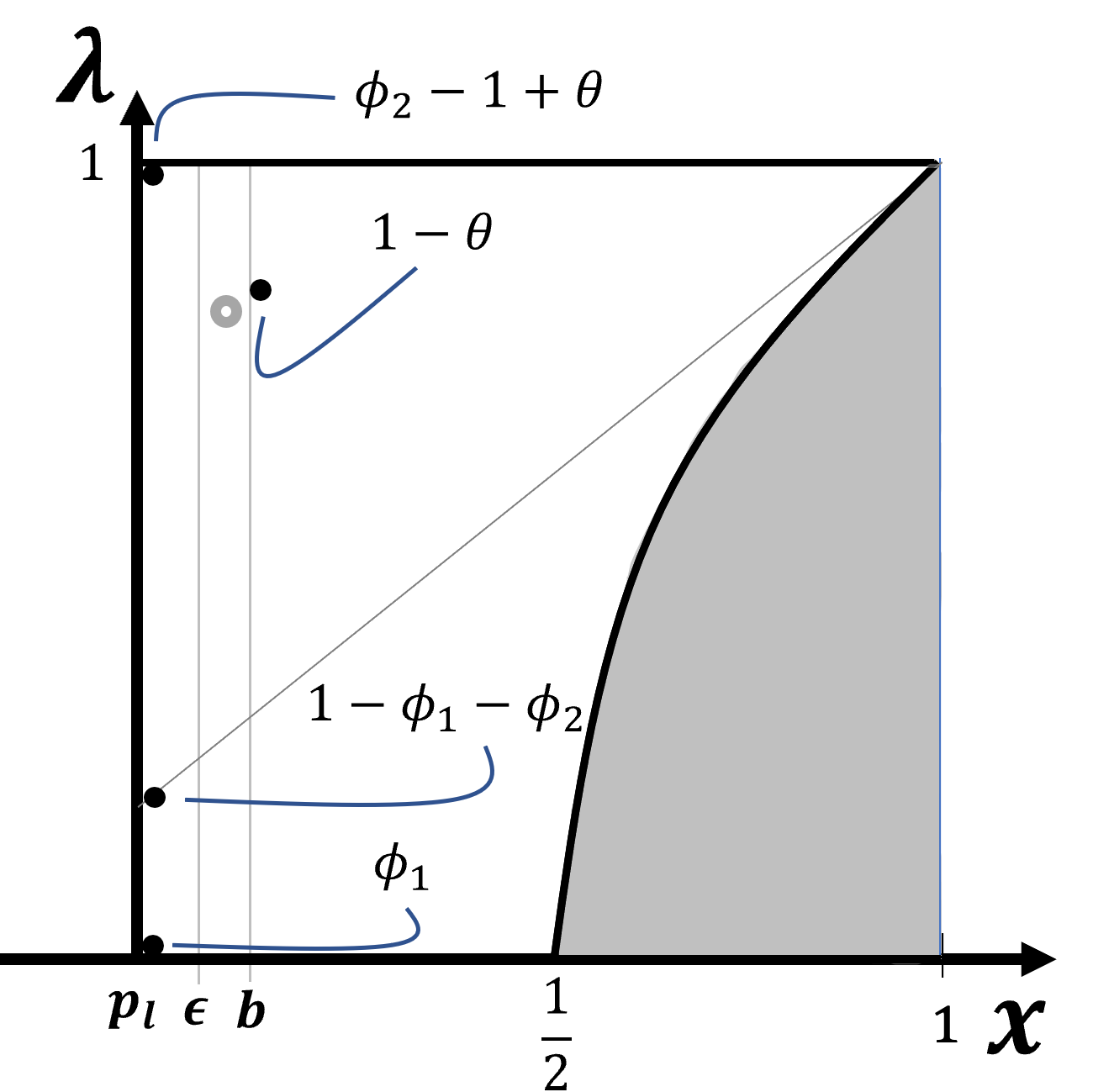}
		\caption{{\footnotesize $\phi_2\geqslant 1-\theta$  }}
		\label{fig:fig_CBIsoln_withFails_Phi2gt1minustheta_2}
	\end{subfigure}
	\begin{subfigure}[]{0.3\linewidth}
		\centering
			\includegraphics[width=1.0\linewidth]{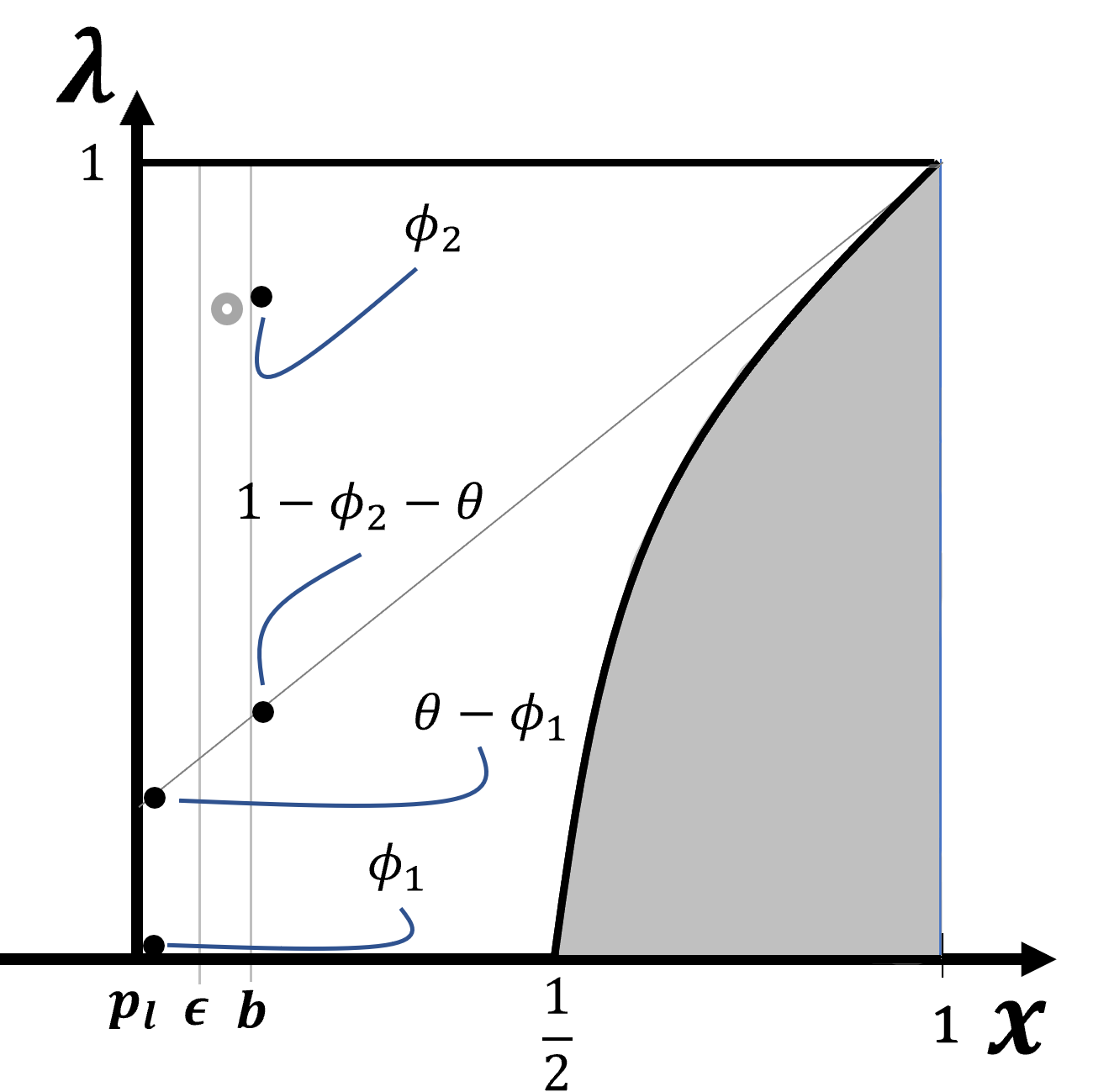}
		\caption{{\footnotesize $\phi_2\leqslant 1-\theta\ $ and $\ \phi_1\leqslant \theta$ }}
		\label{fig:fig_CBIsoln_withFails_Phi2lt1minustheta_2}
	\end{subfigure}
	\begin{subfigure}[]{0.3\linewidth}
		\centering
			\includegraphics[width=1.0\linewidth]{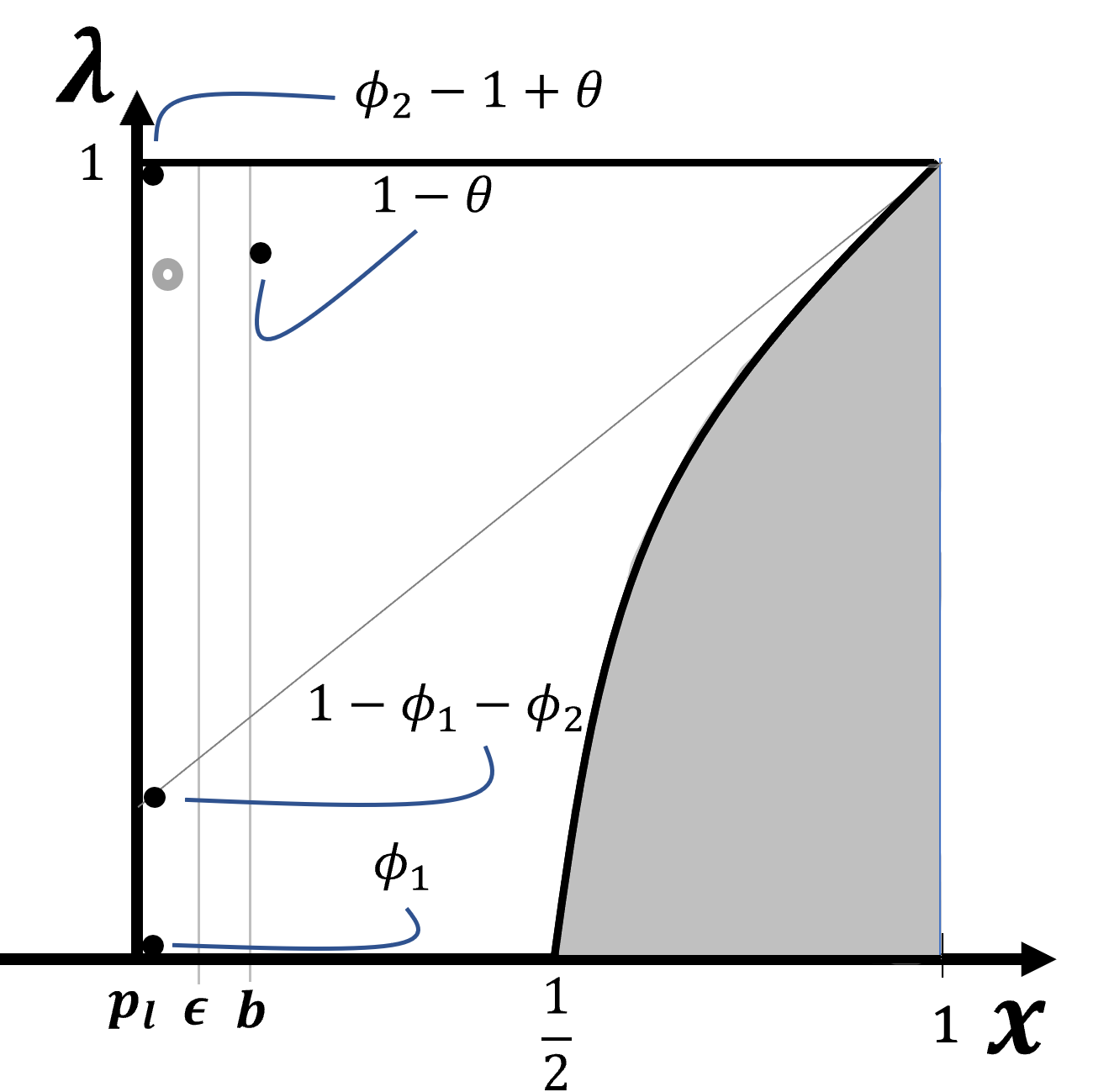}
		\caption{{\footnotesize $\ L(p_l, p_l ;\alpha,\ldots)$ $\leqslant$ $L(\epsilon,\epsilon;\alpha,\ldots)$ and $\ \phi_2\geqslant 1-\theta$ }}
		\label{fig:fig_CBIsoln_withFails_Phi2gt1minustheta_3}
	\end{subfigure}
	\begin{subfigure}[]{0.3\linewidth}
		\centering
			\includegraphics[width=1.0\linewidth]{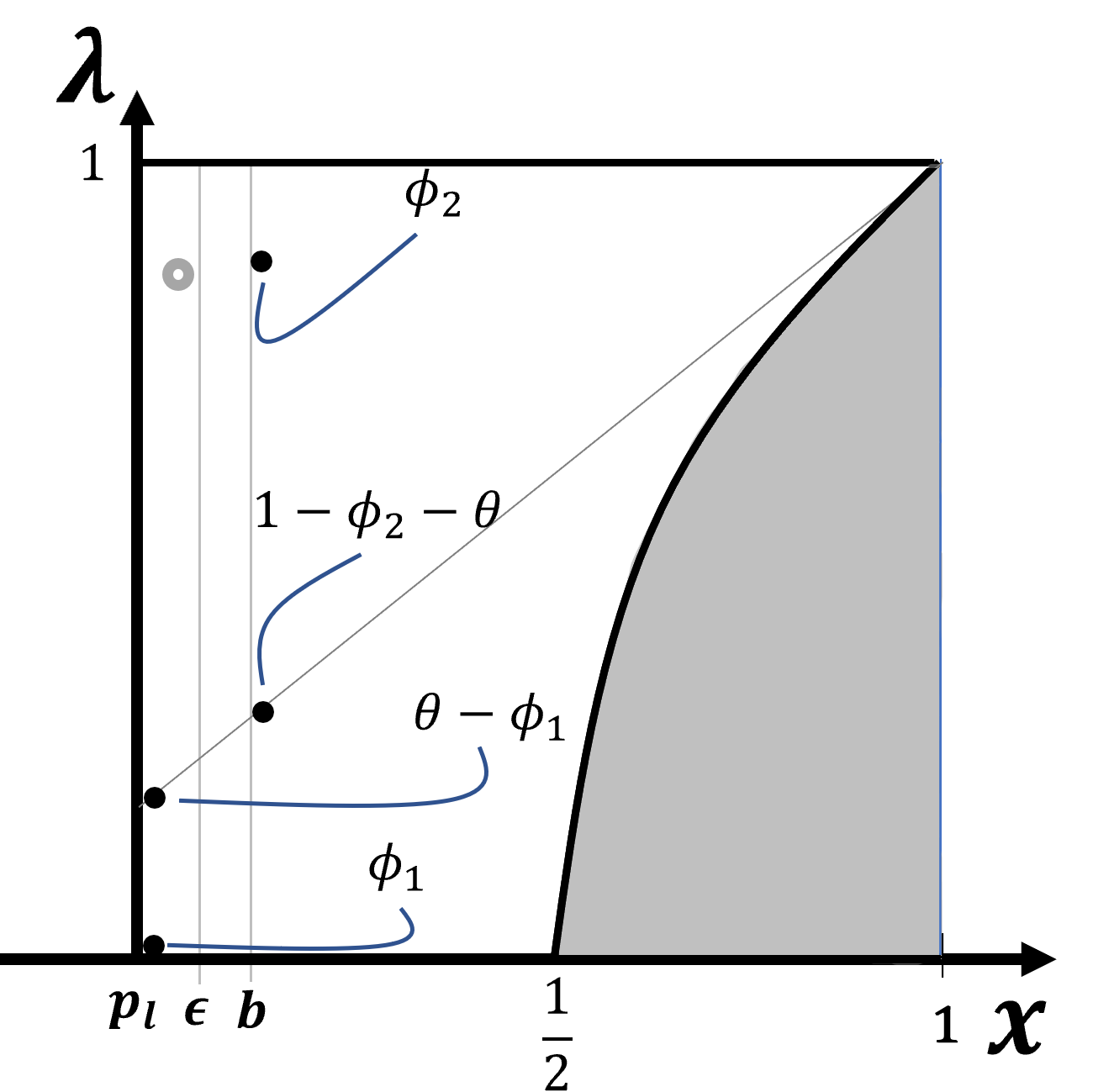}
		\caption{{\footnotesize $\ L(p_l, p_l ;\alpha,\ldots)$ $\leqslant$ $L(\epsilon,\epsilon;\alpha,\ldots)\ $ and $\ \phi_1\leqslant \theta\ $ and $\ \phi_2\leqslant 1-\theta$  }}
		\label{fig:fig_CBIsoln_withFails_Phi2lt1minustheta_3}
	\end{subfigure}
	\begin{subfigure}[]{0.3\linewidth}
		\centering
			\includegraphics[width=1.0\linewidth]{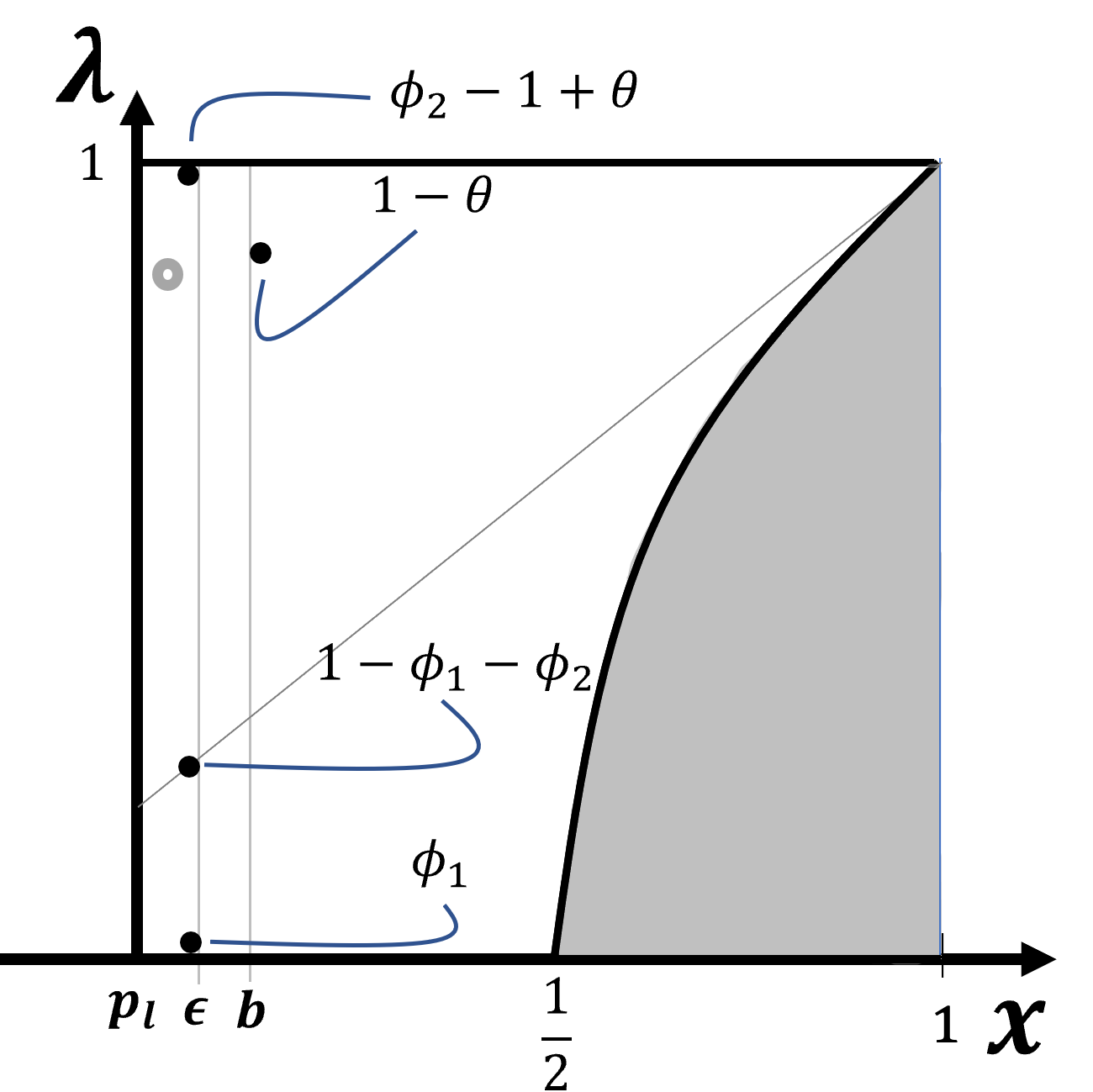}
		\caption{{\footnotesize $\ L(p_l, p_l ;\alpha,\ldots)$ $\geqslant$ $L(\epsilon,\epsilon;\alpha,\ldots)\ $ and $\ \phi_2\geqslant 1-\theta\ $ }}
		\label{fig:fig_CBIsoln_withFails_Phi2gt1minustheta_4}
	\end{subfigure}
	\begin{subfigure}[]{0.3\linewidth}
		\centering
			\includegraphics[width=1.0\linewidth]{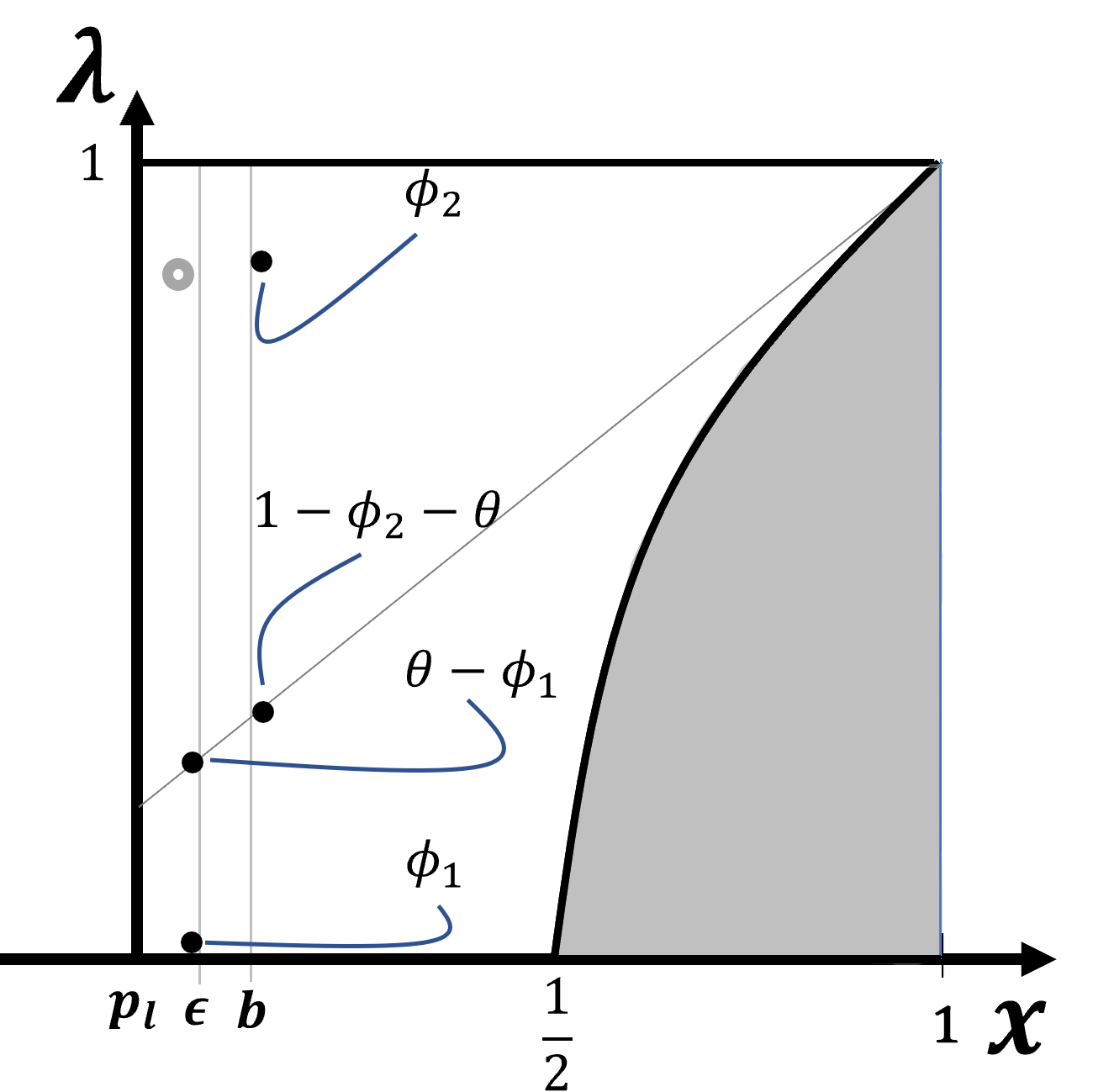}
		\caption{{\footnotesize $\ L(p_l, p_l ;\alpha,\ldots)$ $\geqslant$ $L(\epsilon,\epsilon;\alpha,\ldots)\ $ and $\ \phi_1\leqslant \theta\ $ and $\ \phi_2\leqslant 1-\theta$ }}
		\label{fig:fig_CBIsoln_withFails_Phi2lt1minustheta_4}
	\end{subfigure}
	\caption[CBI priors: with failures]{{\footnotesize Worst case prior distributions that solve the optimisation problem in Theorem \ref{theorem_CBI_withFailures_and_pl} when consecutive failures are observed (i.e. $r>0$). It's important to note the following: the precise locations of the ``black dots'' for each such distribution are determined by 1) the values of $\alpha, \beta, \gamma$ and $\delta$, 2) whether the first execution is a success or failure, and 3) the indicated parameter ranges in each subfigure. The location $(x^\ast, \lambda^\ast)$ of the global maximum for the Klotz likelihood is indicated by the grey circle. The $4$ priors, illustrated in subfigures \ref{fig:fig_CBIsoln_withFails_Phi2gt1minustheta_1},  \ref{fig:fig_CBIsoln_withFails_Phi2gt1minustheta_2}, \ref{fig:fig_CBIsoln_withFails_Phi2gt1minustheta_3} and \ref{fig:fig_CBIsoln_withFails_Phi2gt1minustheta_4}, are solutions when $\phi_2\geqslant 1-\theta$. While the priors in \ref{fig:fig_CBIsoln_withFails_Phi2lt1minustheta_1},  \ref{fig:fig_CBIsoln_withFails_Phi2lt1minustheta_2}, \ref{fig:fig_CBIsoln_withFails_Phi2lt1minustheta_3} and  \ref{fig:fig_CBIsoln_withFails_Phi2lt1minustheta_4} solve the problem when $\phi_2\leqslant 1-\theta\ $ and $\ \phi_1\leqslant \theta$. These solutions assume $\alpha, \beta, \gamma, \delta > 0$.  \normalsize}}  
	\label{fig_CBIsoln_withFails_Phi2glt1minustheta}
\end{figure}

\begin{figure}[htbp!]
	\captionsetup[figure]{format=hang}
	\centering
	\begin{subfigure}[]{0.3\linewidth}
		\centering
			\includegraphics[width=1.0\linewidth]{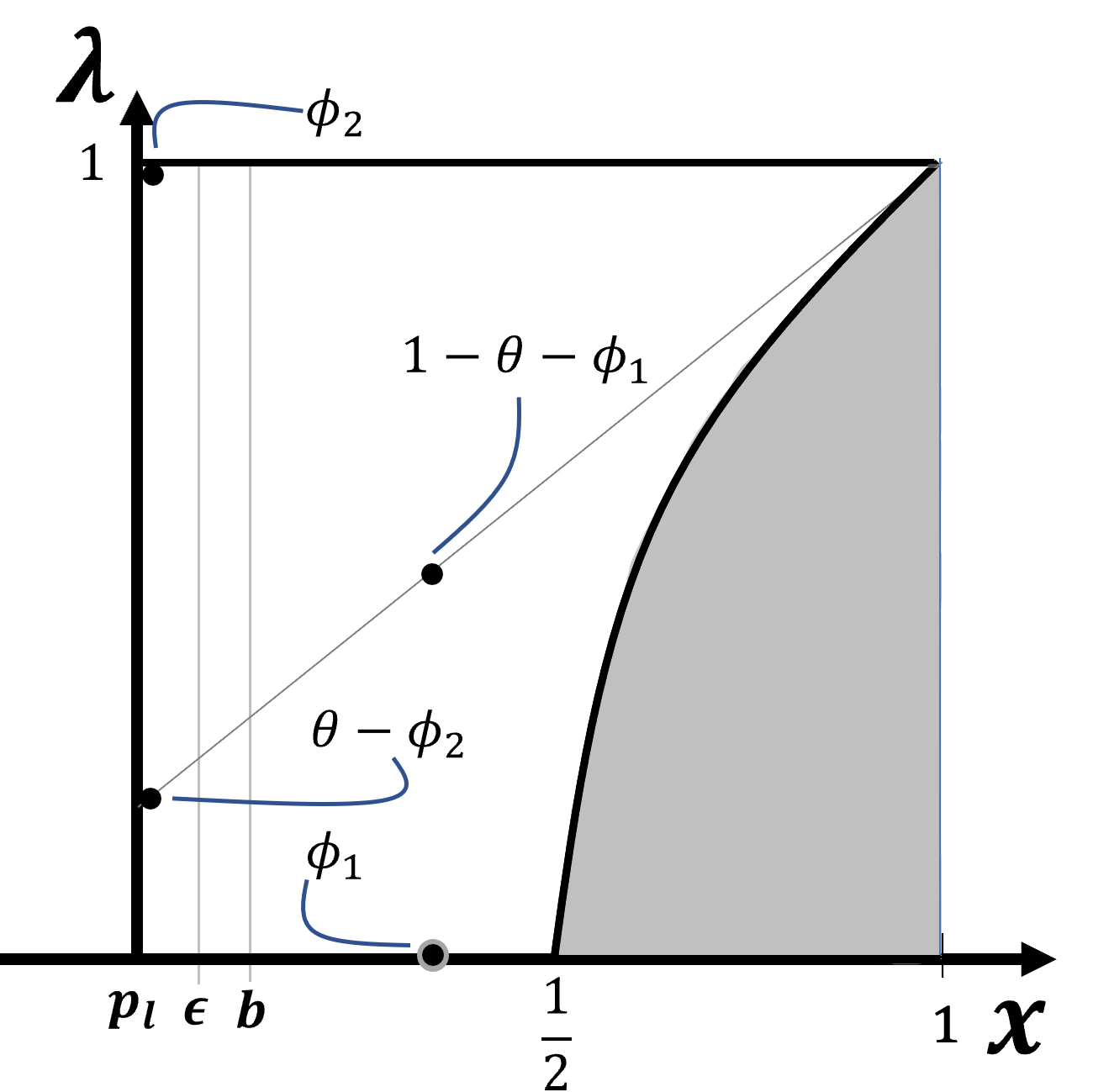}
		\caption{{\footnotesize $\phi_1\leqslant 1-\theta\ $ and $\ \phi_2\leqslant \theta\ $ }}
		\label{fig:fig_CBIsoln_withnoconsFails_Phi1lt1minustheta_1}
	\end{subfigure}
	\begin{subfigure}[]{0.3\linewidth}
		\centering
			\includegraphics[width=1.0\linewidth]{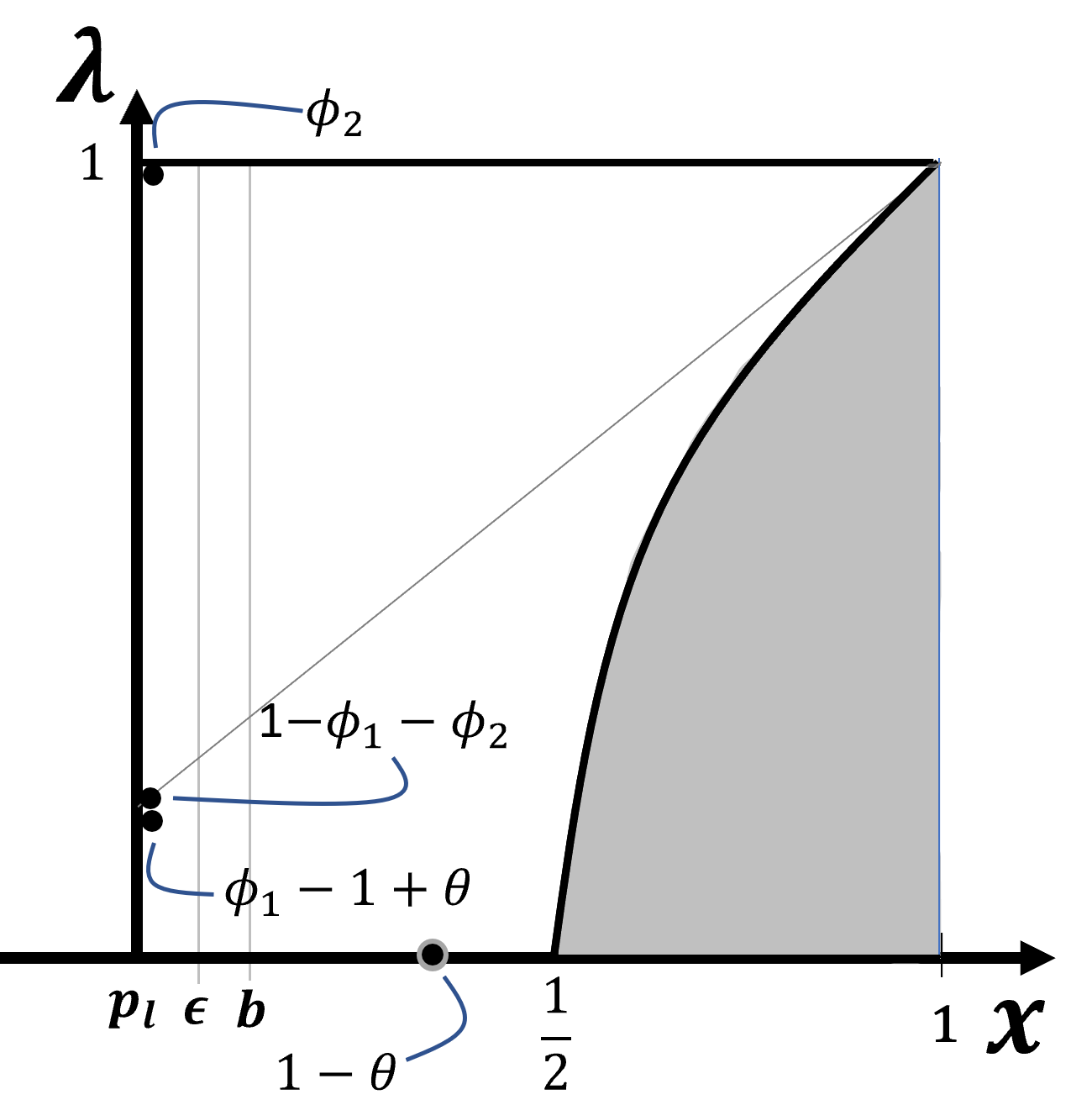}
		\caption{{\footnotesize $\phi_1\geqslant 1-\theta\ $}}
		\label{fig:fig_CBIsoln_withnoconsFails_Phi1gt1minustheta_1}
	\end{subfigure}
	\begin{subfigure}[]{0.3\linewidth}
		\centering
			\includegraphics[width=1.0\linewidth]{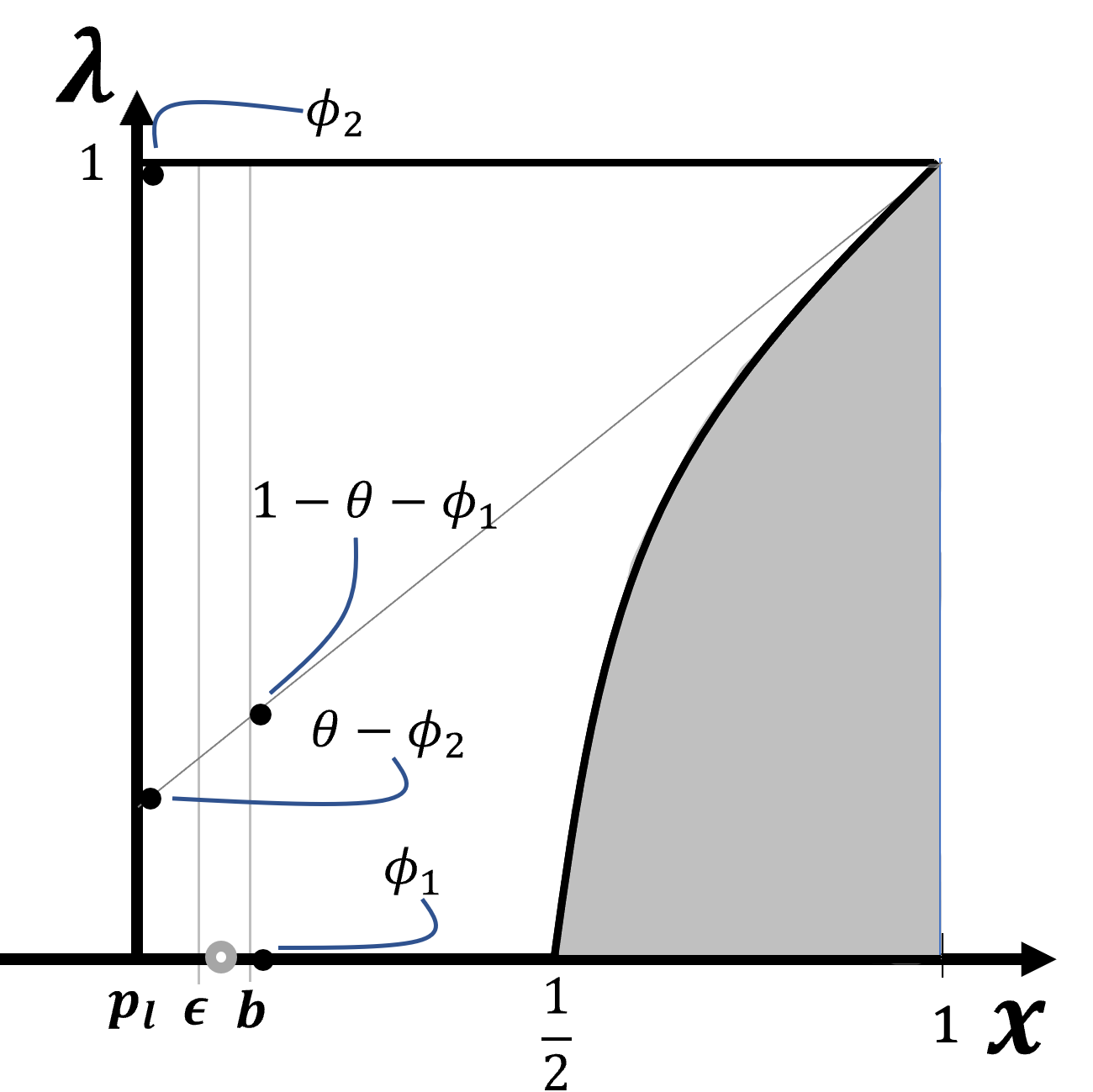}
		\caption{{\footnotesize $\phi_1\leqslant 1-\theta\ $ and $\ \phi_2\leqslant \theta$  }}
		\label{fig:fig_CBIsoln_withnoconsFails_Phi1lt1minustheta_2}
	\end{subfigure}
	\begin{subfigure}[]{0.3\linewidth}
		\centering
			\includegraphics[width=1.0\linewidth]{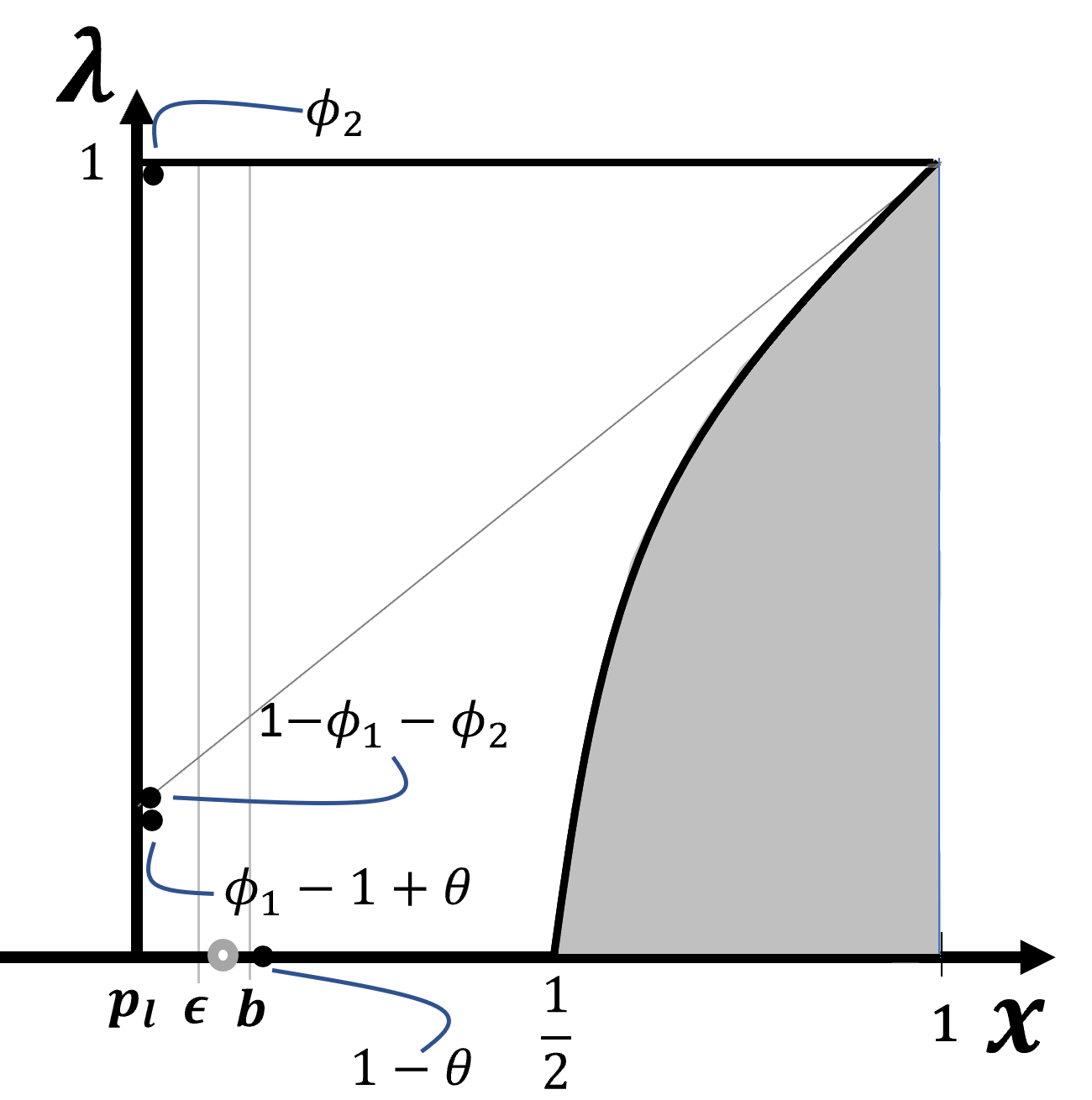}
		\caption{{\footnotesize $\phi_1\geqslant 1-\theta$ }}
		\label{fig:fig_CBIsoln_withnoconsFails_Phi1gt1minustheta_2}
	\end{subfigure}
	\begin{subfigure}[]{0.3\linewidth}
		\centering
			\includegraphics[width=1.0\linewidth]{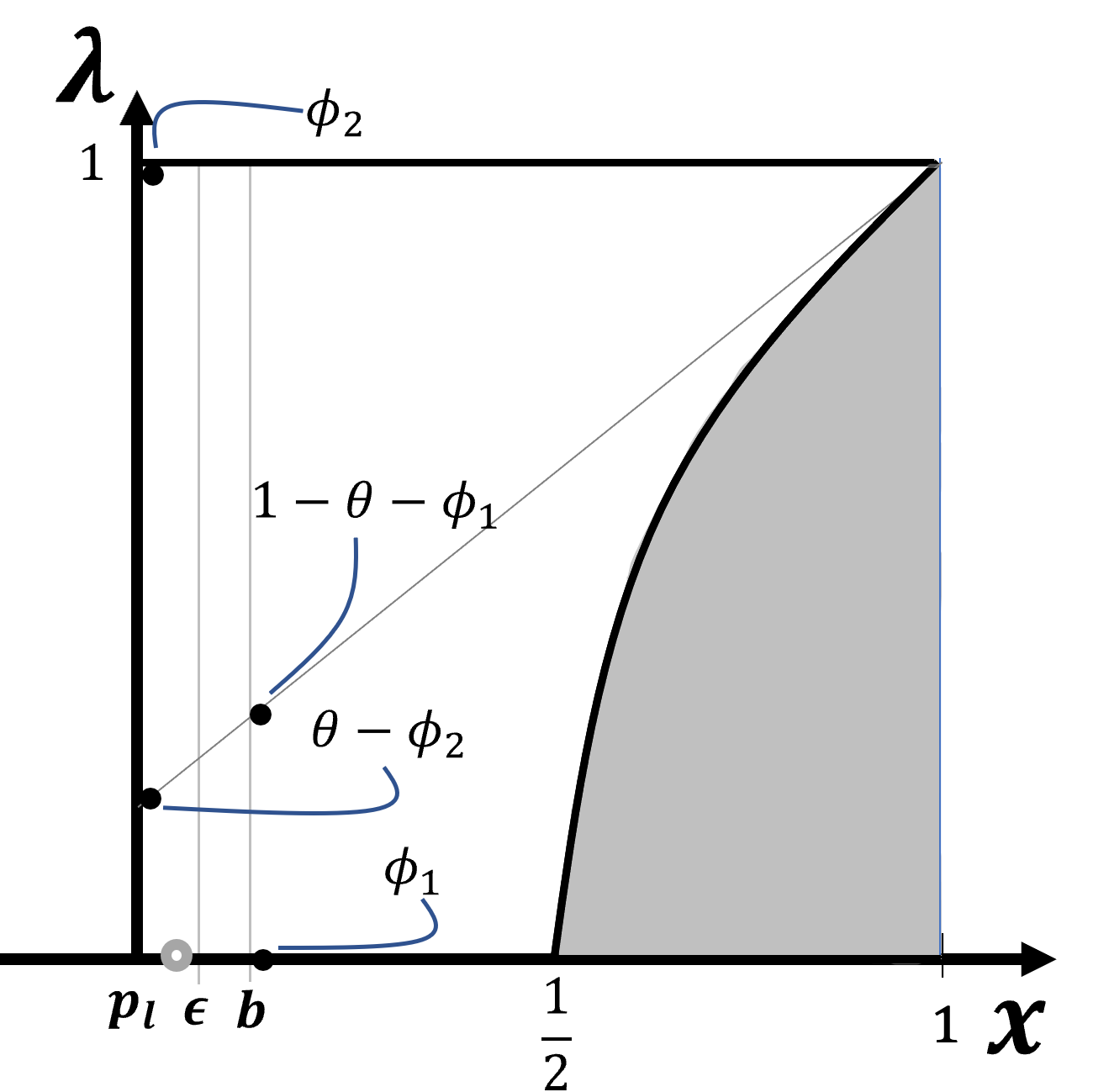}
		\caption{{\footnotesize $\ L(p_l, p_l ;\alpha,\ldots)$ $\leqslant$ $L(\epsilon,\epsilon;\alpha,\ldots)$ and $\ \phi_1\leqslant 1-\theta\ $ and $\ \phi_2\leqslant \theta\ $ }}
		\label{fig:fig_CBIsoln_withnoconsFails_Phi1lt1minustheta_3}
	\end{subfigure}
	\begin{subfigure}[]{0.3\linewidth}
		\centering
			\includegraphics[width=1.0\linewidth]{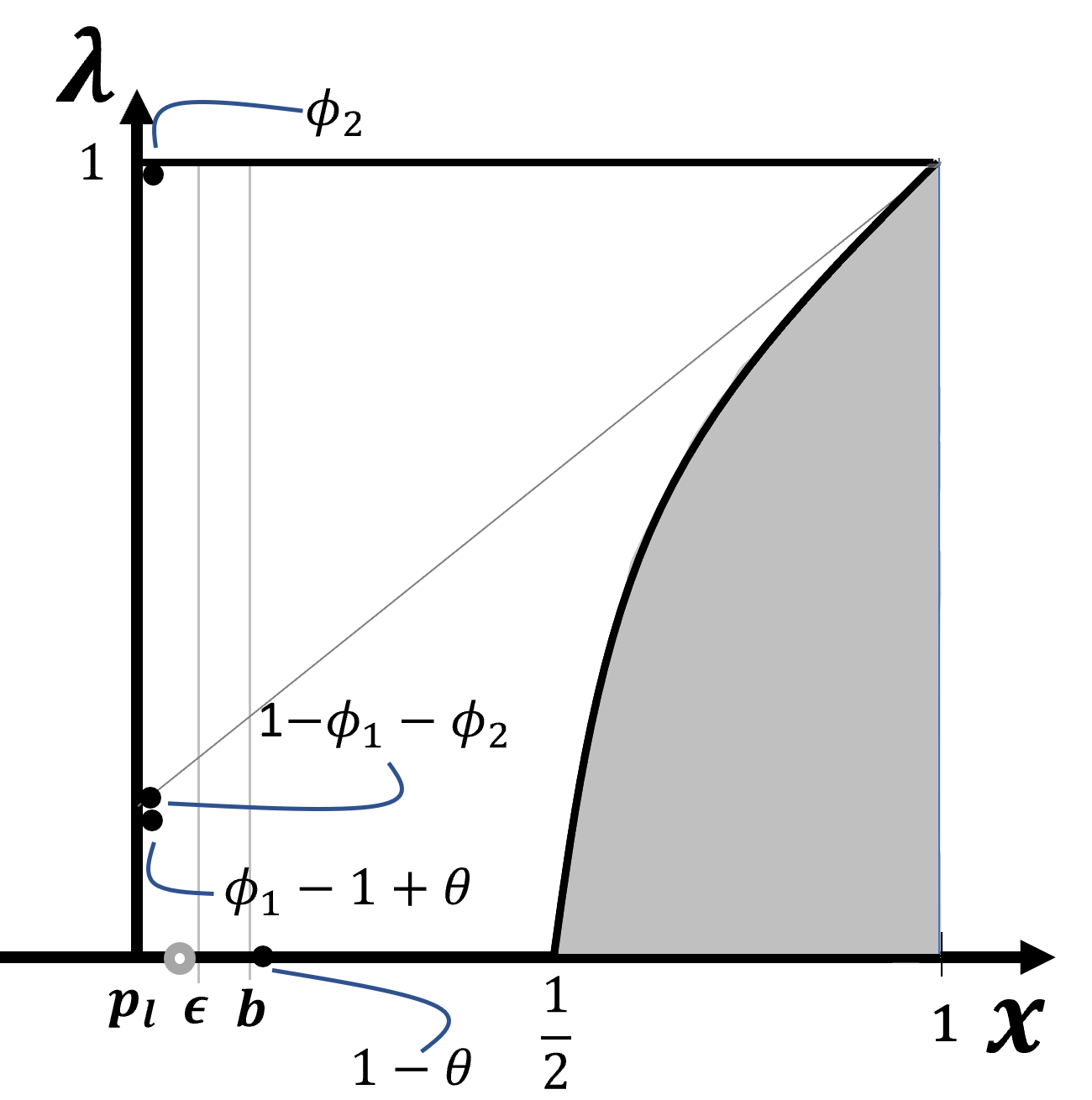}
		\caption{{\footnotesize $\ L(p_l, p_l ;\alpha,\ldots)$ $\leqslant$ $L(\epsilon,\epsilon;\alpha,\ldots)\ $ and $\ \phi_1\geqslant 1-\theta$  }}
		\label{fig:fig_CBIsoln_withnoconsFails_Phi1gt1minustheta_3}
	\end{subfigure}
	\begin{subfigure}[]{0.3\linewidth}
		\centering
			\includegraphics[width=1.0\linewidth]{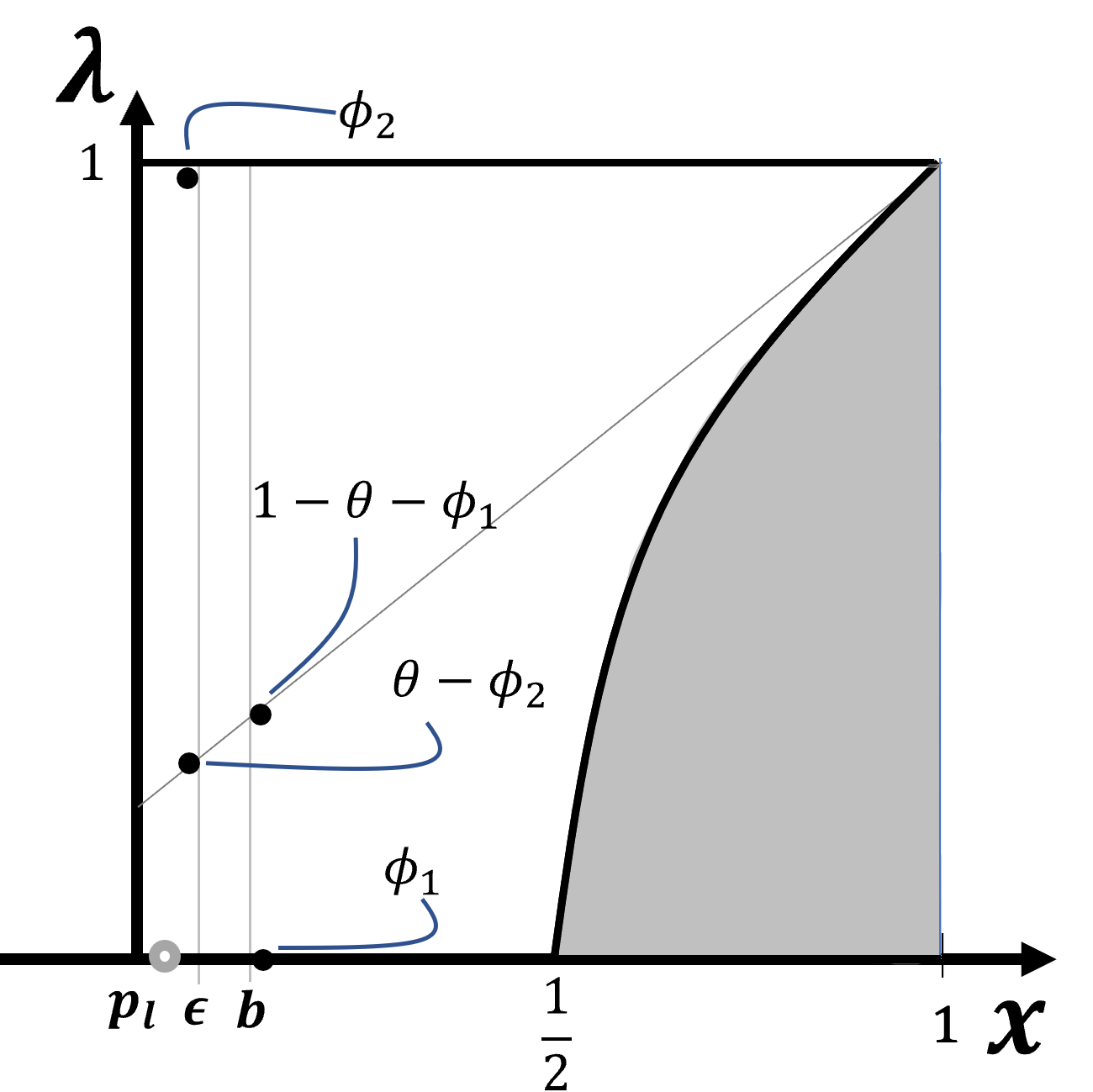}
		\caption{{\footnotesize $\ L(p_l, p_l ;\alpha,\ldots)$ $\geqslant$ $L(\epsilon,\epsilon;\alpha,\ldots)\ $ and $\ \phi_1\leqslant 1-\theta\ $ and $\ \phi_2\leqslant \theta$ }}
		\label{fig:fig_CBIsoln_withnoconsFails_Phi1lt1minustheta_4}
	\end{subfigure}
	\begin{subfigure}[]{0.3\linewidth}
		\centering
			\includegraphics[width=1.0\linewidth]{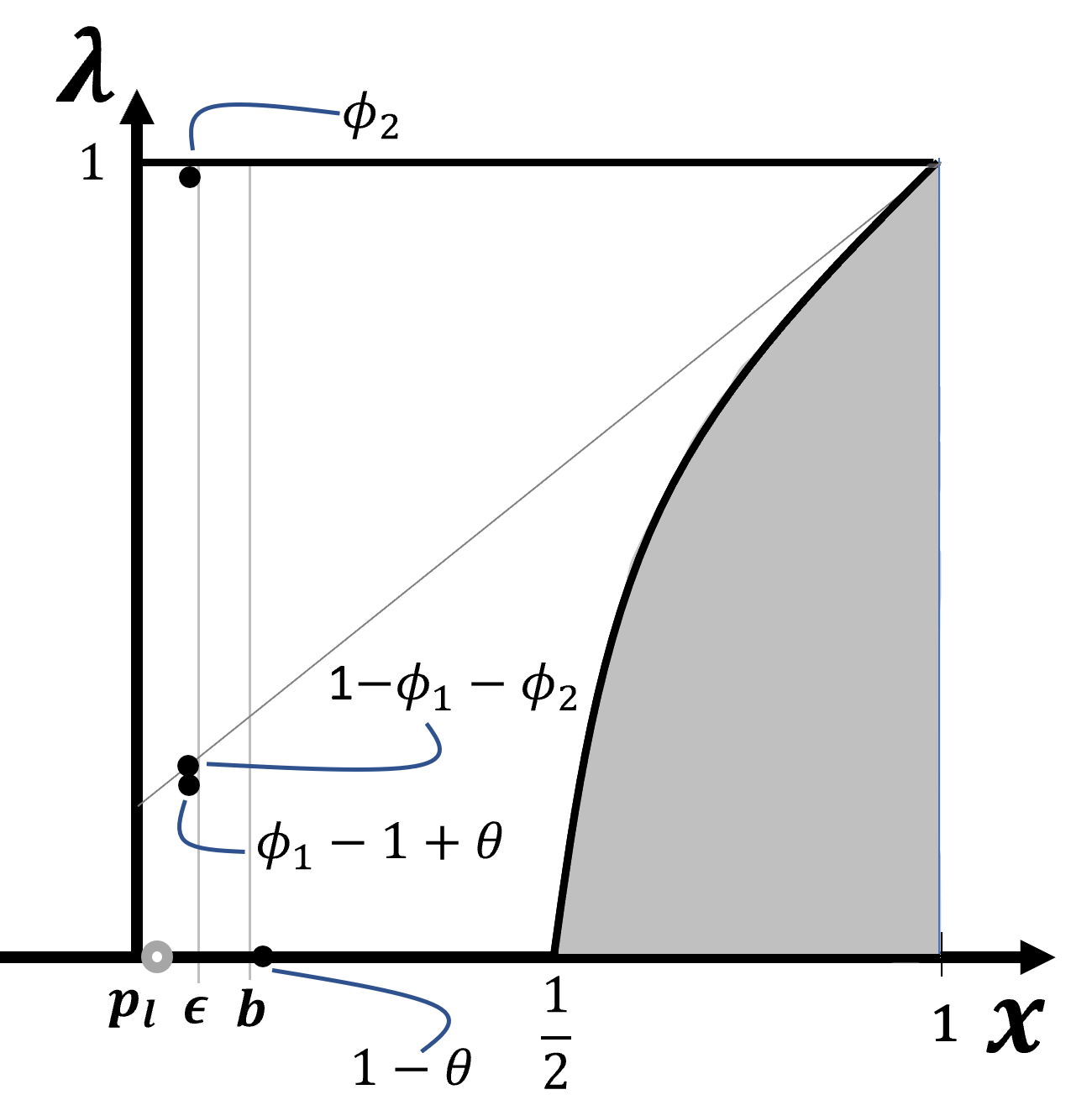}
		\caption{{\footnotesize $\ L(p_l, p_l ;\alpha,\ldots)$ $\geqslant$ $L(\epsilon,\epsilon;\alpha,\ldots)\ $ and $\ \phi_1\geqslant 1-\theta$ }}
		\label{fig:fig_CBIsoln_withnoconsFails_Phi1gt1minustheta_4}
	\end{subfigure}
	\caption[CBI priors: with failures]{{\footnotesize Worst case prior distributions that solve the optimisation problem in Theorem \ref{theorem_CBI_withFailures_and_pl} when failures are observed, without any consecutive failures (i.e. $r=0$). Each distribution's support is determined by $\alpha, \beta, \gamma, \delta$, and whether the first execution succeeds or fails. The location $(x^\ast, \lambda^\ast)$ of the global maximum for the Klotz likelihood is indicated by the grey circle. The $4$ priors, illustrated in subfigures \ref{fig:fig_CBIsoln_withnoconsFails_Phi1lt1minustheta_1},  \ref{fig:fig_CBIsoln_withnoconsFails_Phi1lt1minustheta_2}, \ref{fig:fig_CBIsoln_withnoconsFails_Phi1lt1minustheta_3} and \ref{fig:fig_CBIsoln_withnoconsFails_Phi1lt1minustheta_4},  are solutions when $\phi_1\leqslant 1-\theta\ $ and $\ \phi_1\leqslant \theta$. While the priors in \ref{fig:fig_CBIsoln_withnoconsFails_Phi1gt1minustheta_1},  \ref{fig:fig_CBIsoln_withnoconsFails_Phi1gt1minustheta_2}, \ref{fig:fig_CBIsoln_withnoconsFails_Phi1gt1minustheta_3} and  \ref{fig:fig_CBIsoln_withnoconsFails_Phi1gt1minustheta_4} solve the problem when $\phi_1\geqslant 1-\theta$. These solutions assume $\alpha, \beta, \gamma, \delta > 0$.  \normalsize}}  
	\label{fig_CBIsoln_withnoconsFails_Phi1glt1minustheta}
\end{figure}

\begin{figure}[h!]
	\captionsetup[figure]{format=hang}
	\centering
	\begin{subfigure}[]{0.3\linewidth}
		\centering
			\includegraphics[width=1.0\linewidth]{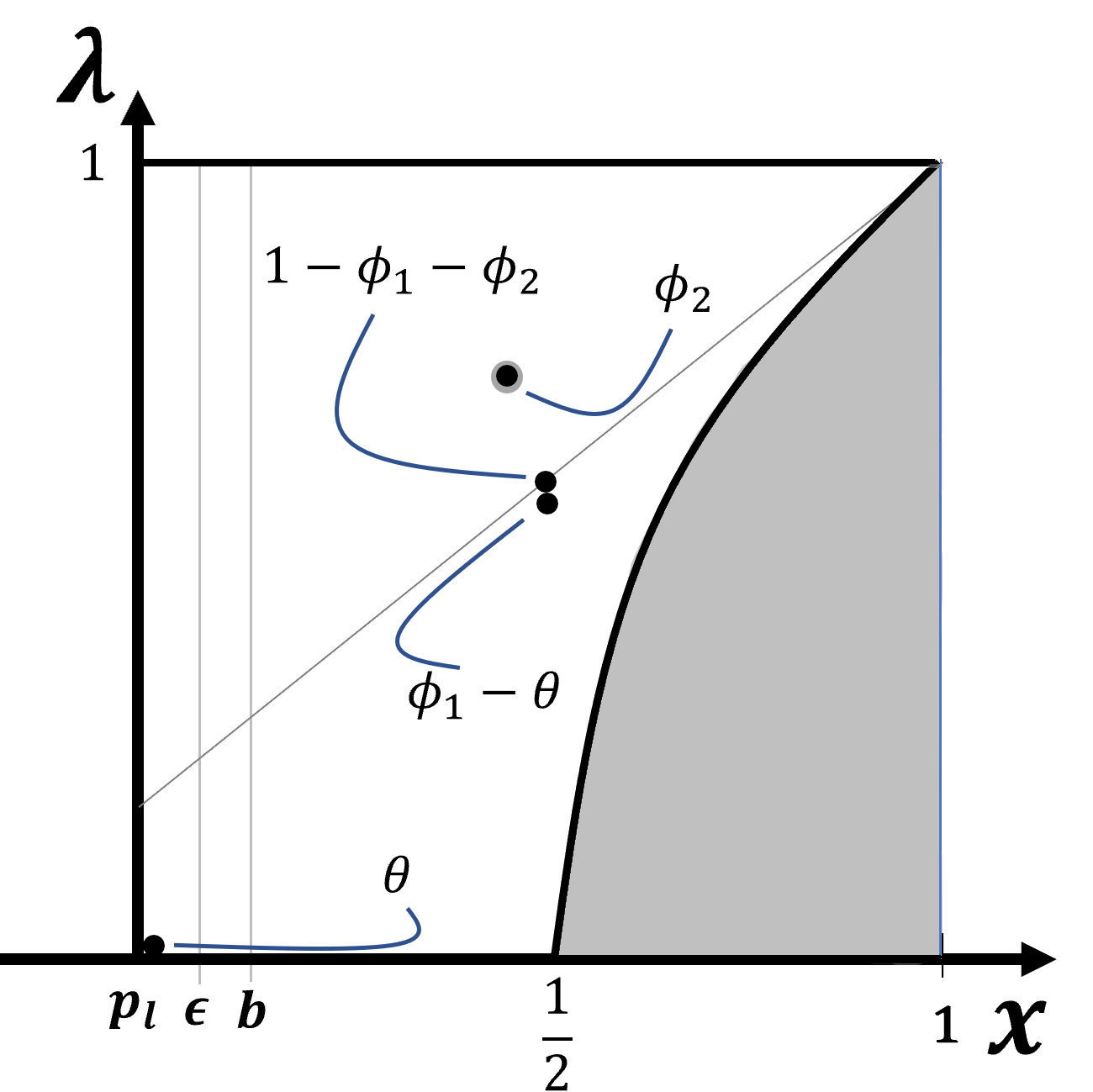}
		\caption{{\footnotesize $\phi_1\geqslant \theta\ $ and $\ r>0$ }}
		\label{fig:fig_CBIsoln_withFails_zeroconf_1}
	\end{subfigure}
	\begin{subfigure}[]{0.3\linewidth}
		\centering
			\includegraphics[width=1.0\linewidth]{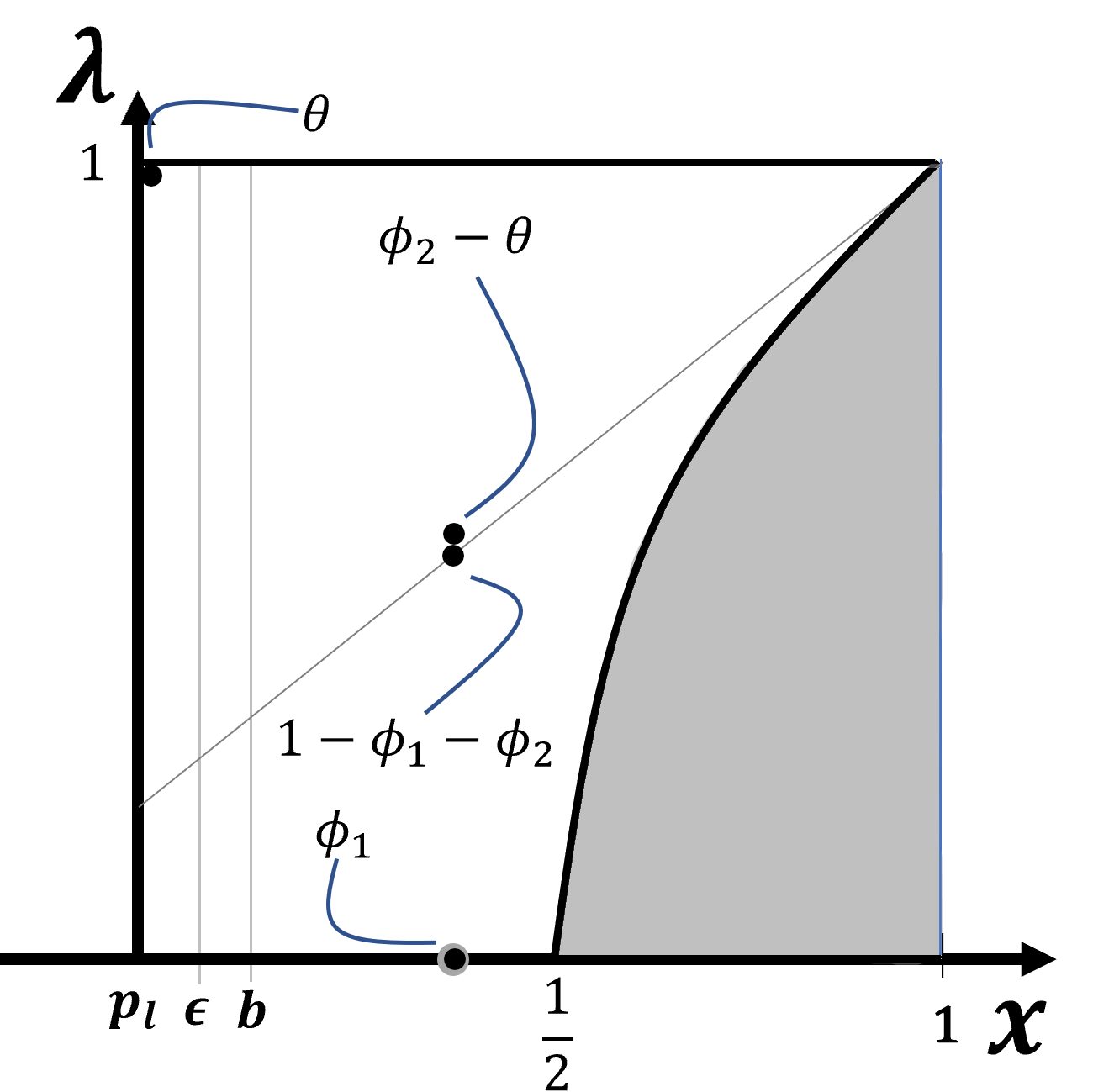}
		\caption{{\footnotesize $\phi_2\geqslant\theta\ $ and $\ r=0$}}
		\label{fig:fig_CBIsoln_withnoconsFails_zeroconf_1}
	\end{subfigure}
	\begin{subfigure}[]{0.3\linewidth}
		\centering
			\includegraphics[width=1.0\linewidth]{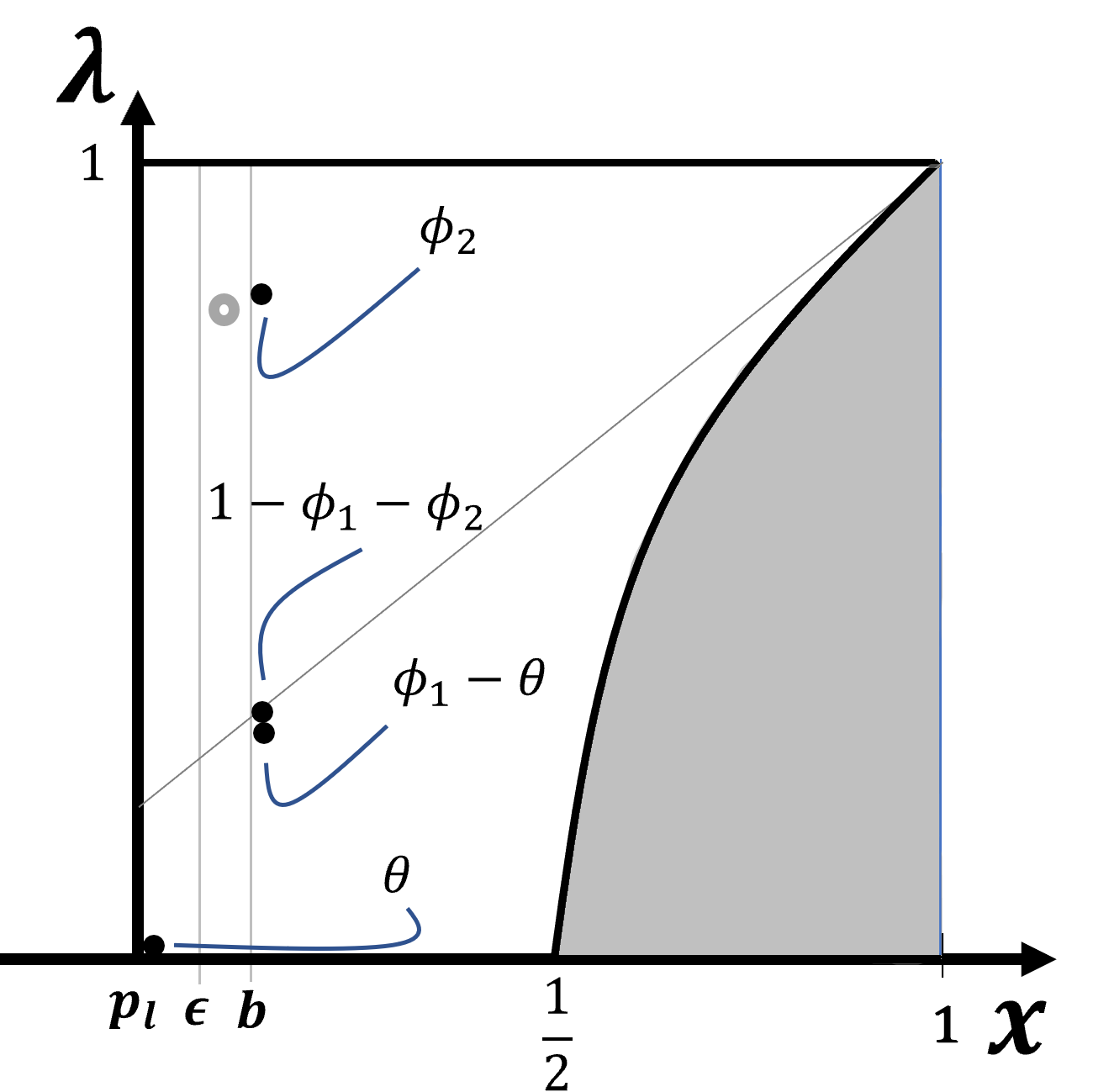}
		\caption{{\footnotesize $\phi_1\geqslant \theta\ $ and $\ r>0$  }}
		\label{fig:fig_CBIsoln_withFails_zeroconf_2}
	\end{subfigure}
	\begin{subfigure}[]{0.3\linewidth}
		\centering
			\includegraphics[width=1.0\linewidth]{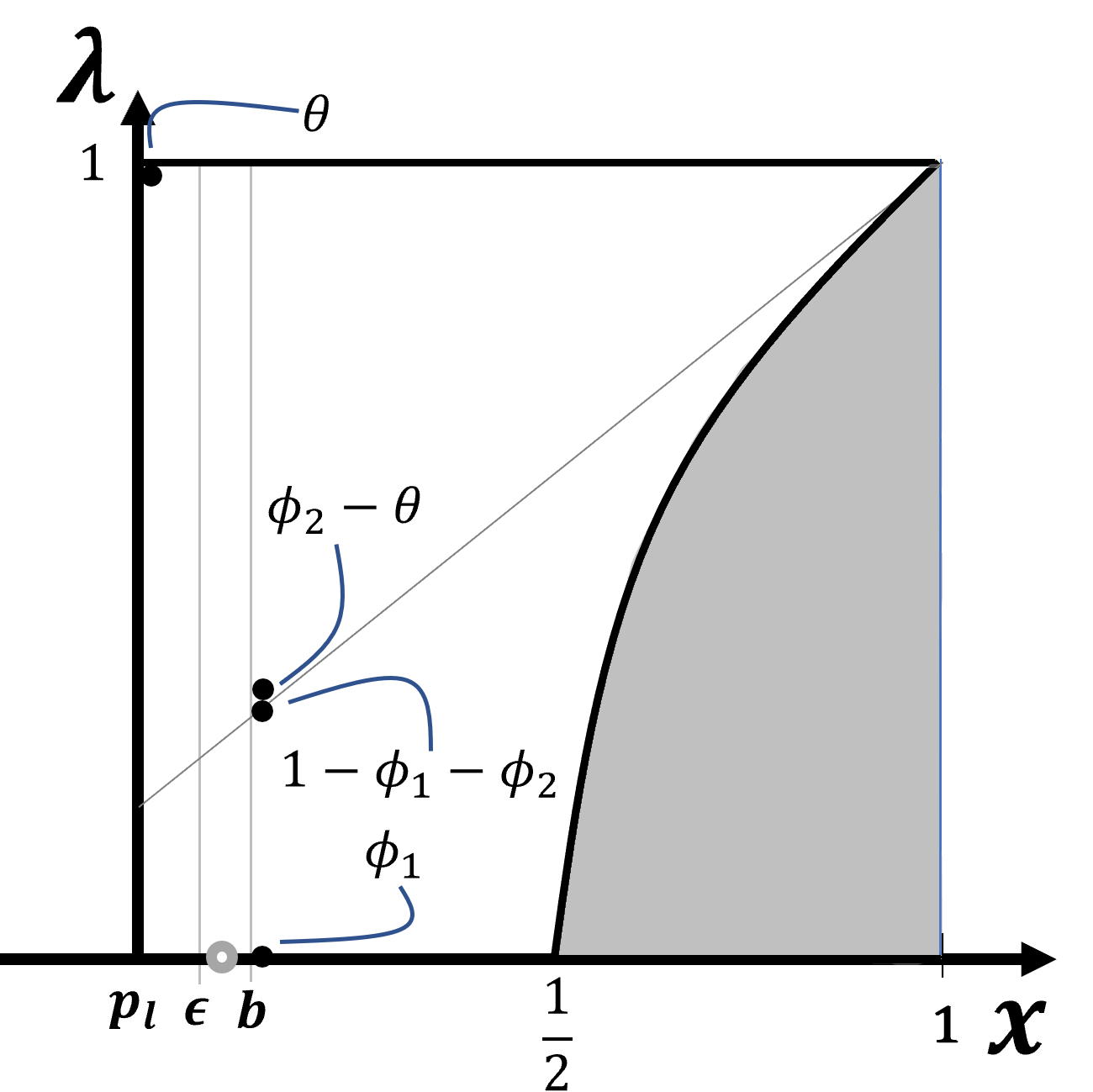}
		\caption{{\footnotesize $\phi_2\geqslant\theta\ $ and $\ r=0$ }}
		\label{fig:fig_CBIsoln_withnoconsFails_zeroconf_2}
	\end{subfigure}
	\begin{subfigure}[]{0.3\linewidth}
		\centering
			\includegraphics[width=1.0\linewidth]{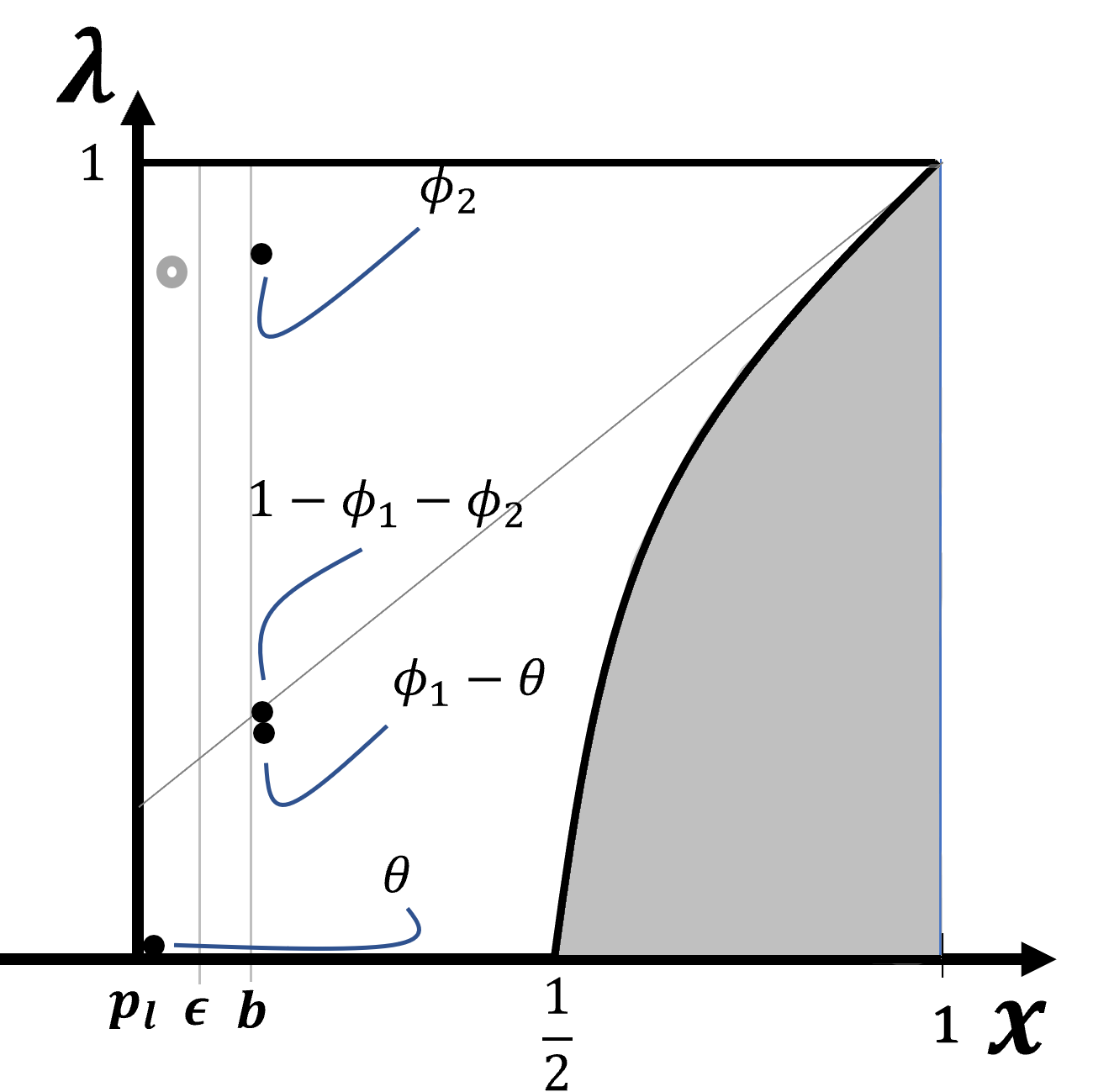}
		\caption{{\footnotesize $\ L(p_l, p_l ;\alpha,\ldots)$ $\leqslant$ $L(\epsilon,\epsilon;\alpha,\ldots)$ and $\phi_1\geqslant \theta\ $ and $\ r>0$ }}
		\label{fig:fig_CBIsoln_withFails_zeroconf_3}
	\end{subfigure}
	\begin{subfigure}[]{0.3\linewidth}
		\centering
			\includegraphics[width=1.0\linewidth]{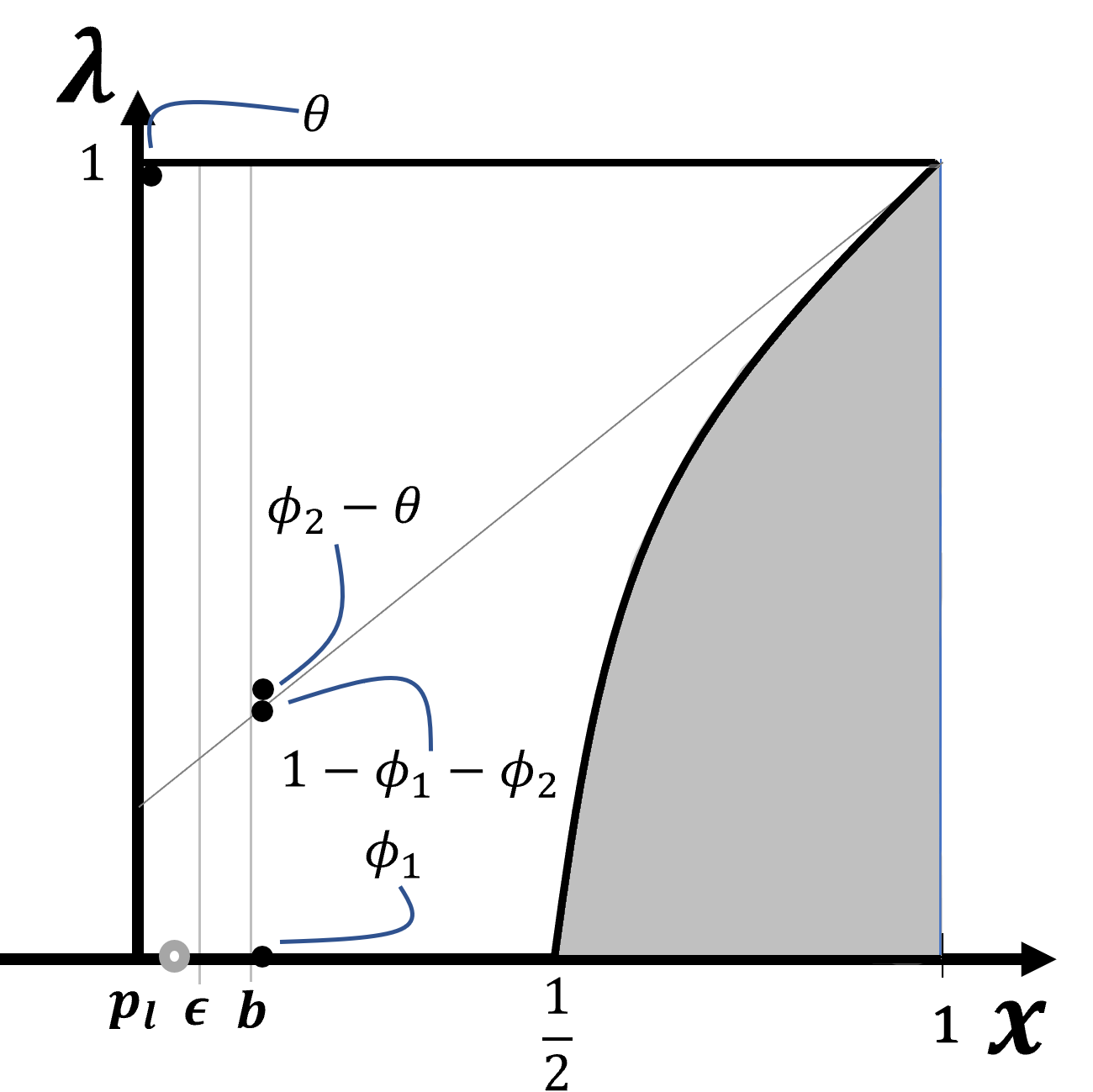}
		\caption{{\footnotesize $\ L(p_l, p_l ;\alpha,\ldots)$ $\leqslant$ $L(\epsilon,\epsilon;\alpha,\ldots)\ $ and $\phi_2\geqslant\theta\ $ and $\ r=0$  }}
		\label{fig:fig_CBIsoln_withnoconsFails_zeroconf_3}
	\end{subfigure}
	\begin{subfigure}[]{0.3\linewidth}
		\centering
			\includegraphics[width=1.0\linewidth]{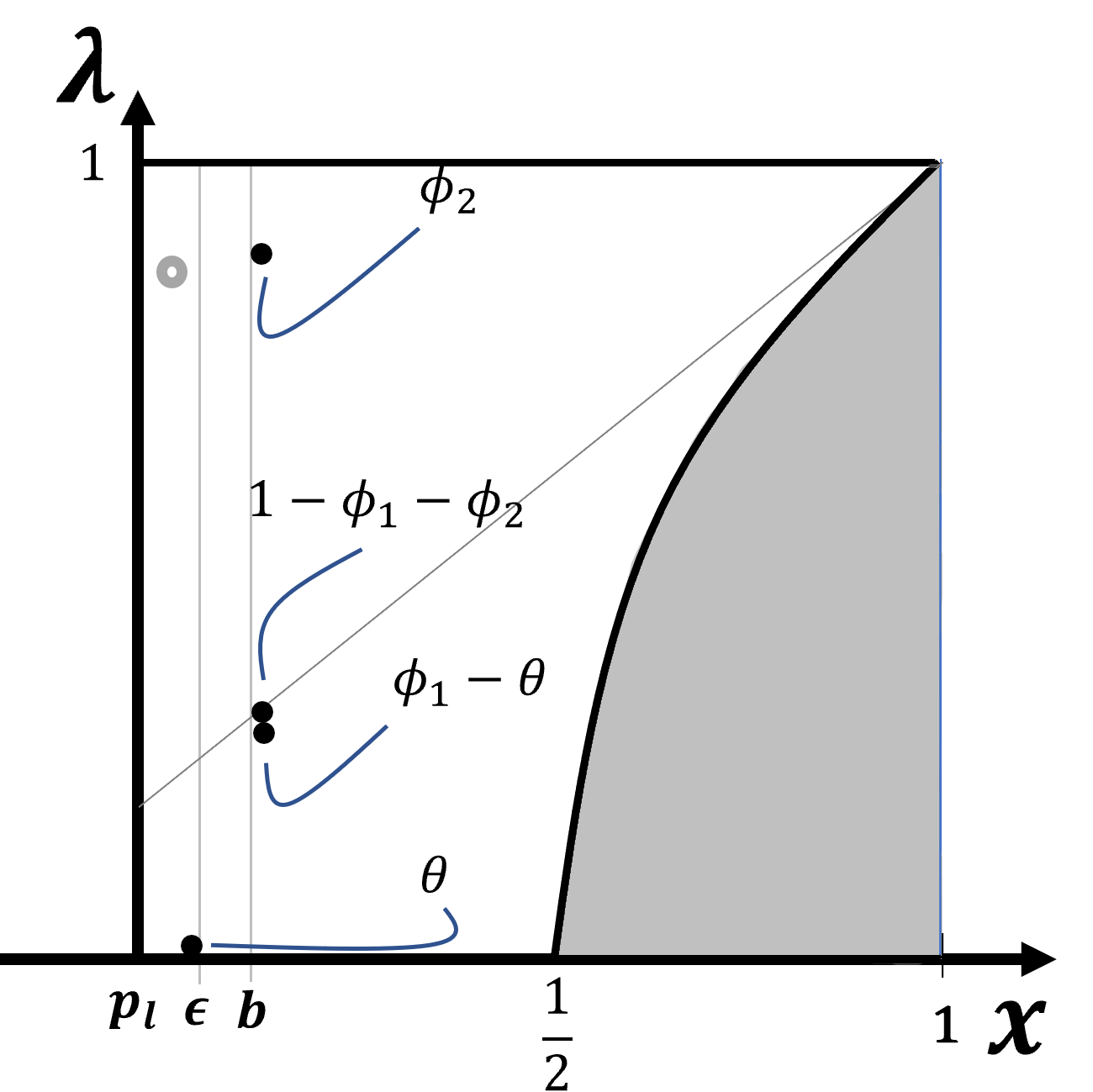}
		\caption{{\footnotesize $\ L(p_l, p_l ;\alpha,\ldots)$ $\geqslant$ $L(\epsilon,\epsilon;\alpha,\ldots)\ $ and $\phi_1\geqslant \theta\ $ and $\ r>0$ }}
		\label{fig:fig_CBIsoln_withFails_zeroconf_4}
	\end{subfigure}
	\begin{subfigure}[]{0.3\linewidth}
		\centering
			\includegraphics[width=1.0\linewidth]{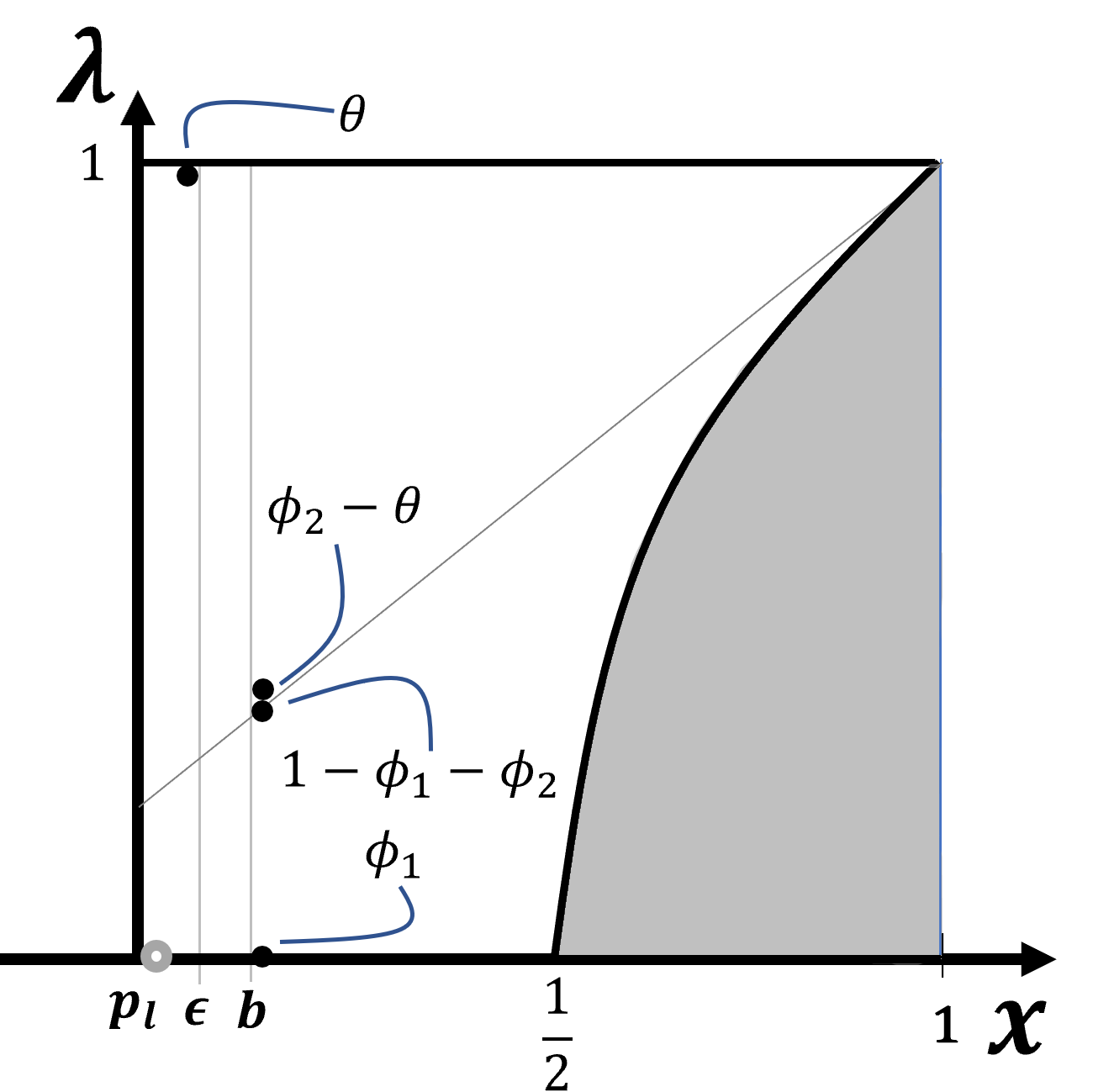}
		\caption{{\footnotesize $\ L(p_l, p_l ;\alpha,\ldots)$ $\geqslant$ $L(\epsilon,\epsilon;\alpha,\ldots)\ $ and $\phi_2\geqslant\theta\ $ and $\ r=0$ }}
		\label{fig:fig_CBIsoln_withnoconsFails_zeroconf_4}
	\end{subfigure}
	\caption[CBI priors: with failures]{{\footnotesize Worst case prior distributions that solve the optimisation problem in Theorem \ref{theorem_CBI_withFailures_and_pl} when failures are observed, giving $0$ posterior confidence. Each distribution's support is determined by $\alpha, \beta, \gamma, \delta$, and whether the first execution succeeds or fails. The location $(x^\ast, \lambda^\ast)$ of the global maximum for the Klotz likelihood is indicated by the grey circle. The $4$ priors, illustrated in subfigures \ref{fig:fig_CBIsoln_withFails_zeroconf_1},  \ref{fig:fig_CBIsoln_withFails_zeroconf_2}, \ref{fig:fig_CBIsoln_withFails_zeroconf_3} and \ref{fig:fig_CBIsoln_withFails_zeroconf_4}, are solutions when $\phi_1\geqslant \theta\ $ and $\ r>0$. While the priors in \ref{fig:fig_CBIsoln_withnoconsFails_zeroconf_1},  \ref{fig:fig_CBIsoln_withnoconsFails_zeroconf_2}, \ref{fig:fig_CBIsoln_withnoconsFails_zeroconf_3} and \ref{fig:fig_CBIsoln_withnoconsFails_zeroconf_4} solve the problem when $\phi_2\geqslant\theta\ $ and $\ r>0$. These solutions assume $\alpha, \beta, \gamma, \delta > 0$.  \normalsize}}  
	\label{fig_CBIsoln_Fails_zeroconf}
\end{figure}
\begin{figure}[h!]
	\captionsetup[figure]{format=hang}
	\centering
	\begin{subfigure}[]{1.0\linewidth}
		\centering
			\includegraphics[width=0.3\linewidth]{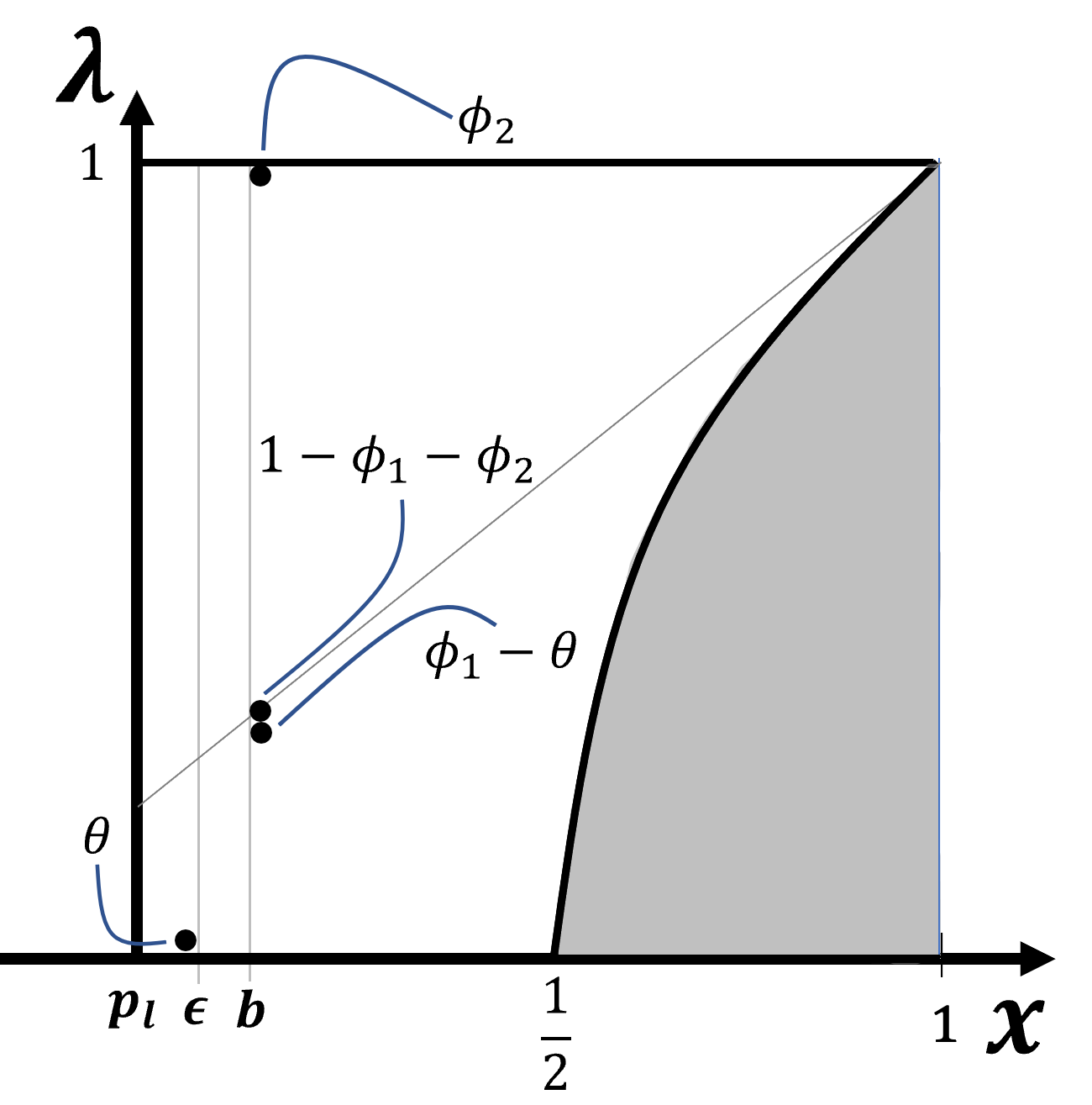}
		\caption{{\footnotesize $\phi_1\geqslant\theta$ }}
		\label{fig:fig_CBIsoln_noFails_Phi1gtTheta}
	\end{subfigure}
	\begin{subfigure}[]{0.3\linewidth}
		\centering
			\includegraphics[width=1.0\linewidth]{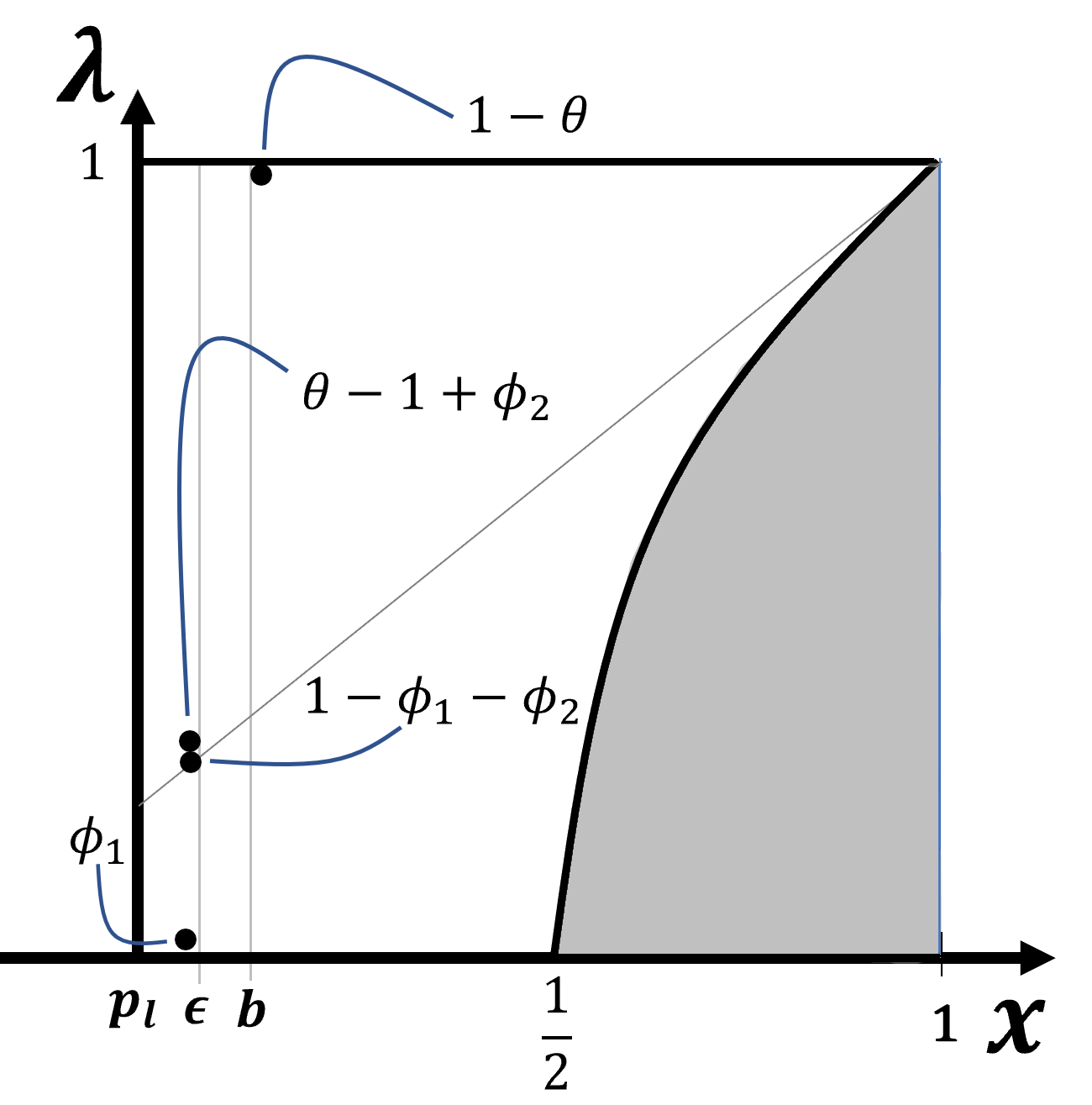}
		\caption{{\footnotesize $\phi_2\geqslant 1-\theta$  }}
		\label{fig:fig_CBIsoln_noFails_Phi1ltThetaPhi2gt1minusTheta}
	\end{subfigure}
	\begin{subfigure}[]{0.3\linewidth}
		\centering
			\includegraphics[width=1.0\linewidth]{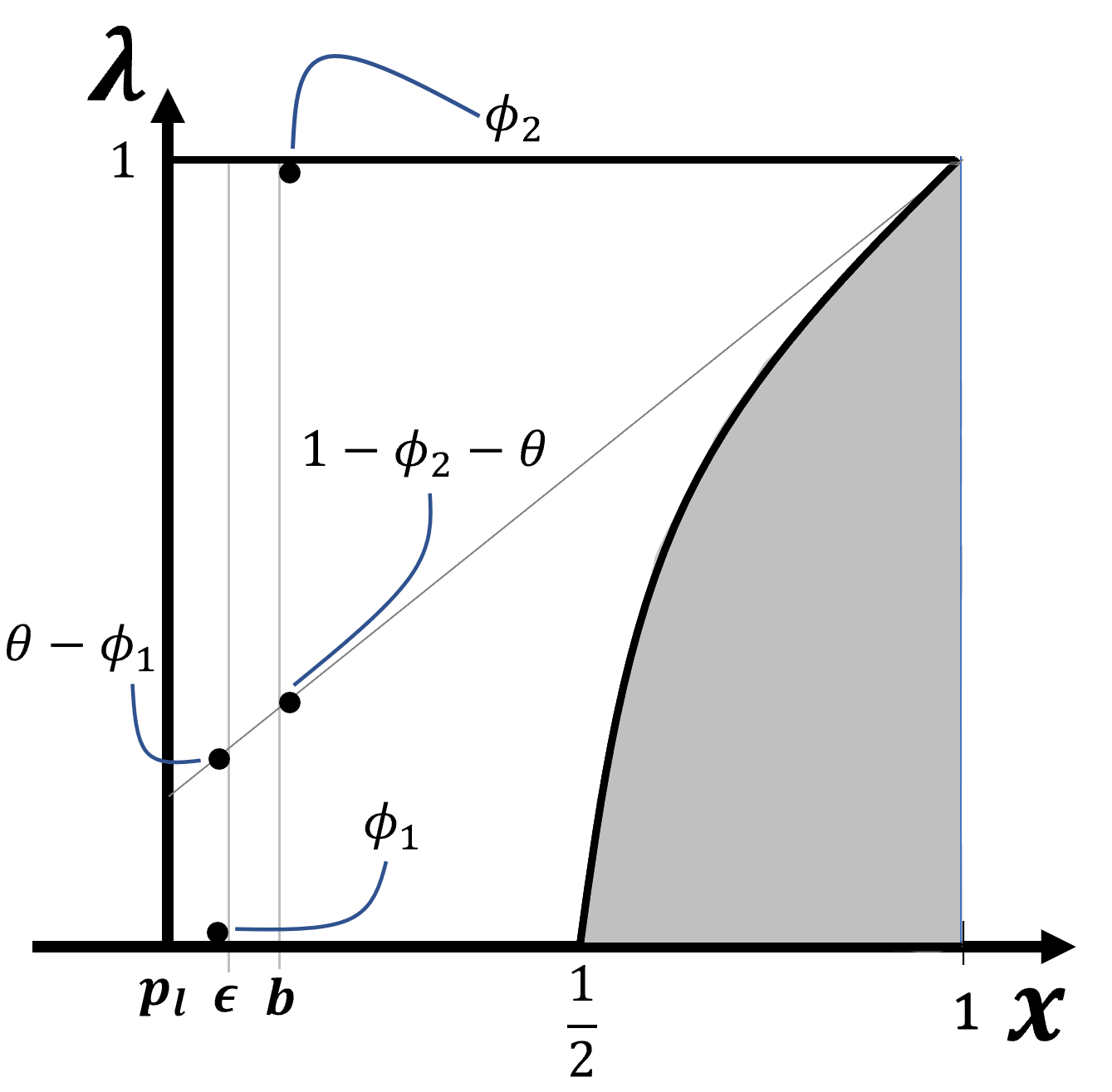}
		\caption{{\footnotesize $\phi_1\leqslant\theta\ $ and $\ \phi_2\leqslant 1-\theta$ }}
		\label{fig:fig_CBIsoln_noFails_Phi1ltThetaPhi2lt1minusTheta}
	\end{subfigure}
	\caption[CBI priors: no failures]{{\small For executions with {\bf no failures}, these worst-case prior distributions solve the optimisation problem in Theorem \ref{theorem_CBIwithdependence}. These priors are relevant for the ranges of parameter values indicated in each subfigure. The support of each distribution is determined by $\beta$. \normalsize}}
	\label{fig_CBIsoln_noFails}
\end{figure}

For example, consider the limit points in Fig.\ref{fig:fig_limitingDist_gammadeltaalpha0}, for the case when $\phi_1\geqslant\theta$ and no failures are observed. Focus on the subset $x\leqslant \epsilon$ and recall the requirement $P(X\leqslant\epsilon)=\theta$. To be pessimistic, we must allocate probability $\theta$ to those limit points within this subset at which the likelihood is smallest. This is the limit point $(\epsilon, 0)$. Since $P(X\leqslant\Lambda)=\phi_1\geqslant\theta$, all of the $\theta$ probability can be allocated to this limit point ``from below'' the $45^\circ$ diagonal and ``from the left'' of the line $x=\epsilon$ (see Fig.~\ref{fig:figB8_proof_step1}). Consequently, because $P(X\leqslant\epsilon)=\theta$, no more probability can be allocated to any other limit points in $x\leqslant\epsilon$.

Now we need to assign probability $1-\theta$ to the remaining limit points in $\mathcal R$. There are two alternative limit points above the diagonal where we may assign the $\phi_2$ probability. Assigning to the point $(b,1)$ gives more pessimistic results than assigning to $(b,b)$. We can see this by sharing the $\phi_2$ probability between the two points, and noting that the objective function monotonically decreases as the amount of $\phi_2$ allocated to $(b,1)$ increases. In effect, all of $\phi_2$ should be allocated to any sequence of points that approximate $(b,1)$ arbitrarily-well, ``from the right'' of the line $x=b$. This justifies Fig.~\ref{fig:figB8_proof_step2}. Similar reasoning shows that allocating probability $\phi_1-\theta$ to the point $(b,b)$, ``from the right'' of $x=b$, is more pessimistic than allocating it to $(b,0)$ ``from the left''. Thus justifying Fig.~\ref{fig:figB8_proof_step3}.  Note that these allocations are possible and do not violate the constraints, because $1-\phi_1-\phi_2\geqslant 0$ and $\phi_1\geqslant\theta$ imply $0\leqslant\phi_1-\theta + \phi_2\leqslant 1-\theta$.

Finally, using similar ``approximation''-based reasoning to how $\phi_2$ and $\phi_1-\theta$ were allocated, we must assign the remaining probability $1-\phi_1-\phi_2$ to the point $(b,b)$. Via any sequence of points that approximate $(b,b)$ arbitrarily-well ``from the right'', along the diagonal (see Fig.~\ref{fig:figB8_proof_step4}). 

Note that probabilities were assigned to limit points that lie along the line $x=b$, but only by assigning the probabilities to points that approximate these limit points ``from the right'' of the line $x=b$. Consequently, our final limiting distribution gives the value of the infimum in the optimisation problem, but only by computing ``$P( X < b\mid\ldots)$'' for this distribution, and not by computing ``$P( X \leqslant b\mid\ldots)$''.

Using similar arguments to allocate probabilities to limit points, all of the remaining worst-case distributions in Fig.s~\ref{fig_CBIsoln_withFails_Phi2glt1minustheta} --  \ref{fig_CBIsoln_noFails} are constructed from limit points analogous to those in Fig.~\ref{fig_limitingDists}.
\end{proof}

\noindent\textbf{Remarks} : with very few modifications, the foregoing arguments can be used to derive worst-case prior distributions subject to the additional constraint of PK\ref{cons_pl_lowerbound}, i.e. $P(X\geqslant p_l)=1$. Such priors solve the optimisation problem in Theorem \ref{theorem_CBI_withFailures_and_pl}. Indeed, after observing $n$ executions of a system (which include some consecutive, failed executions), figure \ref{fig_CBIsoln_withFails_Phi2glt1minustheta} illustrates $2$ groups of priors ($4$ priors in each group) that give the smallest posterior confidence in the system's unknown \emph{pfe} being no worse than the bound $b$. These solutions are subject to PKs \ref{cons_pl_lowerbound}, \ref{cons_engineering_goal}, \ref{cons_negative_dependence} and \ref{cons_positive_dependence}.

\section{Independent Executions and Conservative Assessments}
\label{sec_post_conf_in_independence}

\subsection{Which Independence Beliefs give Conservative Results?}
Fig.~\ref{fig_example2_nuclear} already shows that assessments based on independent executions \emph{are} conservative initially, when the number of inputs during operational testing is relatively small. That is, the univariate CBI curve initially overlaps with the CBI curve for dependent executions.  But, posterior confidence based on independence becomes increasingly optimistic after a large number of tests. That is, the curves begin to deviate significantly as the number of tests increases. So, only assessments that allow for doubt in independence -- i.e. nonzero $\phi_1$ or $\phi_2$ -- can support (in the long run) more pessimistic claims about the bound $b$.

Once some doubt in independence has been expressed, an assessor might want to allow for operational evidence to ``slowly'' allay such doubts. Or, instead, allow for the evidence to ``quickly'' convince them otherwise -- that independence does \emph{not} hold! PK\ref{cons_strongbeliefinIndependence} represents an assessor who's initially very confident the system will exhibit independent, failure-free executions.

\begin{constraint}[strong belief in independence]
\label{cons_strongbeliefinIndependence}
The probability $P(\mbox{ executions will be independent and failure-free })$, from the joint prior distribution of $(X,\Lambda)$, has a value that is the solution to the optimisation problem: 
\begin{align*}
	&\qquad\sup\limits_{\mathcal D} P(\mbox{ executions will be independent and failure-free }) \\
	&\mbox{s.t.} \,\,\,\,\,\,PK\ref{cons_pl_lowerbound},\,\,\,PK\ref{cons_engineering_goal},\,\,\,PK\ref{cons_negative_dependence},\,\,\,PK\ref{cons_positive_dependence}
\end{align*}
\end{constraint}
While PK\ref{cons_weakbeliefinIndependence} is held by an assessor who's initially very doubtful the system will exhibit independent, failure-free executions.

\begin{constraint}[weak belief in independence]
\label{cons_weakbeliefinIndependence}
The probability $P(\mbox{ executions will be independent and failure-free })$, from the joint prior distribution of $(X,\Lambda)$, has a value that is the solution to the optimisation problem:
\begin{align*}
	&\qquad\inf\limits_{\mathcal D} P(\mbox{ executions will be independent and failure-free }) \\
	&\mbox{s.t.} \,\,\,\,\,\,PK\ref{cons_pl_lowerbound},\,\,\,PK\ref{cons_engineering_goal},\,\,\,PK\ref{cons_negative_dependence},\,\,\,PK\ref{cons_positive_dependence}
\end{align*}
\end{constraint}

Note the difference between the optimisation problems in these PKs, and those in CBI Theorems \ref{theorem_univariateCBI}, \ref{theorem_thm1_baseline} and \ref{theorem_CBI_withFailures_and_pl}. Here, the optimisations \emph{constrain} the prior distribution in how it assigns probability mass; hence why these are PKs. While the previous optimisations \emph{are contrained by} the prior distributions -- specifically, constrained by the PKs the priors must satisfy.  

\begin{figure}[h!]
	\captionsetup[figure]{format=hang}
	\centering
	\begin{subfigure}[]{0.3\linewidth}
		\centering
		\includegraphics[width=1.0\linewidth]{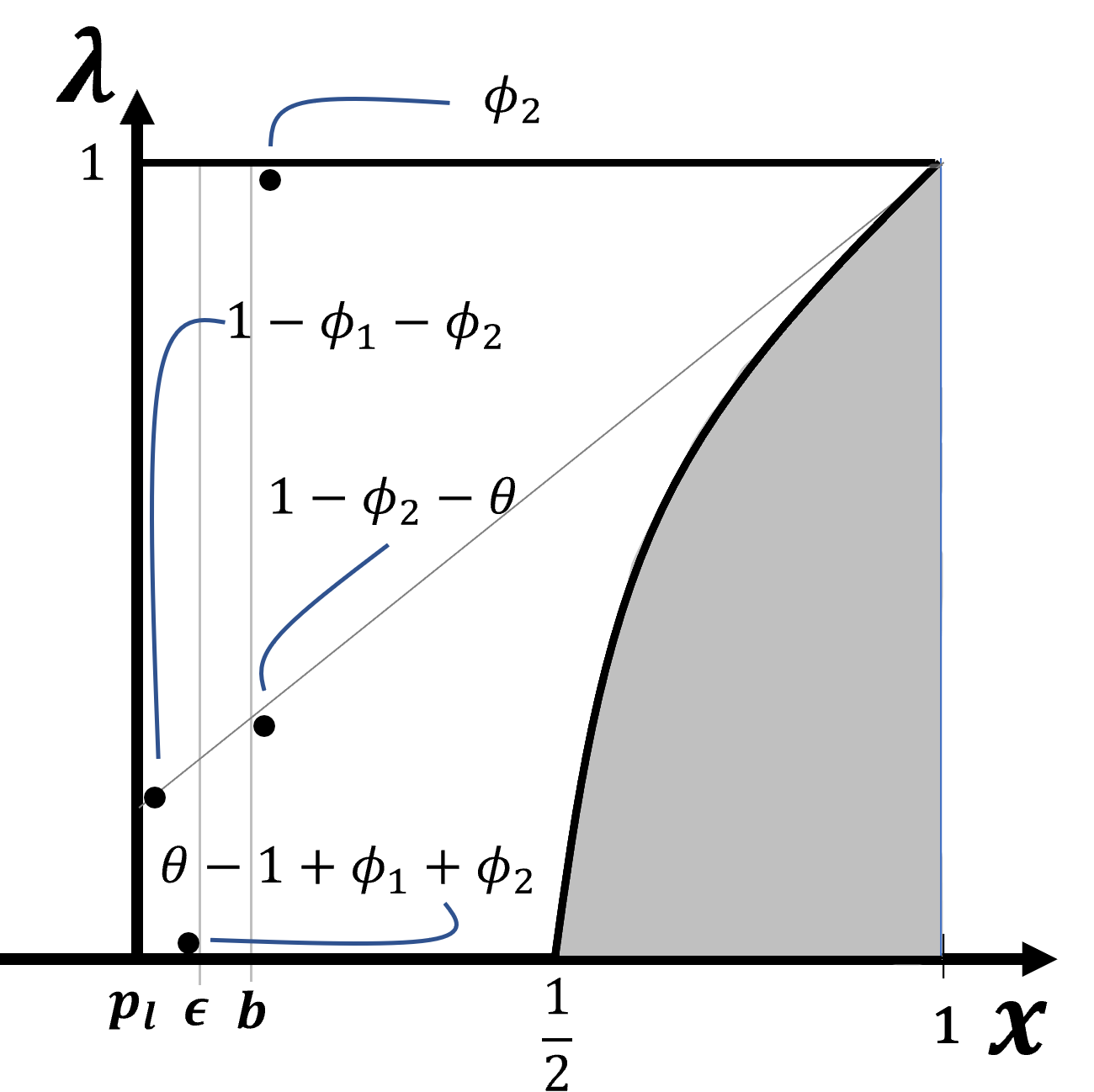}
		\caption{{\footnotesize $\phi_1+\phi_2\geqslant 1-\theta\,$ and $\,\phi_2\leqslant 1-\theta$ \normalsize}}
		\label{fig:fig_strongbeliefInIndependence}
	\end{subfigure}
	\begin{subfigure}[]{0.3\linewidth}
		\centering
			\includegraphics[width=1.0\linewidth]{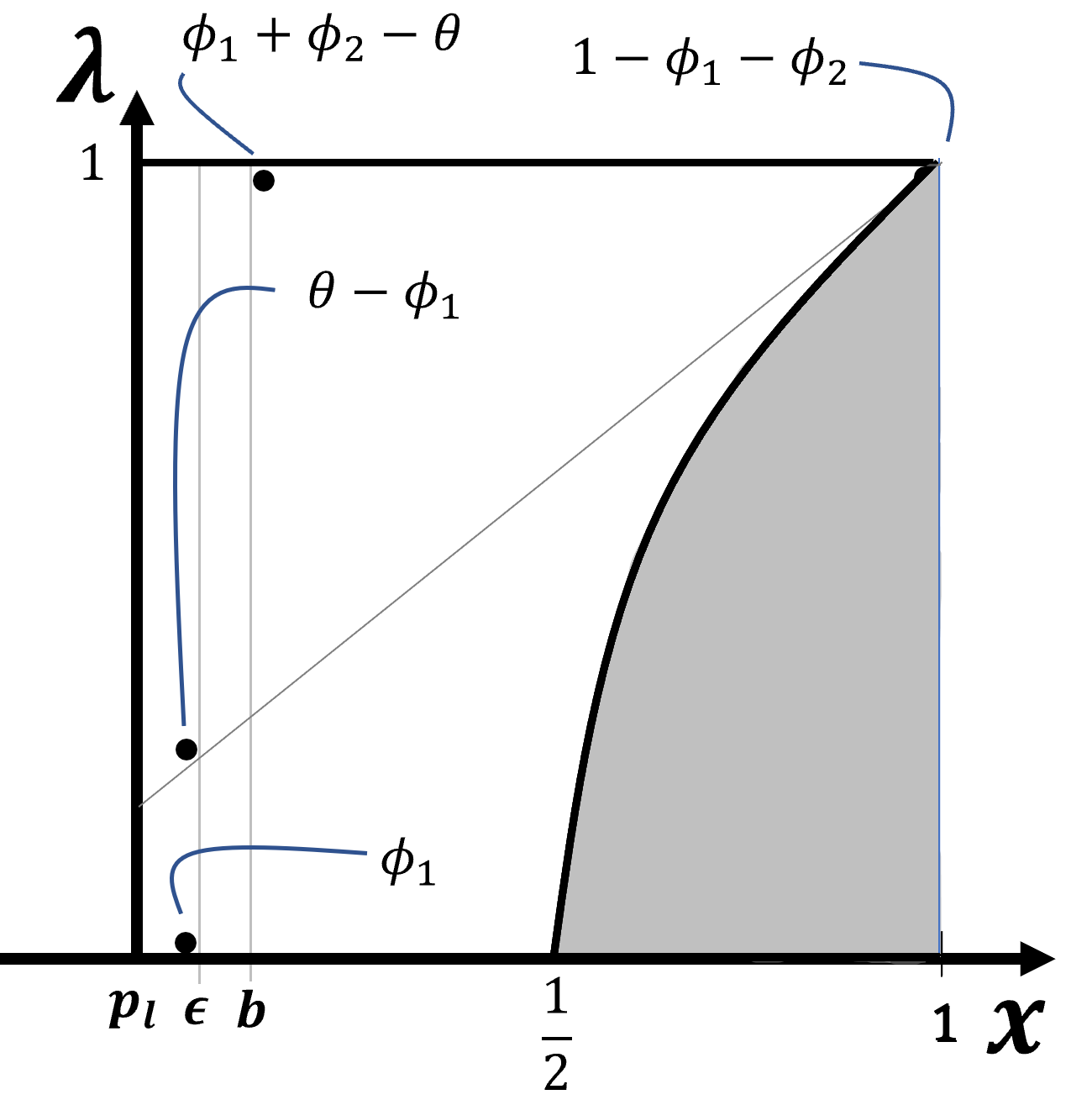}
		\caption{{\footnotesize $\phi_1+\phi_2\geqslant \theta\,$ and $\,\phi_1\leqslant \theta$ \normalsize }}
		\label{fig:fig_strongskepticismAboutIndependence}
	\end{subfigure}
	\caption[Extreme beliefs about independent executions]{{\footnotesize  Prior distributions representing extreme beliefs about whether the executions will be independent (i.e. whether $[X=\Lambda]$), for the $\phi_1,\,\phi_2,\,\theta$ ranges indicated. Consider all prior distributions with the largest prior probability of observing $n$ independent executions with no failures. The most pessimistic posterior confidence in $b$ from these priors is given by the prior in Fig.~\ref{fig:fig_strongbeliefInIndependence}. Similarly, for all priors with the smallest prior probability of $n$ independent failure-free executions, Fig.~\ref{fig:fig_strongskepticismAboutIndependence} gives the most pessimistic posterior confidence.   
			\normalsize}}
	\label{fig_beliefsAboutIndependence}
\end{figure}

Theorems \ref{theorem_beliefinIndependenceIsConserv} and \ref{theorem_doubtinIndependenceIsOptimistic} below give the pessimistic posterior confidence in $b$, for assessors who hold PK\ref{cons_strongbeliefinIndependence} or PK\ref{cons_weakbeliefinIndependence} beliefs, respectively. Proved in \ref{sec_app_D}, some prior distributions that give the pessimistic posterior confidence in these theorems are shown in Fig.~\ref{fig_beliefsAboutIndependence}. And the confidence from these priors, as well as from priors in Theorems \ref{theorem_univariateCBI} and \ref{theorem_thm1_baseline}, are compared in Fig.~\ref{fig_doubtInIndependence_noperfection}.

\begin{theorem}
Using \eqref{eqn_xKlotzlklhdFn_maintxt} and \eqref{eqn_Klotzmodelpostconf}, the optimisation problem
\begin{align*}
	&\qquad\inf\limits_{\mathcal D} P(\,X\leqslant b \mid n\mbox{ executions without failure} ) \\
	&\mbox{s.t.} \,\,\,\,\,\,PK\ref{cons_pl_lowerbound},\,\,\,PK\ref{cons_engineering_goal},\,\,\,PK\ref{cons_negative_dependence},\,\,\,PK\ref{cons_positive_dependence},\,\,\,PK\ref{cons_strongbeliefinIndependence} 
\end{align*}
has the prior distribution in Fig.~\ref{fig:fig_strongbeliefInIndependence} as a solution, since $P(\, X<b \mid n\mbox{ executions without failure} )$ from this prior equals the infimum.
\label{theorem_beliefinIndependenceIsConserv}
\end{theorem}

\begin{theorem}
Using \eqref{eqn_xKlotzlklhdFn_maintxt} and \eqref{eqn_Klotzmodelpostconf}, the optimisation problem
\begin{align*}
	&\qquad\inf\limits_{\mathcal D} P(\,X\leqslant b \mid n\mbox{ executions without failure} ) \\
	&\mbox{s.t.} \,\,\,\,\,\,PK\ref{cons_pl_lowerbound},\,\,\,PK\ref{cons_engineering_goal},\,\,\,PK\ref{cons_negative_dependence},\,\,\,PK\ref{cons_positive_dependence},\,\,\,PK\ref{cons_weakbeliefinIndependence} 
\end{align*}
has the prior distribution in Fig.~\ref{fig:fig_strongskepticismAboutIndependence} as a solution, since $P(\, X<b \mid n\mbox{ executions without failure} )$ from this prior equals the infimum.
\label{theorem_doubtinIndependenceIsOptimistic}
\end{theorem}

Fig.~\ref{fig_doubtInIndependence_noperfection} clearly shows that the confidence from the very skeptical ``PK\ref{cons_weakbeliefinIndependence}''-believing assessor is initially the most optimistic (i.e. the widely-spaced dotted curve lies above all of the other curves). Noticeably more optimistic than even the confidence based on independent executions (i.e. the solid curve). This is in contrast to the assessor who holds the strong PK\ref{cons_strongbeliefinIndependence} belief in independence. Such beliefs actually support conservative claims initially, even as claims based on independence start becoming optimistic -- as suggested by the overlap of the dashed curve and the narrowly-spaced dotted curve, where the solid curve lies above both of them. Eventually, however, the roles are reversed as PK\ref{cons_strongbeliefinIndependence} supported claims become optimistic, while PK\ref{cons_weakbeliefinIndependence} supported claims become conservative and agree with the dashed curve. In this sense, ``strongly believing'' and ``being skeptical of'' the independence assumption are two halves of ``conservatively doubting'' the independence assumption. This is the behaviour for the range of PK parameter values in Fig.~\ref{fig_doubtInIndependence_noperfection}.

\begin{figure}[htbp!]
\centering
\includegraphics[width=0.45\linewidth]{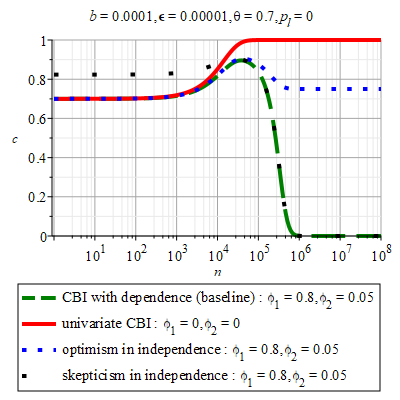}
\caption{{\footnotesize A comparison of posterior confidence in the bound $b$ after operational testing, showing the impact of some confidence in the system's executions being independent. For the parameter values in this example, being very skeptical about independent executions (i.e. the prior in Fig.~\ref{fig:fig_strongskepticismAboutIndependence}) gives the most optimistic posterior confidence in the bound $b$ shown here. While strong beliefs in independent executions (i.e. the prior in Fig.~\ref{fig:fig_strongbeliefInIndependence}) almost results in the most pessimistic confidence in $b$, at least for $n<4\times 10^{4}$ approximately.\kizito{legend needs to be fixed}   \normalsize}}
\label{fig_doubtInIndependence_noperfection}
\end{figure}

Why does PK\ref{cons_strongbeliefinIndependence} initially support less optimistic claims than PK\ref{cons_weakbeliefinIndependence}, then less pessimistic claims as the number of successes $n$ rises? It has to do with which prior beliefs about $(X,\Lambda)$ -- i.e. which locations in $\mathcal R$ -- are crucial for conservative confidence in $b$. There are two principal beliefs a conservative assessor must hold: {\bf i)} strong doubts of failure-free operation being evidence of a ``sufficiently reliable'' system -- i.e., being evidence of a system with a \emph{pfe} just smaller than $b$; {\bf ii)} strong confidence in failure-free operation being evidence of an ``almost sufficiently reliable'' system -- i.e., being evidence of a system with a \emph{pfe} slightly worse than $b$, that exhibits perfectly positively correlated executions. With a PK\ref{cons_strongbeliefinIndependence} belief, failure-free operation initially supports confidence in an ``almost sufficiently reliable'' system (hence, initially conservative confidence in $b$). But eventually, an increasing number of successes could also be due to the system being fault-free (because there is a nonzero probability at $(p_l,p_l)$ in Fig.~\ref{fig:fig_strongbeliefInIndependence} and $p_l=0$ in Fig.~\ref{fig_doubtInIndependence_noperfection}). So, the dotted curve reaches a horizontal asymptote. It's the reverse with a PK\ref{cons_weakbeliefinIndependence} belief, where an ``almost sufficiently reliable'' system is initially very unlikely (hence initially optimistic confidence in $b$), but becomes arbitrarily more likely as $n$ grows (hence conservative confidence in $b$).

When $\epsilon=0$ in PK\ref{cons_engineering_goal},
both PK\ref{cons_strongbeliefinIndependence} and PK\ref{cons_weakbeliefinIndependence} support conservative confidence in $b$ as $n$ grows large -- i.e., Theorems~\ref{theorem_thm1_baseline},  \ref{theorem_beliefinIndependenceIsConserv} (with PK\ref{cons_strongbeliefinIndependence}) and \ref{theorem_doubtinIndependenceIsOptimistic} (with PK\ref{cons_weakbeliefinIndependence}) agree eventually (see Fig.~\ref{fig_doubtInIndependence_withperfection}).

\begin{figure}[h!]
\centering
\includegraphics[width=0.45\linewidth]{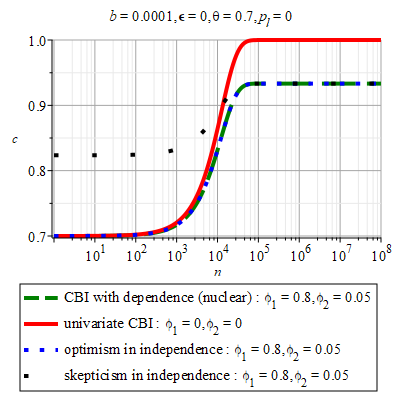}
\caption{{\footnotesize A similar comparison to that of Fig.~\ref{fig_doubtInIndependence_noperfection}, but with the added constraint that there is a probability $\theta=0.7$ of the system containing no faults (i.e., $\epsilon = 0$). Now, a strong belief in independent executions gives the most pessimistic posterior confidence in $b$ -- i.e. the dashed ``\emph{optimism in independence}'' and the dotted ``\emph{CBI with dependence}'' curves are now indistinguishable.\kizito{legend needs to be fixed}\normalsize}}
\label{fig_doubtInIndependence_withperfection}
\end{figure}

\subsection{Proof of Theorem \ref{theorem_beliefinIndependenceIsConserv}}
\label{sec_app_D}
We derive the prior distribution in Fig.~\ref{fig:fig_strongbeliefInIndependence} that solves the optimisation problem in Theorem \ref{theorem_beliefinIndependenceIsConserv}. Analogous steps can be used to derive the prior in Fig.~\ref{fig:fig_strongskepticismAboutIndependence} which solves Theorem \ref{theorem_doubtinIndependenceIsOptimistic}.
\begin{theorem*}
The optimisation problem
\begin{align*}
	&\qquad\inf\limits_{\mathcal D} P(\,X\leqslant b \mid n\mbox{ executions without failure} ) \\
	&\mbox{s.t.} \,\,\,\,\,\,PK\ref{cons_pl_lowerbound},\,\,\,PK\ref{cons_engineering_goal},\,\,\,PK\ref{cons_negative_dependence},\,\,\,PK\ref{cons_positive_dependence},\,\,\,PK\ref{cons_strongbeliefinIndependence} 
\end{align*}
has the prior distribution in Fig.~\ref{fig:fig_strongbeliefInIndependence} as one of its solutions, since $P(\, X<b \mid n\mbox{ executions without failure} )$ from this prior equals the infimum.
\label{theorem_app_beliefinIndependenceIsConserv}
\end{theorem*} 

\begin{proof}
For the prior distributions that solve PK\ref{cons_strongbeliefinIndependence}, the size of \[P(n\mbox{ independent failure-free executions} ) = \MyExp[L(X,1;n,0,0)]\] is made as big as possible (for all $n\geqslant 0$) by assigning as much probability mass as possible to locations along the diagonal in $\mathcal R$ where the Klotz likelihood is largest. The likelihood is largest at $(p_l,p_l)$ and monotonically decreases to $0$ at $(1,1)$. Thus, {\bf i)} if $\theta\geqslant 1-\phi_1-\phi_2$ then the prior assigns all of the mass along the diagonal (i.e. $1-\phi_1-\phi_2$) to the location $(p_l,p_l)$. This gives the largest possible value for $\MyExp[L(X,1;n,0,0)]$ as $(1-\phi_1-\phi_2)(1-p_l)^n$; {\bf ii)} if instead $\theta\leqslant 1-\phi_1-\phi_2$, then the largest amount of probability mass that the prior can assign to $(p_l, p_l)$ is $\theta$, while the remaining $1-\phi_1-\phi_2-\theta$ mass must be assigned to the limit point of the range $x>\epsilon$ (along the diagonal) where the likelihood is largest -- the limit point $(\epsilon,\epsilon)$. In this case, the prior must be the limit of a sequence of priors that re-assign the remaining mass via a sequence of points that converge to $(\epsilon,\epsilon)$ from the right. The largest value of $\MyExp[L(X,1;n,0,0)]$ in this case is thus $\theta(1-p_l)^n+(1-\phi_1-\phi_2)(1-\epsilon)^n$.

Consequently, the priors that solve PK\ref{cons_strongbeliefinIndependence} -- i.e., the feasible priors in theorem~\ref{theorem_beliefinIndependenceIsConserv} -- must allocate probability along the diagonal in one of the two ways just outlined. In particular, from among those priors that allocate all of the $1-\phi_1-\phi_2$ mass to the point $(p_l, p_l)$, the methods of \ref{sec_app_B} justify the prior in Fig.~\ref{fig:fig_strongbeliefInIndependence} as a solution of theorem~\ref{theorem_beliefinIndependenceIsConserv} .
\end{proof}

\end{document}